\documentclass[acmsmall, authorversion=true, nonacm=true]{acmart}
\usepackage{graphicx} % Required for inserting images
\usepackage{amsfonts}
\usepackage{amsthm}
\usepackage{amsmath}
\usepackage{amscd}
\usepackage{mathtools}
\usepackage{mathpartir}
\usepackage{stmaryrd}
\usepackage{tikz-cd}
\usepackage{quiver}
\usepackage{mathrsfs}
\usepackage{ wasysym }
\usepackage{tcolorbox}
\usepackage{lineno}
\usepackage{caption}
\usepackage{subcaption}
\usepackage{pl-syntax}

\usepackage{cleveref}[capitalize]
\usepackage{phaa}

\newcommand{\shorteq}{%
  \settowidth{\@tempdima}{-}% Width of hyphen
  \resizebox{\@tempdima}{\height}{=}%
}
\newcommand{\norm}[1]{\| #1 \|}
\newcommand{\bind}{\mathrel{\scalebox{0.5}[1]{$>\!>=$}}}
\newcommand{\cmon}{\mathbb{C}}
\newcommand{\charge}[1]{\mathsf{charge}\,{#1}}

\newcommand{\Psub}{P_{\leq 1}}
\newcommand{\vdashval}{\vdash_v}
\newcommand{\vdashcomp}{\vdash_c}
\newcommand{\meta}{\mathbf{cert}}
\newcommand\mdoubleplus{\mathbin{+\mkern-5mu+}}
\newcommand*\diff{\mathop{}\!\mathrm{d}}
\newcommand{\weight}{[0,\infty]}
\renewcommand{\twoheadrightarrow}{\rightarrow\mathrel{\mkern-14mu}\rightarrow}
\newcommand{\xtwoheadrightarrow}[2][]{%
  \xrightarrow[#1]{#2}\mathrel{\mkern-14mu}\rightarrow
}

\renewcommand{\V}{\mathcal{V}}
\newcommand{\C}{\mathcal{C}}

\newcommand{\CRel}{\mathrel{\mathcal{C}}}
\newcommand{\VRel}{\mathrel{\mathcal{V}}}
\newcommand{\listty}[1]{\mathsf{list}(#1)}

\newcommand{\caseList}[3]{\mathsf{case}\, #1 \, \mathsf{of}\, \mathsf{nil} \Rightarrow #2 \mid \mathsf{cons}\, x\, xs\Rightarrow #3}

\newcommand{\semcs}[1]{\sem{#1}_{CS}}
\newcommand{\semec}[1]{\sem{#1}_{EC}}
\newcommand{\sempre}[1]{\sem{#1}_{pre}}
\newcommand{\val}{\mathcal{V}al}
\newcommand{\DownarrowFun}{{\Downarrow}}

\newcommand{\wqbs}{\cat{\omega Qbs}}

\DeclarePairedDelimiter\floor{\lfloor}{\rfloor}
\DeclarePairedDelimiter\len{\mid}{\mid}

\title{Denotational Foundations for Expected Cost Analysis}
\author{Pedro H. Azevedo de Amorim}
\email{pedro.azevedo.de.amorim@cs.ox.ac.uk}
\affiliation{%
  \institution{University of Oxford}
  \city{Oxford}
  \country{UK}
}

\begin{document}

\begin{abstract}
Reasoning about the cost of executing programs is one of the fundamental questions
in computer science. In the context of programming with probabilities, however, the 
notion of cost stops being deterministic, since it depends on the probabilistic samples
made throughout the execution of the program. This interaction is further complicated by
the non-trivial interaction between cost, recursion and evaluation strategy.

In this work we introduce $\meta$: a Call-By-Push-Value (CBPV) metalanguage for reasoning about
probabilistic cost. We equip $\meta$ with an operational cost semantics and define two
denotational semantics --- a cost semantics and an expected-cost semantics. We prove
operational soundness and adequacy for the denotational cost semantics and a cost adequacy theorem for
the expected-cost semantics.

We formally relate both denotational semantics by stating and proving a novel \emph{effect simulation}
property for CBPV. We also prove a canonicity property of the expected-cost semantics as the minimal
semantics for expected cost and probability by building on recent advances on monadic probabilistic semantics.

Finally, we illustrate the expressivity of $\meta$ and the expected-cost semantics by presenting case-studies 
ranging from randomized algorithms to stochastic processes and show how our semantics capture their 
intended expected cost.
\end{abstract}

\titlenote{For the purpose of Open Access the author has applied a CC BY public copyright
licence to any Author Accepted Manuscript version arising from this submission.}

\settopmatter{printacmref=false}

\maketitle

% this conflicts with line numbering!!!
%\pagestyle{plain}

\section{Introduction}

This paper addresses the question: what is the semantic essence of expected cost analysis?
Since randomized algorithms \cite{motwani1995} and reasoning about resource consumption in programs are
central pillars of computer science, it is important to lay these analyses on solid theoretical grounds.

Our denotational understanding of deterministic cost analysis has been consolidated by the line of work 
initiated by Danner et al.~\cite{cutler2020, danner2013, danner2015, kavvos2019}. These ideas were later
elegantly refined in the context of modal dependently-typed theories \cite{niu2022}. However, in the case of 
expected cost analysis, while much work has been done in their theory \cite{kaminski2016,avanzini2021,grodin2023}
and practice \cite{wang2020,avanzini2020,leutgeb2022automated}, there is no denotational semantics satisfying the following
desiderata:
\begin{itemize}
    \item Compositional. When reasoning about the expected cost of programs, it should suffice to
    reason about the cost of its components.
    \item Close to probability theory, allowing for useful lemmas to be used when reasoning about programs.
    \item Validates interesting and non-trivial program equations. For instance, in the context of probabilistic
    programming, being able to reorder independent program fragments is natural and useful for justifying program optimizations.
\end{itemize}

In this work we propose the first semantics that validates all of these desiderata. Furthermore, we establish formal
connections and soundness properties between this novel semantics and other denotational approaches to expected cost analysis 
\cite{avanzini2021,grodin2023} by proving a generalization of Filinski's \emph{effect simulation} property \cite{filinski1996controlling}.

\paragraph{Our Work: Denotational Methods for Expected Cost} In this work we shed some light on the semantic foundations of
expected cost analysis in the context of recursive probabilistic functional programs. We start 
by defining $\meta$: a call-by-push-value (CBPV) metalanguage with operations for sampling from uniform distributions and for
incrementing the cost of programs. This metalanguage is expressive enough to represent the cost structure of recursive 
probabilistic algorithms and stochastic processes. It is the first cost-aware metalanguage that can accommodate \textbf{continuous 
distributions}, \textbf{higher-order functions}, \textbf{recursion}, \textbf{compositional expected cost analysis} and 
\textbf{different evaluation strategies}.
% This metalanguage can be seen as the target of translations from more standard languages such as a call-by-value probabilistic
% $\lambda$-calculus but, by explicitly using the $\charge n$ operation, it can also encode other notions of cost, such as the number of
% of calls to particular functions, e.g. comparison functions in the case of sorting algorithms, or the number of samples generated,
% which can be useful when reasoning about stochastic processes.

Besides equipping our metalanguage with an equational theory and an operational cost semantics, we provide two denotational semantics for
reasoning about the cost of probabilistic programs. The first one, which we call the \emph{cost semantics}, uses the familiar
writer monad transformer to combine a cost monad with a subprobability monad. The second semantics, which we call the 
\emph{expected cost semantics}, encapsulates the compositional structure of the expected cost as a monad, allowing us to give a 
denotational semantics that directly tracks the expected cost of programs. This semantics is the first one that validates all of 
the proposed desiderata:
\begin{itemize}
    \item By making the expected cost an explicit component of the expected cost monad defined in \Cref{sec:expectsem}, 
    the semantics enables compositional expected cost reasoning.
    \item It is based on familiar tools from higher-order probability theory \cite{wqbs}. This will be useful
    in \Cref{sec:examples}, where theorems from probability theory are used for reasoning about a stochastic
    convex hull algorithm \cite{golin1988analysis}.
    \item The equational theory validated by the semantics includes useful program equations such as commutativity
    properties that are not validated by other semantics, such as the pre-expectation semantics of Avanzini et al.~\cite{avanzini2020},
    while not validating undesirable equations that are not physically justified (cf. \Cref{sec:eqsound}).
\end{itemize}
In order to justify the mutual validity of these distinct semantics, we show that the expected cost semantics is a sound
approximation to the cost semantics and to the equational theory. In the absence of recursion, we show that this approximation
is an equality while it is an upper bound in the presence of unbounded recursion. In order to achieve this soundness result, we
state and prove a generalization of the effect simulation problem \cite{katsumata2013} beyond base types. This generalization 
interacts well with the CBPV type structure and provides a novel semantic technique for doing relational reasoning of effectful programs.
% Due to the prevalence of probabilistic algorithms and data structures, it is foundationaly important to develop better reasoning principles for them. 
% Because our requirements were just out-of-reach of existing techniques, we had to generalize them in non-trivial ways.

% We define a metalanguage for expected cost. On the syntactic side we define a variant of CBPV extended with cost and probabilistic effects as well as an
% inequational theory that can be used to derive bounds between the expected costs of programs. In order to justify the validity of our formalism we
% show that it is sound with respect to a probabilistic cost-aware denotational semantics. 

Furthermore, we compare our semantics to the influential pre-expectation program transformations introduced by Kaminski et al.~\cite{kaminski2016}. We 
argue that our expected cost semantics provides better abstractions for expected-cost reasoning, since many useful properties that are proved by 
delicate syntactic arguments in the pre-expectation semantics, hold automatically and unconditionally in our semantics. Furthermore, we show that
our expected cost model satisfies useful commutativity equations that are not validated by the pre-expectation one. Finally, we prove an 
unexpected connection between these two monads by showing that the expected cost monad can be obtained as the minimal submonad of the 
pre-expectation monad that can accommodate probability and expected cost, showing that besides providing a compositional account to expected cost, 
the expected cost monad is canonical.

As applications, we showcase the capabilities of our semantics by using it to reason about the expected cost of probabilistic algorithms 
and stochastic processes. We highlight our analysis of a stochastic variant of the convex hull algorithm \cite{golin1988analysis} that
was outside of reach of other logic/PL techniques for expected cost analysis due to its combination of continuous probability, 
modular interaction of cost and behaviour, and deep probabilistic reasoning. As a guiding example, throughout the paper we will use geometric 
distributions to illustrate different aspects of $\meta$ and its semantics.

Our approach contrasts with other work done on expected cost analysis
where the language analyzed was either imperative \cite{kaminski2016, kura2019, avanzini2020} or first-order \cite{sun2023, leutgeb2022automated}, 
or the cost structure was given by non-compositional methods \cite{avanzini2019, chatterjee2016, wang2020, aguirre2022weakest}.
In spirit, the closest to what we have done is work by Kavvos et al.~\cite{kavvos2019}, which does denotational cost analysis in a deterministic recursive setting, work by Avanzini et al \cite{avanzini2021}, which uses a continuation-passing style transformation to reason about the expected cost of a call-by-value language.
% Recurrence relations are a powerful tool used, amongst other things, to analyse recursively defined processes. They have found wide applicability
% in the analysis of algorithms, where resource used by recursive programs, such as time and space, can be expressed as a recurrence relation. That
% being said, these techniques are usually applied informally, either by reasoning about pseudo-code or by reasoning about ode without thinking about
% the formal semantics of the language. In the presence of higher-order functions it is less straightforward to soundly apply these techniques without
% a proper understanding of the semantics. 

% Recent works \cite{kavvos2019} have shown how tricky it can be to properly accommodate a recursive functional calculus with recurrence relations. This
% line of work has culminated in a cost-aware Call-By-Push-Value (CBPV) calculus that can extract recurrence relations for both Call-By-Value (CBV)
% as well as Call-By-Name (CBN) reduction strategies. Though their results are impressive, their framework can only reason about deterministic
% programs, which only constitute a fraction of the use cases of recurrence relations. One notable example is the analysis of probabilistic algorithms,
% where the notion of cost must be generalized to expected cost.

\paragraph{Our contributions} The main contributions of this paper are the following 

\begin{itemize}
    \item We introduce the metalanguage $\meta$, a CBPV variant with primitives for recursion, increasing cost ($\charge c$) and for sampling from uniform distributions ($\uniform$). (\S \ref{sec:probcbpv})
    \item We define a novel expected cost semantics based on an expected cost monad that accommodates familiar and useful reasoning principles. (\S \ref{sec:costsem})
    \item We prove a generalization of the effect simulation problem for cost semantics using a novel logical relations technique for denotational relational 
    reasoning. (\S \ref{sec:costsound})
    \item We prove an adequacy proof of the probabilistic cost semantics with respect to the operational cost semantics. (\S \ref{sec:opsound})
    \item We prove a universal property of the expected cost monad as the minimal submonad of the continuation monad that can accommodate subprobability distributions and cost. (\S \ref{sec:preexpect})
\end{itemize}

We also justify the applicability of our semantics through use-cases that illustrate how the expected cost semantics can be used to reason about expected cost of stochastic processes and randomized algorithms. In particular, the stochastic convex hull example was outside of scope of existing
techniques.

\section{$\meta$ : A Probabilistic Cost-Aware Metalanguage}
\label{sec:probcbpv}
\begin{figure}
    \begin{syntax}
        \category[Value Types]{\tau}
        \alternative{U \overline{\tau}}
        \alternative{1}
        \alternative{\nat}
        \alternative{\R}
        \alternative{\tau \times \tau}
        
        \category[Computation Types]{\overline{\tau}}
        \alternative{F \tau}
        \alternative{\tau \to \overline \tau}

    \separate

        \category[Computation Terms]{t,u}
        \alternative{\lamb x t}
        \alternative{\app t V}
        \alternative{\ifthenelse V t u}
        \alternative{\force\, V}
        \alternativeLine{(x \leftarrow t); u}
        \alternative{\produce V}
        \alternative{\letbe x V t}
        \alternativeLine{\mathsf{succ}\, t}
        \alternative{\mathsf{pred}\, t}\
        \alternative{\letin {(x, y)} V t}

        \category[Value Terms]{V}
        \alternative{x}
        \alternative{()}
        \alternative{n \in \nat}
        \alternative{ r \in \R}
        \alternative{\thunk\, t}
        \alternative{(V_1, V_2)}

    \separate

        \category[Terminal Computations]{T}
        \alternative{\produce \, V}
        \alternative{\lamb x t}
        \alternative{\force \, x}
        \alternative{\ifthenelse x t u}

        \category[Evaluation Contexts]{C}
        \alternative{[\,]}
        \alternative{\lamb x C}
        \alternative{\app C V}
        \alternative{\ifthenelse V C u}
        \alternativeLine{\ifthenelse V t C}
        \alternative{(x \leftarrow C); u}
        \alternative{(x \leftarrow t); C}
        \alternativeLine{\letbe x V C}
        \alternative{\mathsf{succ}\, C}
        \alternative{\mathsf{pred}\, C}
        \alternative{\letin {(x, y)} V C}
    \end{syntax}
\caption{Types and Terms of CBPV}    
\label{fig:syntax}
\end{figure}

\begin{comment}
\begin{figure}
\begin{align*}
 & \overline{\tau} \defined F \tau \smid \tau \to \overline \tau \\
 & \tau \defined U \overline{\tau} \smid 1 \smid \nat \smid \R \smid \tau \times \tau
\end{align*}
%
\begin{align*}
  & t,u \defined \lamb x t \smid \app t V \smid \ifthenelse V t u \smid \force\, V \smid (x \leftarrow t); u\\
  & \smid \produce V \smid \letbe x V t \smid \mathsf{succ}\, t \smid \mathsf{pred}\, t \\
  &\smid \letin {(x, y)} V t \\
  & V \defined x \smid () \smid n \in \nat \smid r \in \R \smid \thunk\, t \smid (V_1, V_2)\\
  & T \defined \produce \, V \smid \lamb x t \smid \force \, x \smid \ifthenelse x t u\\
  & C \defined [\,] \smid \lamb x C \smid \app C V \smid \ifthenelse V C u \smid \ifthenelse V t C\\
  & \smid (x \leftarrow C); u \smid (x \leftarrow t); C \smid \letbe x V C \smid \mathsf{succ}\, C \smid \mathsf{pred}\, C \\
  &\smid \letin {(x, y)} V C \\
\end{align*}
\caption{Types and Terms of CBPV}    
\label{fig:syntax}
\end{figure}
\end{comment}

%
\begin{figure*}
    \begin{mathpar}
      \inferrule{ }{\Gamma_1, x : \tau, \Gamma_2 \vdashval x : \tau}
      \and
      \inferrule{n \in \nat}{\Gamma \vdashval n : \nat}
      \and
      \inferrule{r \in \R}{\Gamma \vdashval r : \R}
\and
      \inferrule{~}{\Gamma \vdashval () : 1}
      \\
      \inferrule{\Gamma \vdashval V : \nat \\ \Gamma \vdashcomp t : \overline{\tau} \\ \Gamma \vdashcomp u : \overline{\tau}}{\Gamma \vdash^c \ifthenelse V t u : \overline{\tau}}
      \and
      \inferrule{\Gamma \vdashval V_1 : \tau_1 \\ \Gamma \vdashval V_2: \tau_2}{\Gamma \vdashval (V_1, V_2) : \tau_1 \times \tau_2}
      \\
      \inferrule{\mathsf{op} \in \mathcal{O}(\tau, \tau')}{\Gamma \vdash \mathsf{op} : \tau \to F \tau'}
      \and
      \inferrule{\Gamma, x : \tau \vdashcomp t : \overline{\tau}}{\Gamma \vdashcomp \lamb x t : \tau \to \overline{\tau}}
      \and
      \inferrule{\Gamma \vdashval V : \tau \\ \Gamma \vdashcomp t : \tau \to \overline{\tau}}{\Gamma \vdashcomp \app t V : \overline{\tau}}
      \\
      \inferrule{\Gamma \vdashval V : \tau }{\Gamma \vdashcomp \produce V : F \tau}
      \and
      \inferrule{\Gamma \vdashcomp t : \overline \tau}{\Gamma \vdashval \thunk t : U \overline \tau}
      \and
      \inferrule{\Gamma \vdashval V : U \overline \tau}{\Gamma \vdashcomp \force V : \overline \tau}
      \and
      \inferrule{\Gamma \vdashcomp t : F \tau' \\ \Gamma, x : \tau' \vdashcomp u : \overline \tau}{\Gamma \vdashcomp (x \leftarrow t); u : \overline \tau}
      \and
      \inferrule{\Gamma \vdashval V : \tau_1 \times \tau_2 \\ \Gamma, x : \tau_1, y : \tau_2 \vdashcomp t : \overline \tau}{\Gamma \vdashcomp \letin {(x,y)} V t : \overline \tau}     
    \end{mathpar}
    \caption{CBPV typing rules}
    \label{fig:cbpvjudgement}
  \end{figure*}    
  
In this section we introduce the type system, equational theory and operational semantics of $\meta$. The
language is a Call-By-Push-Value (CBPV) \cite{levy2001call} calculus extended with operations for cost, probabilistic sampling and recursion.
For the sake of presentation, we introduce it in parts: we begin by going over the core CBPV calculus and then present
its extensions with effectful operations and lists, respectively.
The operational semantics is defined as a weighted Markov chain that accounts for the cost and output distributions
of programs. It is defined using the standard measure-theoretic treatment of operational semantics for continuous distributions,
cf. V\'{a}k\'{a}r et al.~\cite{wqbs}. 

We chose CBPV as the core of $\meta$ due to its type and term-level separation of values and computations. Such a separation allows 
for a fine grained control over the execution of programs, providing a uniform treatment of different evaluation strategies such as Call-By-Name 
(CBN) and Call-By-Value (CBV).

\Cref{fig:syntax} depicts the CBPV syntax. Note that the base types $1$, $\nat$ and $\R$ are value types, the product types is also a value and arrow types are computation types that receive a value type as input and a computation type as output. At the center of the CBPV formalism are the type
constructors $F$ and $U$ which allows types to move between value types and computation types. The constructor $F$ plays a similar role to the monadic type constructor $T$ from the monadic $\lambda$-calculus \cite{moggi1989}, while $U$ is used to represent suspended --- or thunked --- computations. 

This two-level approach also manifests itself at the type judgment level, where the judgment $\Gamma \vdashval V : \tau$ is only defined for value types, while $\Gamma \vdashcomp t : \overline{\tau}$ is defined for computation types, as shown in \Cref{fig:cbpvjudgement}. Both contexts only bind values, which justifies the arrow type having a value type in its domain, so that lambda abstractions only introduce values to the context. The if-then-else operation checks if the guard $V$ is $0$, in which case it returns the first branch, and otherwise it returns the second branch. The language is also parametric on a set of operations $\mathcal{O}$ which contains arithmetic functions. The product introduction rule pairs two values while its elimination rule unpairs a product and uses them in a computation. Lambda abstraction binds a new value to the context while application applies a function to a value. 

The less familiar rules are those for the type constructors $F$ and $U$. The introduction rule for computations is $\produce V$, which is the computation that does not incur any effect and just outputs the value $V$, while the introduction rule for $U$, $\thunk\, t$, suspends the computation $t$. Its elimination rule $\force V$ resumes the
suspended computation $V$. The last rule, $x \leftarrow t; u$ is what makes it possible to chain effectful computations together, since it receives a computation of type $F\tau$ as input, runs it and binds the result to the continuation $u$, which eventually will output a computation of type $\overline \tau$. This is a generalization of the monadic $\mathsf{let}$ rule where the output type does not have to be of type $F\tau$.

The syntax differs a bit from the monadic semantics of effects, but, as it is widely known, every strong monad over a Cartesian closed category can interpret
the CBPV calculus, as we describe in \Cref{app:cbpvsemantics}\footnote{The full appendix can be found on the supplementary material or on the ArXiv version} 
and described in Section~12 of \cite{levy2001call}.

In \Cref{fig:syntax} there is also a grammar for terminal computations and computation contexts $C$. 
These are standard and the latter represents computations with a single
hole $[\, ]$  that may be filled in by a computation $t$ by a non-capture-avoiding substitution, 
which we denote by $C[t]$.

Though this language is effective as a core calculus, by itself it cannot do much, since it has no ``native'' effect operations,
meaning that there are no programs with non-trivial side-effects. In this section we extend CBPV so that it can program with
three different effects: cost, probability and unbounded recursion. We call this extension $\meta$, for \textbf{c}alculus for \textbf{e}xpected \textbf{r}un \textbf{t}ime,
and we conclude the section by presenting its equational theory and operational semantics.

\subsection{Cost and Probabilistic Effects}
 \begin{figure}[]
    \centering
    \begin{minipage}{.50\textwidth}
\begin{align*}
&\fix f : \listty \tau \to F \nat.\, \lambda l : \listty \tau.\,\\
&\mathsf{case} \, l \, \mathsf{of}\\
& \mid \mathsf{nil} \Rightarrow \\
& \quad \produce 0\\
&\mid \mathsf{(hd, tl)} \Rightarrow \\
& \quad n \leftarrow \app {(\force f)} {tl}\\
&\quad \produce (1 + n) 
\end{align*}
    \end{minipage}%
    \begin{minipage}{0.50\textwidth}
    \begin{align*}
&\fix f : \listty \tau \to F (\listty \tau \times \listty \tau) .\,\\
&\lambda l : \listty \tau\,.\\
&\lambda p : \tau \to F \nat.\,\\
&\mathsf{case} \, l \, \mathsf{of}\\
& \mid \mathsf{nil} \Rightarrow \\
& \quad \produce (\mathsf{nil}, \mathsf{nil})\\
&\mid \mathsf{(hd, tl)} \Rightarrow \\
& \quad n \leftarrow \app p hd\\
& \quad (l_1, l_2) \leftarrow \app {\app{(\force f)} p} tl\\
&\quad \mathsf{if}\, n \, \mathsf{then}\\
& \quad \quad \produce (\mathsf{cons}\, hd\, l_1, l_2)\\
&\quad \mathsf{else}\\
& \quad \quad \produce (l_1,\mathsf{cons} \, hd\, l_2)\\
\end{align*}
    \end{minipage}
    \caption{Length function $\mathsf{length}$ (left) and filter function $\mathsf{biFilter}$ (right).}
    \label{fig:programs}
\end{figure}

As it is common in denotational approaches to cost semantics, it is assumed that there is a cost 
monoid $\cmon$ --- usually interpreted by $\nat$ and addition --- which acts on programs by 
operations $\charge c$ that increases the current cost of the computation 
by $c$ units, for every $c : \cmon$. The value types are extended with a type $\cmon$, constants $\cdot \vdashval c : \cmon$
and the monoid structure $\cdot \vdashval 0 : \cmon$ and $\cdot \vdashval V_1 + V_2 : \cmon$.
Furthermore, since we also want to program with probabilities and unbounded recursion, we 
extend the language with a sampling primitive, as well as recursive definitions:
\begin{mathpar}
    \inferrule{\Gamma \vdashval V : \cmon}{\Gamma \vdashcomp \charge V : F 1}
    % \and
    % \inferrule{\Gamma \vdashval V : \nat}{\Gamma \vdashcomp \rand \, V : F \nat}
    \and
    \inferrule{~}{\Gamma \vdashcomp \uniform : F \R}
    \and
    \inferrule{\Gamma, x : U \overline \tau \vdashcomp t : \overline\tau}{\Gamma \vdashcomp \fix x.\, t : \overline \tau}
\end{mathpar}
The operation $\uniform$ uniformly samples a real number from the
interval $[ 0, 1 ]$ and $\fix$is the familiar fixed-point operator used for defining recursive programs.
In interest of reducing visual pollution and simplifying the presentation, 
$\charge V; t$ desugars to $(x \leftarrow \charge V);  t$, when $x$ is not used in the body of $t$, and
we will assume that the cost monoid is $\nat$.

With the uniform distribution primitive it is possible to define uniform distributions over discrete sets which,
given a natural number $n$, outputs a uniform distribution $\rand\,n$ over the set $\{0,\dots, n-1\}$. This can be desugared
to the program $\lamb n {x \leftarrow \uniform; \produce (\floor{nx})} : \nat \to F \nat$, where $\floor{\cdot}$ is the
floor function. Biased choices $t_1 \oplus_p t_2$ desugar to $\lamb p {x \leftarrow \uniform; \ifthenelse {x \leq p}{t_1}{t_2} : \R \to \tau}$, where $t_1, t_2$ have type $\tau$ and $\leq : \R \to \R \to F \nat$ is the comparison function that returns $0$ if the first argument is less or equal to the
second argument and $1$ otherwise.

\begin{example}[Geometric distribution]
    With these primitives we can already program non-trivial distributions. For instance, the geometric distribution
    can be expressed as the program 
    \[\cdot \vdashcomp \fix x.\, (\produce 0) \oplus_{0.5} ((y \leftarrow \force x); \produce (1 + y)) : F \nat,\]
   Operationally, the program flips a fair coin, if the output is $0$, it outputs $0$, otherwise it recurses on $x$, 
   binds the value to $y$ and outputs $1 + y$. By the typing rule of recursive definitions, the variable $x$ is a thunk,
   meaning that it must be forced before executing it.
\end{example}

By having fine-grained control over which operations have a cost, it is possible to orchestrate your program with $\charge c$
operations in order to encode different cost models.
For instance, if we want to keep track of how many coins were tossed when running the geometric distribution, we can modify it
as such
\[
\fix x.\, \charge 1 ; (\produce 0) \oplus_{0.5} (y \leftarrow \force x; \produce (1 + y)) : F \nat
\]
\begin{example}[Deterministic Programs]
    The $\mathsf{charge}$ operation can also be used to keep track of the number of recursive calls
    in your program. For instance, a recursive program that computes the factorial function can be instrumented
    to count the number of recursive calls as follows:
    \[\fix f.\, \lamb n {\ifthenelse{n}{(\produce 0)}{(\charge 1; n * (\force f)(n-1))}}\]
    Whenever the if-guard is false, the cost is incremented by $1$ and the function is recursively called.
\end{example}
\subsection{Lists}

Frequently, cost analyses are defined for algorithms defined over inductive data types, such as lists. 
As such, we will also extend our language with lists over value types.
%
% \begin{figure}
% \begin{align*}
%   &\tau \defined \cdots \smid \listty \tau\\
%   &V \defined \cdots \smid \mathsf{nil} \smid \mathsf{cons}\, V_1 \, V_2\\
%   &t \defined \cdots \smid \caseList{x}{t}{u}\\
%   &T \defined \cdots \smid \caseList{x}{t}{u}
% \end{align*}    

% \begin{mathpar}
%     \inferrule{~}{\Gamma \vdash^v \mathsf{nil} : \listty \tau}
%     \and
%     \inferrule{\Gamma \vdash^v V_1 : \tau \\ \Gamma \vdash^v V_2 : \listty \tau}{\Gamma \vdash^v \mathsf{cons}\, V_1 \, V_2 : \listty \tau}
%     \and
%     \inferrule{\Gamma \vdash_v V : \listty \tau \\ \Gamma \vdash^c t : \overline{\tau} \\ \Gamma, x : \tau, xs : \listty \tau \vdash^c u : \overline \tau}
%     {\Gamma \vdash^c \caseList{V}{t}{u} : \overline{\tau}}
% \end{mathpar}
% \end{figure}
\begin{align*}
  &\tau \defined \cdots \smid \listty \tau\\
  &V \defined \cdots \smid \mathsf{nil} \smid \mathsf{cons}\, V_1 \, V_2\\
  &t \defined \cdots \smid (\caseList{x}{t}{u})\\
  &T \defined \cdots \smid (\caseList{x}{t}{u})\\
  &C \defined \cdots \smid (\caseList{x}{C}{u}) \smid (\caseList{x}{t}{C}) 
\end{align*}

\begin{mathpar}
    \inferrule{~}{\Gamma \vdash^v \mathsf{nil} : \listty \tau}
    \and
    \inferrule{\Gamma \vdash^v V_1 : \tau \\ \Gamma \vdash^v V_2 : \listty \tau}{\Gamma \vdash^v \mathsf{cons}\, V_1 \, V_2 : \listty \tau}
    \and
    \inferrule{\Gamma \vdash_v V : \listty \tau \\ \Gamma \vdash^c t : \overline{\tau} \\ \Gamma, x : \tau, xs : \listty \tau \vdash^c u : \overline \tau}
    {\Gamma \vdash^c \caseList{V}{t}{u} : \overline{\tau}}
\end{mathpar}

The primitive $\mathsf{nil}$ is the empty list, $\mathsf{cons}$ appends a value to the front of a list and $\mathsf{case}$ is for 
pattern-matching on lists and, in the presence of $\fix$, can be used for defining non-structurally recursive functions over lists.

\begin{example}
    The function $\mathsf{length}$ that computes the length of a list and a binary version of the familiar filter function $\mathsf{biFilter}$
    that outputs two lists, one for the true elements and one for the false elements, are, respectively, defined in the left and right parts of \Cref{fig:programs}. 
    Note that since we have adopted a $\nat$-valued if-statement, the predicate $p$ above outputs a natural number. Furthermore, since 
    the recursion operation adds a thunk to the context, in order to call the recursive function you must first force its execution.
\end{example}

\begin{comment}
\begin{figure}
\begin{align*}
&\fix f : \listty \nat \to F (\listty \nat) .\,\\
&\lambda l : \listty \nat.\,\\
&\mathsf{case} \, l \, \mathsf{of}\\
& \mid \mathsf{nil} \Rightarrow \\
& \quad \produce \mathsf{nil}\\
&\mid \mathsf{(hd, tl)} \Rightarrow \\
& \quad len \leftarrow \app \mathsf{length}\, l\\
& \quad r \leftarrow \app \rand len\\
& \quad pivot \leftarrow l[r]\\
& \quad (l_1, l_2) \leftarrow \app {\app{biFilter} {(\lamb n {\charge 1; n \leq pivot})}} (\mathsf{drop}\, l)\\
& \quad less \leftarrow \app {\force f} {l_1}\\
& \quad greater \leftarrow \app {\force f} {l_2}\\
& \quad \produce (less \mdoubleplus pivot :: greater)\\
\end{align*}
\caption{Randomized quicksort}
\label{fig:quicksort}
\end{figure}

\begin{example}
\label{ex:quicksort}
 In \Cref{fig:quicksort} we have defined     a randomized version of the quicksort algorithm that counts the number of comparisons done.
In in, we are accessing the $r$-th element of a list $l$ using the familiar syntax $l[r]$.
The algorithm is very similar to the non-randomized quicksort with the exception of choosing the pivot element
with the command $r \leftarrow \mathsf{rand}\, len$ that uniformly chooses an element from the set $\{0,\dots, |l|\}$,
where $\len -$ is the length function.
\end{example}
\end{comment}

We conclude this section by mentioning that there are many other sensible extensions, such as recursive and sum types. For our purposes,
they are not necessary and so, in order to keep the language simple, we omit them. That being said, from a semantic point of view,
these extensions are well-understood and straightforward to be accommodated by the denotational semantics we present in
\Cref{sec:semantics}.

\subsection{Equational Theory}
\label{sec:ineq}

% \begin{figure}
%     \begin{align*}
%         x \oplus_0 y &= x\\
%         x \oplus_{p} y &= y \oplus_{1-p} x\\
%         x \oplus_p x &= x\\
%         x\oplus_p (y \oplus_q z) &= (x \oplus_{pq} y) \oplus_{\frac{p(1-q)}{1 - pq}} z
%     \end{align*}
% \caption{Barycentric Algebra Axioms}
% \label{fig:barycentric}
% \end{figure}

We want to define a syntactic sound approximation to the expected cost of programs. We do this by extending the usual equational
theory of CBPV with rules for the monoid structure of the $\mathsf{charge}$ operation. We present some of the equational theory
in \Cref{fig:ineq}, with other rules which are standard in CBPV languages shown in \Cref{app:cbpvsemantics}. The first two rules 
are the familiar $\beta$ and $\eta$-rules for the arrow type, the \textsc{0Mon} and \textsc{ActMon} rules are the monoid equations for the charge operation.
The rule \textsc{ThunkForce} says that forcing a thunked
computation is the same thing as running the computation, the rules \textsc{IfZ} and \textsc{IfS} explain how if-statements interact with natural numbers and
the rule \textsc{Fix} is the fixed point equation that unfolds one recursive call of the recursive computation $t$.

\subsection{Operational Semantics}

Since we are interested in modeling the cost of running programs, we will define an operational cost semantics which is closer to the execution model of programs. 
When defining semantics for probabilistic languages with continuous distributions, one must be careful to define it so that it is a measurable function. 

In this section, we begin by showing the the $\meta$ syntax can be equipped with a measurable space structure that makes the natural syntax operations, such as substitution, measurable. Then, the operational semantics can be defined as a Markov kernel over the syntax, as it is usually done in probabilistic operational semantics for continuous distributions \cite{ehrhard2017, wqbs}. After defining the operational semantics, we conclude by proving the subject reduction property.

\begin{figure*}
\begin{mathpar}
    
  \inferrule[$\beta$-law]{\Gamma, x : \tau \vdashcomp t : \overline{\tau} \\ \Gamma \vdash_v V : \tau}{\Gamma \vdash \subst t V x = \app {(\lamb x t)} V : \overline{\tau}}
  \and
  \inferrule[$\eta$-law]{\Gamma \vdashcomp t : \tau \to \overline{\tau}}{\Gamma \vdashcomp \app {(\lamb x {\app t x})} = t : \tau \to \overline{\tau}}
%
  % \\
%
  % \inferrule{\Gamma, x : \tau \vdash_c t : \overline{\tau} \\ \Gamma \vdash_v V : \tau}{\Gamma \vdash \subst t V x = \app {(\lamb x t)} V : \overline{\tau}}
%
  % \and
%
  % \inferrule{\Gamma \vdash V_1 : \tau \\ \Gamma \vdash V_2 : \tau}{\Gamma \vdash \produce V_1 = \produce V_2 : F \tau}
%
  \and
  \inferrule[0mon]{\Gamma \vdashcomp t : \overline \tau}{\Gamma \vdash (\charge 0; t) = t : \overline{\tau}}
\and
  \inferrule[ActMon]{~}{\Gamma \vdash \charge c; \charge d = \charge {c + d} : F 1}
\and
\inferrule[ThunkForce]{\Gamma \vdashcomp t : \overline\tau}{\Gamma \vdash \force (\thunk (t)) = t: \overline{\tau}}
  \\
  \inferrule[IfZ]{\Gamma \vdashcomp t : \overline\tau \quad \Gamma \vdashcomp u : \overline\tau}{\Gamma \vdash \ifthenelse 0 t u = t: \overline \tau}
\and
  \inferrule[IfS]{\Gamma \vdashcomp t : \overline\tau \quad \Gamma \vdashcomp u : \overline\tau}{\Gamma \vdash \ifthenelse {(n + 1)} t u = u : \overline \tau}
  \and
  \inferrule[Fix]{\Gamma, x : U \overline \tau \vdashcomp t : \overline \tau}{\Gamma \vdash (\fix x.\, t) = \subst t x {\thunk (\fix x.\, t)} : \overline{\tau}}  
\end{mathpar}
\caption{Equational Theory (Selected Rules)}
\label{fig:ineq}
\end{figure*}

\paragraph{Syntax Spaces and Kernels}
    Before defining the operational semantics, we need some definitions from measure theory.
    \begin{definition}
        A measurable space is a pair $(X, \Sigma_X \subseteq \mathcal{P}(X))$, where $X$ is a set and $\Sigma$ is a $\sigma$-algebra, i.e. a collection of subsets that contains the empty set and is closed under complements and countable union.
    \end{definition}

    \begin{definition}
        A measurable function $f : (X, \Sigma_X) \to (Y, \Sigma_Y)$ is a function $f : X \to Y$ such that for every $A \in \Sigma_Y$, $f^{-1}(A) \in \Sigma_X$.
    \end{definition}
    \begin{definition}
        A subprobability distribution over a measurable space $(X, \Sigma_X)$ is a function $\mu : \Sigma_X \to [0,1]$ such that $\mu(\emptyset) = 0$,
        $\mu(X) \leq 1$ and $\mu(\uplus_{n \in \nat} A_n) = \sum_{n\in \nat} \mu(A_n)$.
    \end{definition}

    The operational semantics will be modeled as a (sub)Markov kernel, a generalization of transition matrices and Markov chains.
    \begin{definition}
        A sub-Markov kernel between measurable spaces $(X, \Sigma_X)$ and $(Y, \Sigma_Y)$ is a function $f : X \times \Sigma_Y \to [0,1]$ such that
    \begin{itemize}
        \item For every $x : X$, $f(x,-) : \Sigma_Y \to [0,1]$ is a subprobability distribution
        \item For every $A : \Sigma_Y$, $f(-, A) : X \to [0,1]$ is a measurable function
    \end{itemize}
    \end{definition}    
    
    We denote the set of computation terms by $\Lambda$, the set of values by $\val$ and the set of terminal computations by $T$. Let $t$ (resp. $V$) be a computation (resp. value), fix the term traversal order left-to-right and let $z_1, z_2,\dots z_n, \dots$ be a sequence of distinct and ordered variables disjoint from the set of term variables. The traversal order gives rise to a canonical enumeration of $t$'s (resp. $V$'s) occurrences of numerals $r \in \R$, which we denote by the sequence $r_1, r_2, \dots, r_n$. By substituting these occurrences by the variables $z_1, z_2, \dots, z_n$, we obtain the term $\subst t {z_1,\dots, z_n} {r_1,\dots, r_n}$ (resp. $\subst V {z_1,\dots, z_n} {r_1,\dots, r_n}$). Let $\Lambda_n$ (resp. $\val_n$) be the set of such substituted terms with exactly $n$ numerals. 

    Note that the sets $\Lambda_n$ and $\val_n$ are countable and that there are bijections $\Lambda \cong \Sigma_{n : \nat, t : \Lambda_n} \R^n$
    and $\val \cong \Sigma_{n : \nat, t : \val_n} \R^n$, where $\Sigma$ is the dependent sum operation: for instance, every computation term $t$ can be decomposed into a sequence of its numerals $r_1,\dots, r_n$ and substituted term $\subst t {z_1,\dots,z_n}{r_1,\dots, r_n}$, and, conversely, every sequence of numerals and substituted term $t$ can be mapped to the term $\subst t {r_1,\dots, r_n}{z_1,\dots,z_n}$ --- note the reversed order of substitution. 
    
    Therefore, we can equip $\Lambda$ with the coproduct $\sigma$-algebra: $(\Lambda, \Sigma_\Lambda) = \Sigma_{n : \nat, t : \Lambda_n} (\R^n, \Sigma_{\R^n})$. More concretely, a subset $A \subset \Lambda$ is measurable if, and only if,
    \[\forall n \in \nat, t:\Lambda_n, \{(r_1,\dots,r_n) \in \R^n \mid \subst t {r_1,\dots, r_n} {z_1,\dots, z_n} \in A\} \in \Sigma_{\R^n}\]
    The measurable space structure of $\val$ is defined using a similar coproduct. 
    
    Given a pair of a context $\Gamma$ and a computation type $\overline \tau$, the measurable spaces $\Lambda^{\Gamma \vdash \overline \tau}$ and $T^{\Gamma \vdash \overline \tau}$ are the subspaces of well-typed computations and terminal computations under context $\Gamma$ and output type $\overline{\tau}$, respectively. Given a value type $\tau$, the measurable space $\val^{\Gamma \vdash \tau} \subseteq \val$ is the subspace of well-typed values under context $\Gamma$ and output type $\tau$. The measurable space structures of $\Lambda^{\Gamma\vdash\overline{\tau}}$, $T$, $T^{\Gamma\vdash\overline{\tau}}$ and $T^{\Gamma\vdash\tau}$ are defined using the appropriate subspace $\sigma$-algebras. The following metatheoretic lemmas are useful and standard.

    \begin{lemma}
        If $\Gamma, x : \tau \vdashcomp t : \overline{\tau}$ and $\Gamma \vdashval V : \tau$ then
        $\Gamma \vdashcomp \subst t V x : \overline \tau$.
    \end{lemma}
    \begin{lemma}[cf. Lemma~3.7 of Ehrhard et al.~\cite{ehrhard2017}]
        For every variable $x$, the function $\subst \cdot \cdot x : \Lambda \times \val \to \Lambda$ is measurable.
    \end{lemma}

    Therefore, for every context $\Gamma$, computation type $\overline{\tau}$ and value type $\tau$, the substitution function restricts to measurable functions $\subst \cdot \cdot x : \Lambda^{\Gamma, x : \tau \vdash \overline\tau} \times \val^{\Gamma \vdash \tau} \to \Lambda^{\Gamma \vdash \overline\tau}$.
\paragraph{Operational Kernels}
In order to capture the cost of running a computation, we define \emph{costful} kernels as a sub-Markov kernel $X \times \Sigma_{\nat \times Y} \to [0,1]$.
Given two costful kernels $f$ and $g$, their composition $f \circ g : X \times \Sigma_{\nat \times Z} \to [0,1]$ is defined, with slight abuse of notation, as:
    \[(g \circ f)(x, n_1 + n_2, C) = \int_{Y}g(y, n_2, C)f(x, n_1, -)(\diff y)\]
    %\int_{B \times \nat}\int_{A\times \nat} \delta_{(n_1 + n_2, y)} g(x)(\diff y, \diff n_2) f(a)(\diff x, \diff n_1)\]
%
In plain terms, every term reduces to a subprobability distribution over costs and terminal computations. Their composition is defined so that the probability that the composition cost is $n_1 + n_2$ is the probability that the input $f$ will cost $n_1$ and the continuation $g$ will cost $n_2$, i.e. the product of both events averaged out using an integral. We use the Haskell syntax $\bind$ for kernel composition.

We can now define the operational semantics as the limit of a sequence of approximate semantics given by costful kernels $\Downarrow_n : \Lambda \times \Sigma_{\nat \times T} \to [0,1]$. The approximate semantics $\Downarrow_n$ are defined by recursion on $n$, where the base case is defined as $\Downarrow_0 = \bot$, i.e. the $0$ measure. When $n > 0$, the recursive definition is depicted in \Cref{fig:opsem}. 

Though we are using the familiar relational definition of operational semantics, they are functional in nature. We use the notation of V\'{a}k\'{a}r et al.~\cite{wqbs}, where the inference rule
\[
\inferrule{k_1(t) w_1\\ k_2(t, w_1) w_2\\ \cdots\\ k_n(t, w_1,\dots, w_n) v}{l(t) f(t, w_1,\dots,w_n,v)}
\]
denotes the kernel $k_1(t) \bind (\lambda w_1.\, k_2(t, w_1) \bind \dots \delta_{f(t, w_1,\dots,w_n, v)})$. We also simplify the presentation by using \emph{guarded} kernel composition. For instance, in the $\beta$-reduction rule, whenever $t$ reduces to something which is not a $\lambda$-abstraction, the kernel composition loops. Besides the
non-standard presentation, the semantics is a fairly standard CBPV big-step semantics \cite{levy2001call}. For example, terminal computations cannot reduce any further, so they output a point mass distribution over $0$ cost and themselves. In the case of effectful operations, the charge operation steps to a point mass distribution over $()$ and the cost; the sample operation reduces to an independent distribution of the point mass distribution at $0$ and the Lebesgue uniform measure $\lambda$ on the interval $[0,1]$.

It follows by a simple induction that the semantics $\Downarrow_n$ is monotonic in $n$. Since the space of subprobability distributions forms a CPO, we can define the semantics as the supremum of its finite approximations ${\Downarrow} = \bigsqcup_n {\Downarrow_n}$. As usual, it is possible to prove subject reduction by induction on well-typed terms.

\begin{lemma}[Subject reduction]
\label{lem:subject}
    If $\Gamma \vdashcomp t : \overline \tau$, then the composition $\DownarrowFun \circ\, \iota : \Lambda^{\Gamma\vdash\overline \tau} \to \Psub(\nat \times T)$ 
    factors as $\Lambda^{\Gamma\vdash\overline \tau} \to \Psub(\nat \times T^{\Gamma\vdash \overline \tau}) \hookrightarrow \Psub(\nat \times T)$,
    where $\iota : \Lambda^{\Gamma \vdash \overline \tau} \hookrightarrow \Lambda$ is the inclusion function. More colloquially, well-typedness is stable under the operational semantics.
\end{lemma}
\begin{figure}
\begin{mathpar}
    \inferrule{~}{\produce \, V \Downarrow_n \delta_{(0, \produce V)}}
    \and
    \inferrule{t \Downarrow_n \mu}{\force (\thunk t) \Downarrow_n \mu}
    \and
    \inferrule{~}{\uniform \Downarrow_n \delta_0 \otimes \lambda}
    \and
    \inferrule{~}{\lamb x t \Downarrow_n \delta_{(0, \lamb x t)}}\\
    \inferrule{t \Downarrow_n \lamb x {u} \\ \subst u V x \Downarrow_{n-1} \mu}{\app t V \Downarrow_n \mu}
    \and
    \inferrule{~}{\charge r \Downarrow_n \delta_{(r, \produce ())}}\\
\inferrule{t \Downarrow_n \produce V \\ \subst u V x \Downarrow_{n-1} \mu}{(x \leftarrow t); u \Downarrow_n \mu}
\and
\inferrule{\subst t {\thunk \fix x.\, t} x \Downarrow_{n-1} \mu}{\fix x.\, t \Downarrow_n \mu}
\and
\inferrule{t \Downarrow_n \mu}{\ifthenelse 0 t u \Downarrow_n \mu}
\and
\inferrule{u \Downarrow_n \mu}{\ifthenelse {(n+1)} t u \Downarrow_n \mu}
\and
\inferrule{\subst t {V_1, V_2} {x_1, x_2} \Downarrow_{n-1} \mu}{\letin {(x_1, x_2)} {(V_1, V_2)} t \Downarrow_{n} \mu}
\\
\inferrule{t \Downarrow_{n-1} \mu}{\caseList{\mathsf{nil}}{t}{u} \Downarrow_n \mu}
\and
\inferrule{\subst u {V_1, V_2} {x, xs} \Downarrow_{n-1} \mu}{\caseList{(\mathsf{cons}\, V_1\, V_2)}{t}{u} \Downarrow_n \mu}
\end{mathpar}
\caption{Big-Step Operational Semantics}
\label{fig:opsem}
\end{figure}

\section{Denotational Semantics}
\label{sec:semantics}

This section presents two concrete denotational semantics to our language:
\begin{itemize}
    \item A cost semantics that serves as a denotational baseline for compositionally computing the
    cost distribution of probabilistic programs.
    \item An expected cost semantics that, while it cannot reason about as many quantitative properties of cost
    as the cost semantics, such as tail-bounds and higher-moments, it provides a compositional account to the \emph{expected cost}.
\end{itemize}
Both semantics will be defined over the category $\wqbs$ of $\omega$-quasi Borel spaces \cite{wqbs}, a Cartesian closed category that admits
a probabilistic powerdomain of subprobability distributions $\Psub$. By using the writer monad transformer $\Psub(\cmon \times -)$, 
it can also accommodate cost operations, as we explain in \Cref{sec:wqbs}. With this monad it is possible to define the expected cost to be
the expected value of the cost distribution $\Psub(\cmon)$.

Unfortunately, this approach is non-compositional. In order to compute the expected cost of a program of type $F\tau$,
we must first compute its (compositional) semantics which can then be used to obtain a distribution over the cost and then apply the expectation formula to it.
We work around this issue by constructing a novel expected cost monad in \Cref{sec:expectsem} that makes the expected cost a part of the semantics 
and, as such, it is compositionally computed. We start this section by going over important definitions and constructions for $\wqbs$, we then define the 
cost semantics, followed by the expected cost semantics.

\subsection{$\omega$-quasi Borel spaces}
\label{sec:wqbs}

We now introduce the semantic machinery used in the interpretation of $\meta$. Due to requirement of higher-order functions, probability and unbounded
recursion, we are somewhat limited in terms of which semantic domain to use. We use the category of $\omega$-quasi Borel spaces, a domain-theoretic version of quasi-Borel spaces \cite{qbs}.

\begin{definition}[\cite{wqbs}]
    An $\omega$-quasi Borel space is a triple $(X, \leq, M_X)$ such that, $(X, \leq)$ is a $\omega$-complete partial order ($\omega$CPO),
    i.e. it is a partial order closed under suprema of ascending sequences, and $M_X \subseteq \R \to X$
    is the set of \emph{random elements} with the following properties:
    \begin{itemize}
        \item All constant functions are in $M_X$
        \item If $f : \R \to \R$ is a measurable function and $p \in M_X$, then $p \circ f \in M_X$
        \item If $\R = \bigcup_{n\in \nat} U_n$, where for every $n$, $U_n$ are pairwise-disjoint
         and Borel-measurable, and $\alpha_n \in M_X$ then the function $\alpha(x) = \alpha_n(x)$
         if, and only if, $x \in U_n$ is also an element of $M_X$.
         \item For every ascending chain $\{f_n\}_n \subseteq M_X$, i.e. for every $x \in \R$, $f_n(x) \leq f_{n+1}(x)$, 
         the pointwise supremum $\bigsqcup_n f_n $ is in $M_X$.
    \end{itemize}
\end{definition}

Note that, in the definition above, $\omega$CPOs do not assume the existence of a least element, e.g. for every set $X$,
the discrete poset $(X, =)$ is an $\omega$CPO.

\begin{definition}
    A measurable function between $\omega$-quasi Borel spaces is a Scott continuous function $f : X \to Y$  --- i.e. preserving
    suprema of ascending chains --- such that for every $p \in M_X$, $f \circ p \in M_Y$.
\end{definition}

\begin{definition}
    The category $\wqbs$ has $\omega$-quasi Borel spaces as objects and measurable functions as morphisms.
\end{definition}

\begin{theorem}[Section~3.3 of V\'{a}k\'{a}r et al.~\cite{wqbs}]
    The category $\wqbs$ is Cartesian closed.
\end{theorem}

Furthermore, there is a full and faithful functor $\cat{Sbs} \to \wqbs$, where $\cat{Sbs}$ is the category
of standard Borel spaces and measurable functions (cf. Proposition~15 \cite{qbs}). More concretely, if
you interpret a program that has as inputs and output standard Borel spaces,
its denotation in $\wqbs$ will be a measurable function, even if the program uses higher-order functions,
and any measurable function could potentially be the denotation of the program.

\paragraph{Inductive types}
As shown in previous work \cite{wqbs}, $\wqbs$ can also soundly accommodate full recursive types.
In particular, it can give semantics to lists over $A$ by solving the domain equation
$\listty A \cong 1 + A \times \listty A$. 

It is convenient that in $\wqbs$, the set of lists over $A$
with appropriate random elements and partial order is a solution to the domain equation and it is the smallest one, i.e. it is an initial algebra.
This means that when reasoning about lists expressed in $\meta$, you may assume that they are just the set of lists over sets.
% As it was noted in previous work \cite{kavvos2019}, the cost of programs defined over lists frequently only depends on 
% the size of the list, not on its elements. If we were to follow the denotational semantics, it only would
% give us a function from lists to expected cost, which is not amenable to size-only analysis. We will eventually
% show in \Cref{sec:examples} that by proving an additional lemma it is sometimes possible to abstract lists 
% by their length.

%Our approach differs slightly from the on in 
%\cite{kavvos2019}, where they define a translation between their lists into natural numbers. Though
%their solution is more convenient

\paragraph{Probability Monads}
It is possible to construct probabilistic powerdomains in $\wqbs$, making it possible to use this category as a semantic
basis for languages with probabilistic primitives. Furthermore, the $\omega$CPO structure can also be used
to construct a partiality monad, making it possible to give semantics to programs with unbounded recursion. We are assuming
familiarity with basic concepts from category theory such as monads and use the notation $(T, \eta^T, (-)^\#_T)$, where $T$ is
an endofunctor, $\eta^T_A : A \to T A$ and $(-)^\#_T: (A \Rightarrow TB) \to (TA \Rightarrow TB)$ are the unit and
bind natural transformations, respectively.
% \begin{definition}
%     Let $\cat{C}$ be a category, a monad is a triple $(T, \eta^T, (-)^{\#}_T)$ where $T:\cat{C} \to \cat{C}$ is a functor,
%     $\eta^T : id \to T$ and $(-)^\#_T : \cat{C}(-, T =) \to \cat{C}(T -, T =)$ are natural transformations such that
%     \begin{align*}
%         f^\#_T\circ \eta &= f\\
%         \eta^\#_T &= id\\
%         (f^\#_T \circ g)^\#_T &= f^\#_T \circ g^\#_T
%     \end{align*}
% \end{definition}
The monad is said to be \emph{strong} if there is a natural transformation $st_{A, B} : A \times T B \to T(A\times B)$
making certain diagrams commute \cite{moggi1989}. When it is clear from the context, we will simply write $\eta$ and $(-)^\#$, 
without the sub and superscript, respectively.

\begin{lemma}[Section~4.5 of V\'{a}k\'{a}r et al.~\cite{wqbs}]
    The category $\wqbs$ admits strong commutative monads $P$ and $\Psub$ of probability and sub-probability distributions, respectively.
\end{lemma}

Categorically, $\Psub$ is defined as a submonad of the continuation monad $( - \to [0, 1]) \to [0, 1]$ and
its monad structure is similar to the one from probability monads in $\cat{Meas}$, i.e. the unit at a point $a : A$ is given
by the point mass distribution $\delta_a$ and $f^\#(\mu)$ is given by integrating $f$ over the input distribution $\mu$. A
more detailed presentation of this construction will be given in \Cref{sec:preexpect}.

Furthermore, by construction, $\wqbs$ admits a morphism $\int_A : (\Psub A) \times (A \to \{0,1\}) \to [0,1]$ 
that maps a subprobability distribution and a ``measurable set'' of $A$ into its measure. For example, if $A$ is a measurable
space, for every measurable set $X : A \to \{0,1\}$, and for every subprobability distribution $\mu : \Psub A$, the map $(\mu, X) \mapsto \mu(X)$
is an $\wqbs$ morphism and is equal to $\int_A$.

\begin{comment}
As we have mentioned above, it is also possible to define a lifting monad in $\wqbs$ that adds a least element $\bot$ to a space, making them
\emph{pointed} $\omega$CPOs. This monad, combined with the $\omega$CPO structure, is used to guarantee the existence of fixed points of 
endomorphisms between pointed $\omega$CPOs.

\begin{definition}[Section 3.3 of \cite{wqbs}]
 The lifting monad in $\wqbs$ $X_\bot$ adds a fresh element $\bot$ to $X$, makes it the least element and the random elements $M_{X_\bot}$
 are the functions $f : \R \to X_\bot$ such that there is a Borel measurable set $\mathcal B$ and a map $\alpha : \R \to X$ in $M_X$
 such that $f(x) = \alpha (x)$ for $x \in \mathcal B$ and $\bot$ otherwise.
\end{definition}

\end{comment}
% As we will see in \Cref{sec:expectsem}, the lifting monad applied to the positive real numbers $(\R^+, =)$ can
% be thought of as the closed interval $[0, \infty]$, where $\infty$ is the least element of the partial order\footnote{which
% is different from the object used in the construction of $\Psub$, where the interval is ordered by the usual real line ordering,
% with $\infty$ being the largest element.}

The machinery we have defined so far is expressive enough to interpret $\meta$, with exception of
its cost operations. In non-effectful languages, the writer monad $(\cmon \times -)$ can be used to give 
semantics to cost operations such as $\charge c$. 
\begin{definition}
    If $(\cmon, 0, +)$ is a monoid, then $\cmon \times -$ is a monad --- the \emph{writer} monad --- where the unit at a point $a$ is $(0, a)$
    and given a morphism $f : A \to \cmon \times B$, $f^\#(c, a) = (c + (\pi_1 \circ f)(a), (\pi_2 \circ f)(a))$, where $\pi_i : A_1 \times A_2 \to A_i$ is the $i$-th projection.
\end{definition}

What follows is how to combine the non-probabilistic cost monad $(\cmon \times -)$ with
$\Psub$ in order to define a probabilistic cost semantics.

\subsection{A probabilistic cost semantics}
\label{sec:costsem}

Contrary to the deterministic case, the cost of a probabilistic computation is not a single value, but rather a distribution over costs. 
For instance, consider the program:
\[\cdot \vdashcomp (\charge 1; \produce 0) \oplus (\produce 2) : F \nat\]
It either returns $2$ without costing anything, or it returns $0$ with a cost of $1$. Denotationally, this program should
be the distribution $\frac 1 2 (\delta_{(1,0)} + \delta_{(0, 2)})$. With equal probability, the program will either
cost $1$ and output $0$ or cost $0$ and output $2$.

In the deterministic case, it is possible to encode the cost at the semantic-level by using the \emph{writer} monad $\cmon \times -$.
For probabilistic cost-analysis we can use the writer monad transformer.

\begin{lemma}
    If $T : \cat{C} \to \cat{C}$ is a strong monad then $T(\cmon \times -)$ is a strong monad. 
\end{lemma}
\begin{comment}
\begin{proof}
    The strength of a monad is a natural transformation $A \times T B \to T(A \times B)$. When
    instantiating $A$ to be $\cmon$, we can conclude that there is a distributive law between
    the writer monad and $T$, which allows us to conclude that $T(\cmon \times -)$ is a monad.
    Its strength is defined as 
$st^T;T (st^{\cmon \times -}) : A \times T(\cmon \times B) \to T( A \times (\cmon \times B)) \to T(\cmon \times (A \times B))$.
\end{proof}
\end{comment}

When instantiating $T$ to be the subprobability monad $P_{\leq 1}$, we get a monad for probabilistic cost,
which justifies the denotation of the program $(\charge 1; \produce 0) \oplus (\produce 2)$ being
a distribution of a pair of a cost and natural number. 

By using the monadic semantics of CBPV, we get a cost-aware probabilistic semantics, where
most of its definitions follow the standard monadic CBPV semantics shown in \Cref{app:cbpvsemantics} --- denoted
as $\semcs{\cdot}^v$ for values and $\semcs{\cdot}^c$ for computations. The noteworthy interpretations are 
for the effectful operations, whose semantics are depicted in \Cref{fig:semop}, and for the cost monoid,
which is interpreted as the additive natural numbers $(\nat, 0, +)$.
\begin{figure}
    \begin{mathpar}
        \semcs{\charge c} = \delta_{(c, ())}
        \and
        \semcs{\uniform} = \delta_0 \otimes \lambda
        \and \semcs{\fix x.\, t} = \bigsqcup_n \semcs{t}^n(\bot)
    \end{mathpar}
    \caption{Cost semantics of operations}
    \label{fig:semop}
\end{figure}

With this semantics, we now define the expected cost of a distribution:

\begin{definition}
\label{def:expected}
    Let $\mu : \Psub [0,\infty]$, its expected value is $\mathbb{E}(\mu) = \int_{[0,\infty]} x \diff\mu$.
\end{definition}

In the definition above we have chosen the most general domain for $\mathbb{E}$, but for every measurable subset 
$X\subseteq[0,\infty]$ the expected distribution formula can be restricted to distributions over $X$.
In particular, it is possible to restrict this function to have $\Psub(\nat)$ as its domain.

\begin{example}[Geometric Distribution]
    This semantics makes it possible to reason about the geometric distribution defined as the program\footnote{For the sake of simplicity we have elided the some of the bureaucracy of CBPV, such as $\produce$and $\force$.}
    $\cdot \vdash \fix x. 0 \oplus (1 + x) : F \nat$. It is possible to show that this program indeed denotes
    the geometric distribution by unfolding the semantics and obtaining the fixed point equation
    $\mu =  \frac{1}{2}(\delta_0 + P_{\leq 1}(\lambda x. 1 + x)(\mu))$, for which the geometric distribution
    is a solution. 
    
    Since we are interested in reasoning about the cost of programs, it is possible to reason about the
    expected amount of coins flipped during its execution by adding the $\charge 1$ operation as follows
    $\fix x.\, \mathsf{charge}_1; (0 \oplus (1 + x)) : F \nat$, i.e. whenever a new coin is flipped, as modeled 
    by the $\oplus$ operation, the cost increases by one. By construction, the cost distribution will also
    follow a geometric distribution. 
    
    If we want to compute the actual expected value, we must compute
    $\sum_{n:\nat} \frac{n}{2^{n}}$. This particular infinite sum can be calculated by using a standard trick.
    In the next section we show how to encode this trick in the semantics itself, 
    significantly simplifying the computation of the expected value.
\end{example}

\begin{theorem}(cf. \Cref{app:cbpvsemantics})
    For every computation $t$ and value $V$, $\semcs{\subst t V x} = \semcs{t}(\semcs{V})$. 
\end{theorem}

\begin{theorem}(cf. \Cref{app:cbpvsemantics})
    For every computation context $C$, if $\semcs{t} = \semcs{u}$ then $\semcs{C[t]} = \semcs{C[u]}$.
\end{theorem}

\subsection{A semantics for expected cost}
\label{sec:expectsem}

The semantics just proposed can compositionally compute the cost distribution of programs, but compositionality
is broken when computing its expected cost. Indeed, after computing the distribution, 
we must compute the expected cost of an arbitrarily complex distribution. We fix this by defining a novel expected cost monad.

We achieve this by defining a monad structure on the functor $\weight \times \Psub$: every computation will be
denoted by an extended positive real number, i.e. its expected cost, and a subprobability distribution over its output.
The monad's  unit at a point $a : A$ is the pair $(0, \delta_a)$ and the bind operation $(-)^\#$ adds the expected cost of the input
with the average of the expected cost of the output. Formally,
given an $\wqbs$ morphism $f : X \to [0,\infty] \times \Psub Y$, its Kleisli extension is the function 
$f^\#(r, \mu) = (r + \int (\pi_1 \circ f) \diff \mu, (\pi_2\circ f)^\#_{\Psub}(\mu))$.

\begin{theorem}
\label{th:expectmonad}
    The triple $([0,\infty] \times \Psub, \eta, (-)^\#)$ is a strong monad.
\end{theorem}
\begin{proof}
    The proof can be found in \Cref{app:proofs}.
\end{proof}

With this monad it is possible to define a new semantics for $\meta$ that interprets
the effectful operations a bit differently from the cost semantics, as we depict in
\Cref{fig:semopex}, where $\semec{\cdot}^c$ is the computation semantics while $\semcs{\cdot}^v$ is
the value semantics; the cost monoid is still interpreted as $\nat$. 

\begin{example}[Revisiting the geometric distribution]
Unfolding the expected cost $\pi_1 \circ \semec{\mathsf{geom}}$ gives us the fixed point equation
$E = 1 + (1 - \frac 1 2) E$, i.e. $E = 2$. This can be readily generalized to arbitrary $p\in[0,1]$,
giving the equation $E_p = 1 + (1 - p) E_p$.
\end{example}

As we have noted in the previous section,
the cost semantics can be used to reason about the expected cost by using \Cref{def:expected}.
Something which will play an important role in our soundness proof is the fact that this
definition interacts well with the monadic structure of $\Psub$.

\begin{figure}
    \begin{mathpar}
        \semec{\charge c} = (c, \delta_{()})
        \and
        \semec{\uniform} = (0, \lambda )
        \and \sem{\fix x.\, t} = \bigsqcup_n \semec{t}^n(\bot)
    \end{mathpar}
    \caption{Expected cost semantics of operations}
    \label{fig:semopex}
\end{figure}

\begin{lemma}
\label{lem:expected_linear}
    Let $\mu : \Psub A$ and $f : A \to \Psub ([0,\infty])$, $\mathbb E(f^\#(\mu)) = \int_A \mathbb E(f(a))\mu(\diff a)$.
\end{lemma}
\begin{proof}
    This can be proved by unfolding the definitions 
    \begin{align*}
        &E(f^\#(\mu)) = \int_{[0,\infty]} x \left ( \int_A f(a) \mu(\diff a)\right ) (\diff x) = \int_A \int_{[0,\infty]} x f(a)(\diff x)\mu(\diff a) = \int_A \mathbb E (f(a))\mu(\diff a)
    \end{align*}
    In the third equation we had to reorder the integrals, which is valid
    because $\Psub$ is commutative.
\end{proof}

% The equality above can be intuitively explained as the expected cost of running $\mu$,
% producing an output $a$ and running $f(a)$ being equal the cost of running $\mu$ plus
% the average of running $f$ given the input distribution. 

With this lemma in mind, we state some basic definitions and lemmas that allows us to describe
precisely how the cost and expected cost semantics relate.

\begin{definition}
    A monad morphism is a natural transformation $\gamma : T \to S$, where $(T, \eta^T, (-)^\#_T)$ 
    and $(S, \eta^S, (-)^\#_S)$ are monads over the same category, such that $\gamma \circ \eta^T = \eta^S$
    and $(\gamma \circ g)^\#_S \circ \gamma= \gamma \circ g^\#_T$, for every $g: A \to T B$.
\end{definition}

\begin{theorem}
\label{th:monmorphism}
    There is a monad morphism $E : P(\nat \times -) \to [0,\infty] \times P$.
\end{theorem}
\begin{proof}
    We define the morphism $E_A(\mu) = (\mathbb{E}(P(\pi_1)(\mu)), P(\pi_2)(\mu))$. The first monad morphism 
    equation follows by inspection and the second one follows mainly from \Cref{lem:expected_linear},
    when restricting it to the probabilistic distributions, i.e. total mass equal to $1$.
\end{proof}

\begin{lemma}
    \label{lem:monmor}
    The natural transformation $E$, when extend to subprobability distributions, is not a monad morphism.
\end{lemma}
\begin{proof}
     Let $\frac 1 2 (\delta_{(0,1)} + \delta_{(1, 2)})$
    be a distribution over $\cmon \times \nat$ and $f(0) = \frac 1 2 \delta_0$, $f(n+1) = 0$ be a subprobability
    kernel. It follows by inspection that $((E \circ f)^\#_{[0,\infty] \times \Psub} \circ E) (\mu) \neq (E \circ f^\#_{\Psub(\nat \times -)})(\mu)$
\end{proof}

\Cref{th:monmorphism} says that the different cost semantics interact
well in the probabilistic case. In the subprobabilistic case this is not 
true, as illustrated by \Cref{lem:monmor}. This formalizes the intuitions
behind the subtleties in the interaction of expected cost and non-termination 
explained in the previous section. This semantics also validates the following compositionality properties.

\begin{theorem}(cf. \Cref{app:cbpvsemantics})
    For every computation $t$ and value $V$, $\semec{\subst t V x} = \semec{t}(\semec{V})$. 
\end{theorem}

\begin{theorem}(cf. \Cref{app:cbpvsemantics})
    For every context $C$, if $\semec{t} = \semec{u}$ then $\semec{C[t]} = \semec{C[u]}$.
\end{theorem}

\section{Soundness Theorems}
\label{sec:soundness}

We have proposed four ways of reasoning about expected cost: by using the equational theory, the operational semantics 
or by using either of the denotational semantics. Something which is not clear at first is the extent of how
much these semantics are reasoning about the same property. We begin this section by proving soundness properties of
the expected cost semantics with respect to the denotational cost semantics by stating and proving a generalization of 
the \emph{effect simulation problem} \cite{filinski1996controlling}. Due to subtle interactions between cost, probability 
and non-termination, we provide separate analyses for the probabilistic and subprobabilitic cases.

In order to justify that the denotational semantics is reasoning about an operational notion of cost, we prove standard
soundness and adequacy results of the denotational cost semantics with respect to the operational cost semantics. We also
show that both cost and expected cost semantics are sound with respect to the equational theory.

% Furthermore, the subtle interactions between cost, probability
% and non-termination, makes 

% we divide our analysis of denotational soundness in the probabilistic and non-probabilistic cases.

% Now, we want to understand how they relate to one
% another. When we restrict the language and semantics to the probabilistic case, i.e. without unbounded
% recursion, we can prove strong guarantees about the different semantics. For instance, the expected cost
% semantics gives the same value as the expected cost of the cost distribution in the cost semantics. Furthermore,
% both of these semantics satisfy the same equations. 

% Unfortunately, in the presence of unbounded recursion,
% i.e. subprobability distributions, the two semantics are not ``the same'' anymore. What we show in this section
% is that in the presence of subprobability distributions, the expected cost semantics is an upper bound on the 
% expected cost of the cost distribution. Besides, at the equational level, the expected cost semantics is better
% suited for capturing the intended equational properties of cost, as we explain in \Cref{sec:eqsound}. We conclude 
% the section by proving soundness and adequacy results of the cost semantics with respect to the operational semantics.

\subsection{Denotational Soundness}
\label{sec:costsound}

The soundness property we are interested in is the following: given a closed program $\cdot \vdashcomp t : F \tau$,
then the expected value for the second marginal of $\semcs{t}$ is equal to $\pi_1(\semec{t})$. In the literature, 
similar properties have been called the \emph{effect simulation problem} and many semantic techniques have been 
developed for solving it, such as $\top\top$-lifting \cite{katsumata2013}.

Unfortunately, the available categorical techniques for effect simulation have two limitations. 
The first one is that it only proves properties about programs with ground type context and output, cf.
Theorem~7 of Katsumata \cite{katsumata2013}. The second one is the lack of existence of an appropriate monad morphism
$\beta : \Psub(\nat \times -) \to \weight \times \Psub$. Since we are trying to relate the expected value
of the first marginal with the weight given by the expected cost monad, $\beta$ would have to be equal to $E$, 
as defined above, which according to \Cref{lem:monmor} is not a monad morphism.

We deal with both of these issues by defining a two-level logical relation. This structure mimics the 
value/computation types distinction present in CBPV and, when compared to previous work \cite{katsumata2013}, 
gives stronger soundness properties. We begin by defining a probabilistic relational lifting in $\wqbs$.

\begin{definition}
    Let $\mathcal{R} \subseteq A \times B$ be a complete binary relation, i.e. it is closed under suprema of
    ascending sequences, its lifting
    $\mathcal{R}^{\#} \subseteq \Psub(A) \times \Psub(B)$ is defined as $\mu\, \mathcal{R}^\#\, \nu$
    if there is a distribution $\theta : \Psub(\mathcal R)$ such that its first and second
    marginals are, respectively, $\mu$ and $\nu$.
\end{definition}

This definition interacts well with the monadic structure of $\Psub$. Concretely, it is stable with respect to the unit and bind of $\Psub$.
This definition can be restricted to the probability monad $P$.  We now define our logical relations as two families of relations, one for value types 
and one for computation types:
\begin{align*}
    &\V_\tau \subseteq \semcs{\tau}^v \times \semec{\tau}^v &\C_{\overline\tau} \subseteq \semcs{\overline\tau}^c \times \semec{\overline\tau}^c \\
    &\V_\nat = \set{(n,n)}{n\in\nat}   &\C_{F \tau} = \set{((r, \nu), \mu)}{\mathbb E(\mu_1) \leq r \land \nu \V^\#_\tau \mu_2}\\
    &\V_{U\overline{\tau}} = \C_{\overline{\tau}}    &\C_{\tau \to \overline{\tau}} = \set{(f_1, f_2)}{\forall x_1, x_2, x_1 \V_\tau x_2 \Rightarrow f_1(x_1) \C_{\overline{\tau}} f_2(x_2)}\\
    &\V_{\tau_1 \times \tau_2} = \V_{\tau_1} \times \V_{\tau_2} &\\
    &\V_{\mathsf{list}(\tau)} = \mathsf{list}(\V_\tau)&
\end{align*}

Where $\mathsf{list}(\V_\tau)$ relates two lists if, and only if, the lists have the same length and are componentwise
related by $\V_\tau$. As it is often the case, the ``meat'' of the logical relations is in the definition of relation for the 
$F$ connective, where we specify the desired relationship between the expected value of the cost distribution and the 
real number given by the expected cost semantics. We can now state the fundamental theorem of logical relations. Its
proof is detailed in \Cref{app:denproof}.

\begin{theorem}
    For every $\Gamma = x_1 : \tau_1, \dots, x_n : \tau_n$, $\Gamma \vdashval V : \tau$, $\Gamma \vdashcomp t : \overline\tau$ and if for every $1 \leq i \leq n$, $\cdot \vdashval V_i : \tau_i$ and $\semcs{V_i} \VRel_{\tau_i} \semec{V_i}$, then 
    \begin{align*}
    &\semcs{\letin {\overline{x_i}} {\overline{V_i}} t}^c \CRel_{\overline \tau}^c \semec{\letin {\overline{x_i}} {\overline{V_i}} t}^c \text{ and }\\   
    &\semcs{\letin {\overline{x_i}} {\overline{V_i}} V}^v \VRel_\tau \semec{\letin {\overline{x_i}} {\overline{V_i}} V}^v,
    \end{align*}

    where the notation $\overline{x_i} = \overline{V_i}$ means a list of $n$ let-bindings or, in the case of values, 
    a list of substitutions.
\end{theorem}

The denotational soundness theorem now follows as a simple corollary.

\begin{theorem}
\label{cor:sound2}
    The expected-cost semantics is sound with respect to the cost semantics, i.e. for every program
    $\cdot \vdash_c t : F\tau$, the expected cost of the second marginal of $\semcs{t}^c$ is at most
    $\pi_1(\semec{t}^c)$. Furthermore, when restricted to the recursion-free fragment of $\meta$, these
    costs are equal.
\end{theorem}

In order to prove the recursion-free fragment property above we must modify the $\C_{F\tau}$ 
definition so that $\mathbb E(\mu_1) = r$. By propagating this change through the proof we get
the desired result.
\subsection{Equational Soundness}
\label{sec:eqsound}

In the previous section, we have argued that in the presence of unbounded recursion, the cost semantics has some undesirable properties
which are not present in the expected cost semantics. That being said, when it comes to the base equational theory of \Cref{fig:ineq},
both semantics are sound with respect to it, showing that there is still some harmony between them.
% \begin{theorem}%     If $\Gamma \vdashcomp t = u : \overline{\tau}$ then $\sem{t} \prec_{\overline{\tau}} \sem{u}$
% \end{theorem}
\begin{theorem}
\label{th:eqsound}
    If $\Gamma \vdashcomp t = u : \overline{\tau}$ then $\semec{t}^c = \semec{u}^c$ and $\semcs{t}^c = \semcs{u}^c$.
\end{theorem}
\begin{proof}
    The proof follows by induction on the equality rules, where the inductive cases follow directly from the inductive hypothesis
    while the base cases follow by inspection. For instance, the equation $\charge 0; t = t$ is true because $\nat$ is a monoid
    and $0$ is its unit. Once again, the full equational theory is shown in \Cref{fig:fulleqtheory} and the full proof can be
    found in \Cref{app:proofs}.
\end{proof}

It is useful to understand the extent of which these equational theories differ.
For instance, the cost semantics validates the equation $\bot; t = \bot = t; \bot$. This equation is too extreme 
for the purposes of expected cost, since it says that $\charge c; \bot = \bot$.
An even more egregious equation that it satisfies is $\fix x.\, \charge 1; x = \bot$. That equation says that 
the cost of infinity is the same as no cost at all, as long as the program does not terminate.

These equations are connected to the commutativity equation:
\[
    \inferrule{\Gamma \vdashcomp x : F \tau_1 \\ \Gamma \vdashcomp u : F \tau_2 \\ \Gamma, x : \tau_1, y : \tau_2 \vdashcomp t' : \overline\tau}{\Gamma \vdashcomp (x \leftarrow t; y \leftarrow u; t')  = (y \leftarrow u; x \leftarrow t; t') : \overline{\tau}}
\]
\begin{theorem}[cf. \Cref{app:proofs}]
\label{th:commutativity}
    The cost semantics validates the commutativity equation.
\end{theorem}
% \begin{proof}
%     The proof can be found in \Cref{app:proofs}.
% \end{proof}
This equation is usually useful for reasoning about probabilistic programs. However, it is too strong for the purposes
of reasoning about expected cost. Indeed, consider the programs 
\begin{align*}
&t = \bot; \charge c; \produce ()\\
&u =\charge c; \bot; \produce ()    
\end{align*}
From an operational point of view, the first program will run a non-terminating program and never reach the $\charge c$ operation,
while the second one increases the cost by $c$ and then diverges. 

When it comes to modeling the cost of real programs,
these two programs should not be the same since, for instance, if the charge operation is modeling a monetary cost, such as a call to 
an API, only the second program will cost something. Fortunately, the expected cost monad is not commutative, making it a more  
physically-justified monad. When $c > 0$, the terms $t$ and $u$ show:
\begin{lemma}
    The monad $[0, \infty] \times \Psub$ is not commutative.
\end{lemma}
\begin{comment}
\begin{proof}
    When $c > 0$, the terms above are a counterexample:
    $\semec t = (0,0) \neq (c, 0) = \semec u$
\end{proof}
\end{comment}

That being said, assuming that the program $t$ in the commutativity equation has no cost and terminates with probability $1$,
the expected cost semantics validates the commutativity equation, and maximally so.

\begin{theorem}
\label{th:expectedcostcenter}
    The commutativity equation is sound in the expected cost model  if, and only if, $t$ 
    terminates with probability $1$ and has $0$ expected cost.
\end{theorem}
\begin{proof}
    In \Cref{app:proofs}, we prove this by showing that the probability monad $P$ is the center of $\weight \times  \Psub$ \cite{carette2023central},
    a concept generalizing the center of monoids and algebraic theories \cite{wraith1970}.
\end{proof}
%
% \paragraph{Inequational theory} 

% In the context of cost analysis, it can also be useful to reason about upper/lower bounds on the cost. The different versions
% of geometric distributions have already demonstrated this. Though reasoning semantically about these bounds is immediate, 
% this is not the case at the syntactic level. One way of addressing this is by defining \emph{inequational} theories where
% you can reason about programs being less than or equal to other programs.

% For the expected cost, however, one must be careful in terms of which rules to include. While the rule $\charge c \leq \charge d$,
% whenever $c \leq d$ is reasonable whenever the cost monoid comes equipped with a partial order, in the presence of probability
% it is not so easy to syntactically reason about these properties. A simple non-trivial example would be comparing the two variants of
% quicksort.

% Since in this work we were mainly interested on the denotational aspects of expected cost, we leave a more thorough investigation
% of the inequational properties of expected cost to future work.
\subsection{Operational Soundness and Adequacy}
\label{sec:opsound}

We conclude this section by proving some metatheoretic properties of the operational semantics and show how it relates to the denotational cost semantics. 
The main results of this section is the adequacy of the operational semantics with respect to the cost semantics and
the cost adequacy theorem of the operational semantics with respect to the expected cost semantics.

A consequence of \Cref{lem:subject} is that it becomes possible to compose the operational and denotational semantics and obtain morphisms $\sem{\Downarrow}^{\Gamma\vdash\overline\tau}:\Lambda^{\Gamma\vdash\overline \tau} \to (\semcs\Gamma \to \semcs{\overline \tau})$, as explained in Section~6.7 of V\'ak\'ar et al.~\cite{wqbs}. In particular, for
closed programs, $\sem{\Downarrow}^{\cdot\vdash\overline\tau}:\Lambda^{\cdot \vdash\overline \tau} \to (\semcs{\overline \tau})$. We now
state the soundness theorem.

\begin{theorem}[Soundness]
    For every closed computation $\cdot \vdashcomp t : \overline\tau$, $\semcs{\Downarrow\hspace{-3pt}(t)} \leq \semcs{t}$.
\end{theorem}
\begin{proof}
    As usual, since the operational semantics is defined as the supremum of $\Downarrow_n$, the proof follows by induction on $n$ and $t$.
    The proof can be found in \Cref{app:opsound}.
\end{proof}
The next step is proving the adequacy theorem. As it is usually the case, the proof follows from a logical relations argument.
For this relation, we will use a different, proof-irrelevant, relational lifting which we now define. 

\begin{definition}[cf. Section~4.3.1 of Katsumata et al.~\cite{katsumata2018codensity}]
    If $\mathcal{R} \subseteq A \times B$ is a binary relation, its lifting $\widetilde{\mathcal{R}} \subseteq \Psub(\nat \times A) \times \Psub (\nat \times B)$ is defined as $\mu \widetilde{\mathcal{R}} \nu$ if, and only if, for every $f : \nat \times A \to [0,1]$ and $g : \nat \times B \to [0,1]$
    such that if $a \mathcal{R} b$ and $g(n, b) \leq f(n, a)$, for every $n : \nat$, then $\left ( \int g \diff \nu \right) \leq \left ( \int f \diff \mu \right )$.
\end{definition}

We now show the definition of the adequacy logical relations.

\begin{align*}
    \vartriangleright_{\tau} &\subseteq \val^{\cdot \vdashval \tau} \times \semcs{\tau} &    \LHD_{\overline \tau} &\subseteq \Psub(\nat \times T^{\cdot \vdashcomp \overline \tau}) \times \semcs{\overline{\tau}}\\
    n \vartriangleright_{\nat} n &\iff \top & \mu \LHD_{F \tau} \nu &\iff \mu \widetilde{\vartriangleright_\tau} \nu\\
    r \vartriangleright_{\R} r &\iff \top & \mu \LHD_{\tau \to \overline \tau} f &\iff (\forall a V, V \vartriangleright_\tau a \Rightarrow\\
    & & &  (\lamb x t) \leftarrow \mu; \DownarrowFun(\subst t V x) \LHD_{\overline{\tau}} f(a))\\
    () \vartriangleright_{1} () &\iff \top\\
    c\vartriangleright_{\cmon} c &\iff \top \\
    (V_1, V_2)\vartriangleright_{\tau_1 \times \tau_2} (x, y) &\iff (V_1 \vartriangleright_{\tau_1} x) \land (V_2\vartriangleright_{\tau_2} y)\\
    V \vartriangleright_{\listty \tau} x &\iff V\, \mathsf{list}(\vartriangleright_\tau) \, x\\%|l| = |x| \land \forall i :\{0\dots |x|\}, l[i] \vartriangleright_{\tau} x[i]\\
     (\thunk\, t) \vartriangleright_{U\overline{\tau}} x &\iff \DownarrowFun(t) \LHD_{\overline\tau} x\\
\end{align*}

The relation $\mathsf{list}(\vartriangleright_\tau)$ holds if, and only if, both lists have the same length and are elementwise
related by $\vartriangleright_\tau$.

%
% \begin{align*}
%     \LHD_{\overline \tau} &\subseteq \Lambda^{\cdot \vdashcomp \overline \tau} \times \semcs{\overline{\tau}}\\
%     t \LHD_{F \tau} \mu &\iff \mu \leq \sem{\Downarrow t}\\
%     t \LHD_{\tau \to \overline \tau} f &\iff \forall V a, V \vartriangleright_\tau a \Rightarrow (\app t V) \LHD_{\overline{\tau}} f(a)\\
% \end{align*}

\begin{theorem}[Adequacy]
    For every closed computation $\cdot \vdashcomp t : F\tau$, $ \semcs{t} \leq \semcs{\DownarrowFun(t)}$.    
\end{theorem}
\begin{proof}
    The proof follows by showing a similar fundamental theorem of logical relations. The detailed proof can be found in \Cref{app:opproof}.
\end{proof}

\begin{corollary}[Cost Adequacy]
    For every closed computation $\cdot \vdashcomp t : F\tau$, where $\tau$ is a ground type, 
    \[\pi_1(\semec{t}) \leq \mathbb{E}(\Psub(\pi_1)(\semcs{\DownarrowFun(t)})\]  
\end{corollary}
\begin{proof}
    This theorem is a consequence of the adequacy theorem above and \Cref{cor:sound2}.
\end{proof}

\section{Pre-expectation Soundness}
\label{sec:preexpect}

We conclude the technical development of the semantics by proving a canonicity property of the expected cost monad with respect
to the pre-expectation monad --- frequently used for expected cost analysis \cite{kaminski2016, avanzini2021}. The main theorem in this
section shows that the expected cost monad is the minimal submonad of the pre-expectaction monad that can accommodate cost and 
sampling operations. This is important because it shows that any pre-expectation-based approach to expected cost analysis has to 
necessarily be a conservative extension of the expected cost monad.

We begin by revisiting some basic concepts from pre-expectation semantics. Then, we go over the categorical machinery used in
our canonicity proof. As an application, we provide a more robust and conceptual proof of Lemma~7.2 of Avanzini et al.~\cite{avanzini2020}, 
a key result in establishing the soundness of their approach. We conclude this section by proving \Cref{th:excostmin}, the 
canonicity property satisfied by the expected cost monad.

\subsection{Pre-expectation semantics}

Pre-expectations were originally proposed as a probabilistic generalization of predicate-transformer semantics \cite{kozen1983probabilistic}
and are an important tool used in the semantic analysis of probabilistic programs. They are defined as an element of the space 
$(X \to \weight) \to \weight$, for some $X$.

Pre-expectations are quite familiar in the functional programming world, being an instance of continuation
monads $K_A = (- \to A) \to A$, with $A = \weight$. These monads are quite flexible in terms of which effects
they can accommodate. Indeed, by choosing appropriate response types $A$, one can encode every monadic effect. 
For example, previous work \cite{kaminski2016,avanzini2020} has noted that pre-expectations can be used can 
accommodate $\charge n$ and sampling operations $t \oplus_p u$ as the pre-expectations $\lamb f {n + f\, *}$
and $\lamb f {p*(t \, f) + (1-p)*(u\, f)}$, respectively. In Avanzini et al.~\cite{avanzini2020}, 
such interpretations define a pre-expectation transformer $\mathsf{ect}$ that maps costful probabilistic programs 
to their pre-expectation semantics.

However, the flexibility afforded by using continuation monads comes at a cost. One of the goals of denotational 
semantics is using mathematical structures that are abstract enough to enable high-level mathematical reasoning, 
while reflecting useful programming invariants at the semantic level. Continuation semantics go against this tenet, 
as they lose structures and invariants of ``tailor made'' semantics, making them not completely adequate for reasoning 
about programs.

In the context of pre-expectation reasoning, as we saw in \Cref{sec:eqsound}, there are useful equations such as 
commutativity or variants thereof that are validated by the expected cost monad. On the other hand, continuation 
monads are notoriously non-commutative \cite{carette2023central}.

Another example of such a denotational inadequacy can be seen in Avanzini et al.~\cite{avanzini2020}.
The soundness of their cost analysis hinges on proving that the function $\mathsf{ect}$ mentioned above can be factored as the sum 
of a positive real number and an integral. In contrast, this property holds unconditionally in the expected cost semantics, 
making it more robust when it comes to language extensions. 

Lastly, the pre-expectation semantics has no a priori connection to probability theory, making it harder to apply
theorems from probability theory when reasoning about programs. This limitation will become clear in \Cref{sec:examples},
where we reason about a stochastic convex hull algorithm by using known theorems from probability theory.

\subsection{Factoring monad morphisms}
Factorization systems capture classes of morphisms that can uniquely factor any morphism. The most familiar example is that
of injections and surjections, where every function can be uniquely factored as the composition of a surjection followed by
an injection. However, in order to properly generalize this idea, one needs to be a bit more careful:

\begin{definition}[e.g. Section~5.5.1 of Borceux \cite{borceux1994I}]
    A (orthogonal) factorization system over a category $\cat{C}$ 
    is a pair $(\mathcal{E}, \mathcal{M})$, where both $\mathcal{E}$
    and $\mathcal{M}$ are classes of morphisms of $\cat{C}$ such that
    \begin{itemize}
        \item Every isomorphism belongs to both $\mathcal{E}$ and $\mathcal{M}$
        \item Both classes are closed under composition
        \item For every $e : \mathcal{E}$, $m : \mathcal{M}$ and morphisms $f$ and $g$ such that the appropriate
        diagram commutes, there is a unique $h$ making the following diagram commute:% https://q.uiver.app/#q=WzAsNCxbMCwwLCJBIl0sWzEsMCwiQiJdLFswLDEsIkMiXSxbMSwxLCJEIl0sWzIsMywibSIsMl0sWzAsMSwiZSJdLFswLDIsImYiLDJdLFsxLDMsImciXSxbMSwyLCJoIiwxLHsic3R5bGUiOnsiYm9keSI6eyJuYW1lIjoiZG90dGVkIn19fV1d
\[\begin{tikzcd}
	A & B \\
	C & D
	\arrow["f", from=1-1, to=1-2]
	\arrow["e"', from=1-1, to=2-1]
	\arrow["h"{description}, dotted, from=2-1, to=1-2]
	\arrow["m", from=1-2, to=2-2]
	\arrow["g"', from=2-1, to=2-2]
\end{tikzcd}\]
        \item Every morphism $f$ in $\cat{C}$ can be factored as $f = m \circ e$,
        where $m : \mathcal{M}$ and $e : \mathcal{E}$.
    \end{itemize}
\end{definition}

Note that under these axioms, the factorization is guaranteed to be unique up to isomorphism. If every
morphism in $\mathcal{M}$ is monic, then the factorization system is said to be \emph{epi-mono}.

\begin{example}
    As mentioned above, in the category $\cat{Set}$, the pair $(\{ f \mid f \text{ is surjective}\}, \{ f \mid f \text{ is injective}\})$
    is a factorization system. Given a function $f : A \to B$, its factorization is $A \to f(A) \to B$, where $f(A)$ is the image of $f$
    and the function $f(A) \to B$ is the set inclusion.
\end{example}

\begin{example}[Example~2.3 of Kammar and McDermott \cite{kammar2018}]
    In the category $\cat{CPO}$, the pair $(\{ f \smid f \text{ has dense image}\}, \{ f \smid f \text{ is order-reflecting}\})$ is
    a factorization system.
\end{example}

A natural extension of the factorization system above can be defined for the category $\wqbs$ by requiring the elements of $\mathcal{E}$ to be
the densely strong epic morphisms \cite{wqbs}. A morphism $f : X \to Y$ in $\wqbs$ is densely strong epic if it maps elements in $M_X$ into a Scott dense
subset of $M_Y$ with respect to the pointwise order. We now recall some definitions and a theorem proved by Kammar and McDermott \cite{kammar2018}
that play important role in how the probability monad in $\wqbs$ is defined.

\begin{definition}
    A monad structure is an endofunctor $S$ equipped with the monad natural transformations $\eta$ and $(-)^\#$ but without the necessity of satisfying the monad laws.
\end{definition}

\begin{definition}
    A monad structure morphism between monad structures $S_1$ and $S_2$ over the same categories is a natural transformation $S_1 \to S_2$ such that the monad morphism laws hold.
\end{definition}

\begin{theorem}[Theorem~2.5 of Kammar and McDermott \cite{kammar2018}\footnote{The proof is explained right after Theorem~2.6.}]
\label{th:fact}
    Let $(\mathcal{E}, \mathcal{M})$ be an epi-mono factorization structure, $S$ a monad structure, $T$ a monad and $\gamma : S \to T$ a monad structure morphism. If the factorization system is closed under $S$ then $\gamma$ can be factored as $S \to M \to T$, where $M$ is a monad, both morphisms are monad structure morphisms and each component of these morphisms is an element of $\mathcal{E}$ and $\mathcal{M}$.
\end{theorem}

This theorem was used by V\'{a}k\'{a}r et al.~\cite{wqbs} when defining their statistical powerdomain as the factor of a monad morphism into the pre-expectation monad $K_{\weight}$. Intuitively, the monad structures we are interested in are those where the monad laws hold up to an equivalence relation and the factored monads
are, roughly, the quotiented monad structures and satisfy a minimality universal property given by the uniqueness of factorizations.
\subsection{Expected cost monad canonicity}

We conclude this section by showing that the expected cost monad is the minimal submonad of $K_{\weight}$ 
that can accommodate cost and subprobability distributions. We extend the analysis done
by V\'{a}k\'{a}r et al.~\cite{wqbs} on using \Cref{th:fact} to define and prove the minimality of their statistical powerdomain. 
We now review its construction and then apply it to our semantics.

The ``random elements'' functor $[0,1] \to -_\bot$, where $(-)_\bot$ is the partiality monad, can be equipped with a monad structure similar to, but distinct from, the reader monad, as explained in Section 4.2 of \cite{wqbs}. The subprobability monad in $\wqbs$ is then defined using the following lemma:

\begin{lemma}
\label{lem:premonmor}
    There is a monad structure morphism $([0,1] \to -_\bot) \to K_{[0,1]}$.
\end{lemma}
\begin{proof}
    The proof is nearly identical to Section~4.2 of V\'{a}k\'{a}r et al.~\cite{wqbs}, with the exception that
    they are using a different, but analog, monad structure.
\end{proof}

The subprobability monad $\Psub$ is defined as the factorization of the monad structure morphism above as
$([0,1] \to -_\bot) \xtwoheadrightarrow{\psi } \Psub \hookrightarrow K_{[0,1]}$.
The component of the natural transformation above maps a pair $(f :[0,1] \to A_\bot, g : A \to [0,1])$ to $\int_{\mathsf{dom}(f)} (g\circ f)\diff \mathcal{U}$, where $\mathcal{U}$ is the Lebesgue uniform measure on $[0,1]$ and $\mathsf{dom}(f)$ is the domain of $f$.
\begin{comment}
The probability powerdomain can be defined as those measures that have total mass $1$ or, more formally,
as the equalizer $P X \to T X \rightrightarrows \weight$,
between the constant function $\mu \mapsto 1$ and the total mass function $\mu \mapsto \int_X \diff \mu$.
The subprobability monad is then defined as $\Psub X = P(X_\bot)$, where the mass of $\bot$ corresponds
to the probability of non-termination.
\end{comment}

We can also prove a similar lemma for the expected cost measure monad.

\begin{lemma}
    There is a monad structure on the functor $\weight \times ([0,1] \to -_\bot)$
    and a monad structure morphism $\gamma : \weight \times ([0,1] \to -_\bot) \to K_{\weight}$.
\end{lemma}
\begin{proof}
    There are similarities between the expected cost monad and the monad structure which is given by $\eta(a) = (0, \lamb r a)$ and $f^\#(r, g) = \left (r + \int_{\mathsf{dom}(g)} (\pi_1 \circ f \circ g) \diff \mathcal{U}, (\pi_2\circ f)^\#_{[0,1] \to -_\bot}(g) \right )$.
    The components of the monad structure morphism are defined as $\gamma_A(r, f) = \lamb k {r + \int (k \circ f)\diff \mathcal{U}}$.
    The proof that this is indeed a monad morphism follows by unfolding the definitions and \Cref{lem:premonmor}.
\end{proof}

By \Cref{th:fact}, and the fact that the factorization system in $\wqbs$ is stable under products, the monad structure morphism 
$\gamma$ above can be factored as $[0,\infty] \times ([0,1] \to -_\bot) \rightarrow M \rightarrow K_{\weight}$.

We now show the sense in which $\weight \times \Psub$ is canonical.

\begin{theorem}[cf. \Cref{app:proofs}]
\label{th:premorphism}
    The factor $M$ is isomorphic to the monad $\weight \times \Psub$.
\end{theorem}

The canonicity of the expected cost monad is given by it being the smallest submonad of the continuation monad
that can accommodate expected cost and subprobability distributions. A consequence of this fact is the following:

\begin{corollary}
    There is a monad morphism $\varphi : \weight \times \Psub \to K_{\weight}$ which is component-wise order-reflecting.
\end{corollary}

% \begin{example}
%     Consider the program $\mathsf{reject} = \fix x. \charge 2; (1 \oplus 2) \oplus (3 \oplus x)$ which
%     implements a rejection sampling algorithm for sampling from a uniform distribution over $3$ elements
%     and the program $\mathsf{unif} = \charge 1; \rand 3$. As expected, neither the expected cost nor the continuation semantics validate this equation. However, the expected cost monad provides a modular description of both programs and it is direct to reason about \emph{refinements} of programs.
% \end{example}

Intuitively, while the (sub)probability monad corresponds to the linear continuation monad, the expected cost 
monad corresponds to the affine continuation monad. An application of the results of this section is that it is
possible to recover a purely denotational proof of Lemma~7.1 of Avanzini et al.~\cite{avanzini2020} that is more robust with 
respect to language extensions.

\begin{theorem}[cf. Theorem~7 of Katsumata \cite{katsumata2013}]
    Let $\sempre{\cdot}$ be the monadic semantics of CBPV for the pre-expectation monad $K_{\weight}$.
    For every program $ \cdot \vdashcomp t : F \R$, $\sempre{t} = \varphi (\semec t)$. More explicitly,
    the pre-expectation semantics can be factored as $\sempre{t}(f) = \pi_1(\semec{t}) + \int f \diff (\pi_2(\semec{t}))$
\end{theorem}
\begin{comment}
\begin{proof}
    The proof is a variation of Theorem~7 in Katsumata \cite{katsumata2013}. 
\end{proof}
\end{comment}

Furthermore, by using the fact that the monad morphism is order-reflecting and, therefore, an injection,
we can also prove, for ground type closed computations $t_1$ and $t_2$, if $\sempre{t_1} = \sempre{t_2}$, 
then $\semec{t_1} = \semec{t_2}$. However, for non-ground types, this result is false. For instance,
the weak commutativity equation of \Cref{th:expectedcostcenter} is not validated by the continuation monad 
(cf. Example~5.11 of Carette et al.~\cite{carette2023central}).

We conclude by showing that the expected cost monad is the minimal cost submonad of $K_{\weight}$
such that it supports the operations $\mathsf{charge}$ and $\uniform$.

\begin{theorem}[cf. \Cref{app:proofs}]
\label{th:excostmin}
     For every monad $T$ with a componentwise monic monad morphism $\gamma : T \to K_{\weight}$ 
     and with the operations $\mathsf{charge}$ and $\uniform$ in its image,
    there is a unique monad morphism $\weight \times \Psub \hookrightarrow T$ that factors the monad morphism
    $\weight \times \Psub \to K_{\weight}$ through $\gamma$.
\end{theorem}
\begin{proof}[sketch.]
    We start by generating the free monad $\wqbs$ on operations $\mathsf{charge}: \nat \to F 1$ and $\uniform : F \R$.
    Similarly to Lemma~4.4 by V\`ak\`ar et al.~\cite{wqbs}, we then show that the unique monad morphism preserving
    $\mathsf{charge}$ and $\uniform$ between $F$ and $\weight \times \Psub$ is component-wise densely strong epic. 
    It then follows by the universal property of orthogonal factorization systems that for every 
    other submonad of $K_{\weight}$ there is a canonical monad morphism $\weight \times \Psub \to K_{\weight}$.
\end{proof}

\section{Examples}
\label{sec:examples}

In this section we will show how the expected cost semantics can be used to reason about the expected cost of probabilistic programs.
We present randomized algorithms and recursive stochastic processes, illustrating the versatility
of $\meta$. Due to space constraints, other examples and a more detailed analysis of these examples can be found in \Cref{app:examples}.

% \subsection{Random Walks and Brownian Motion}
% TODO

% The symmetric random walk is an important stochastic process. We will model
% one as the following program:
% \begin{align*}
%  &\mathsf{randomWalk} = \fix f : \nat \to \nat \to F 1.\, \lambda i : \nat \\
%  &\mathsf{if} \, i \leq a \lor  i \geq b\, \mathsf{then}\\
%  &  \quad \produce ()\\
%  & \mathsf{else}\\
%  &\quad \charge 1;\\
%  &\quad ((\force f) \, (i-1) \, j) \oplus ((\force f) \, (i+1) \, j)
% \end{align*}
% By using theorems from Martingale theory, it is possible to show that the program
% above has an expected cost of $|ab|$, for every $n$. Something useful about
% the symmetric random walk is that it can be used to generate more sophisticated
% stochastic processes. For instance, the Brownian motion can be defined as a kind of
% limit of finer and finer-grained scaled symmetric random walks. It is an interesting
% direction of future work studying if our semantics can prove properties of continuous-time
% stochastic processes by reasoning about their discrete-time behavior.

\subsection{Expected coin tosses}

A classic problem in basic probability theory is computing the expected number of coin flips necessary in order to 
obtain $n$ heads in a row. We can model this stochastic process as probabilistic program in \Cref{fig:expectedcoin},
where its expected cost semantics is also written down.
\begin{figure}[]
    \centering
    \begin{minipage}{.25\textwidth}
    \begin{align*}
    &\mathsf{ET}= \fix f : \nat \to F 1.\, \lambda n : \nat.\,\\
    &\mathsf{if} \, n \, \mathsf{then}\\
    &\quad \produce ()\\
    & \mathsf{else}\\
    &\quad \app {(\force f)} {(n-1)};\\
    &\quad\charge 1; \\
    &\quad (\produce () \oplus \app {(\force f)} n) 
    \end{align*}
    \end{minipage}%
    \begin{minipage}{0.40\textwidth}
    \begin{align*}
    \semec{\mathsf{ET}} = &\bigsqcup_n F^n(\bot)\text{, where}\\
    \\
    F =\, &\lambda \langle T_1, T_2\rangle.\,\lambda n.\\
    &\mathsf{if} \, n \, \mathsf{then}\\
    &\quad (0, \delta_{()})\\
    & \mathsf{else}\\
    &\quad \left (E_1, E_2 \right )
    \end{align*}
    \end{minipage}
    \begin{minipage}{0.25\textwidth}
    \begin{align*}
    &E_1 = T_1(n-1) + \mid T_2(n-1) \mid (1 + \frac 1 2 T_1(n))\\
    &E_2 = \frac {\mid T_2(n-1) \mid} 2(\delta_{()} + T_2(n))
    \end{align*}
    \end{minipage}
    \caption{Expected coin tosses and its expected cost semantics.}
    \label{fig:expectedcoin}
\end{figure}
% \begin{align*}
% &\fix f : \nat \to F 1.\, \lambda n : \nat.\,\\
% &\mathsf{if} \, n \, \mathsf{then}\\
% &\quad \produce ()\\
% & \mathsf{else}\\
% &\quad \app {(\force f)} {(n-1)};\\
% &\quad\charge 1; \\
% &\quad (\produce () \oplus \app {(\force f)} n) 
% \end{align*}

For every $n$, the program above simulates the probabilistic structure of flipping coins until obtaining $n$ heads
in a row. 
When its input is $0$, it outputs $()$ without flipping any coins. If the input is greater than $0$, in
order to flip $n$ heads in a row it must first flip $n - 1$ heads in a row --- hence the call to $f (n - 1)$ ---
flip a new coin while increasing the current counter by $1$ and, if it is heads, you have obtained $n$ heads in a 
row and may output $()$, otherwise you must recursively start the process again from $n$: the left and
right branches of $\oplus$, respectively. 

The denotational semantics of this program is also shown in \Cref{fig:expectedcoin}.
We use the notation $\langle T_1, T_2 \rangle$ to denote the pairing of functions $T_1 : \nat \to \weight$ 
and $T_2 : \nat \to \Psub(1)$. The cost component of the denotation adds the cost of running the program
with $n-1$ as input, while its distribution part must multiply the continuation, which costs $1$ for
the charge operation plus $\frac 1 2 T_1(n)$, for the recursive call. The distribution part $T_2$ is
the usual probabilistic semantics. Before we reason about the cost of this program, we must show that it 
terminates with probability $1$.

\begin{lemma}
    For every $n : \nat$, $\pi_2 \circ \semec{ET}(n) = \delta_{()}$, i.e. it terminates with probability $1$.
\end{lemma}
\begin{proof}
    The proof follows by induction on $n$. When $n = 0$, $\semec{ET}(0) = (0, \delta_{()})$.
    For the inductive case, assume $n > 0$. Unfolding the recursive definition gives us 
    \[(\pi_2 \circ \semec{ET}(n+1)) = \frac{|\semec{ET}(n)|}2(\delta_{()} + (\pi_2 \circ \semec{ET}(n+1)),\] 
    by the induction hypothesis $|\semec{ET}(n)| = 1$, which gives us that $\pi_2 \circ \semec{ET}(n+1) = \delta_{()}$.
\end{proof}

By unfolding the semantics and using the lemma above, we get the following recurrence relation:
\begin{align*}
    T_1(0) &= 0\\
    T_1(n+1) &= 1 + T_1(n) + \frac{1}{2}T_1(n + 1) 
\end{align*}
Which, by inspection, has the closed-form solution $T_1(n) = 2(2^n - 1)$.

\subsection{Randomized Quicksort}

\definecolor{neworange}{RGB}{227,103,28}
\definecolor{purple}{RGB}{150, 71, 184}
\definecolor{blue}{RGB}{0, 161, 255}
\definecolor{green}{RGB}{45, 210, 58}
\definecolor{cornellred}{RGB}{196,18,48}
\definecolor{cornellred}{RGB}{196,18,48}

\begin{figure}
\addtocounter{subfigure}{4}
\begin{minipage}{.65\textwidth}
\captionsetup{type=figure} % -- This line added
\subcaptionbox{Randomized parametric quicksort}[.9\linewidth]
{\begin{align*}
& \fix f : \listty \tau \to F (\listty \tau) .\,\\
& \lambda l : \listty \tau.\,\\
& \lambda pred : \tau \times \tau \to F \nat.\,\\
& \mathsf{case} \, l \, \mathsf{of}\\
& \mid \mathsf{nil} \Rightarrow \\
& \quad \produce \mathsf{nil}\\
& \mid \mathsf{(hd, tl)} \Rightarrow \\
& \quad len \leftarrow \app \mathsf{length}\, l\\
& \quad r \leftarrow \app \rand len\\
& \quad (pivot, l') \leftarrow \mathsf{nthDrop}\, r\, l\\
& \quad (l_1, l_2) \leftarrow \app {\app{biFilter} {pred}} l'\\
& \quad less \leftarrow \app {\force f} {l_1}\\
& \quad greater \leftarrow \app {\force f} {l_2}\\
& \quad \produce (less \mdoubleplus pivot :: greater)
\end{align*}}
\end{minipage}%
\begin{minipage}{.35\textwidth}
    \captionsetup{type=figure}
    \subcaptionbox{Natural number quicksort}[.9\linewidth]
{\begin{align*}
&qck_\nat = \fix f : \nat \to F \nat.\, \lambda n : \nat.\,\\
&\mathsf{if} \, n \, \mathsf{then}\\
&\quad \produce 0\\
& \mathsf{else}\\
&\quad \charge {(n - 1)}\\
&\quad x \leftarrow \rand \, n\\
&\quad (\force \, f)\, x\\
&\quad (\force \, f)\, (n - x - 1)\\
&\quad \produce n\\
& ~\\
& ~\\
& ~\\
& ~\\
\end{align*}}
\end{minipage}%
\caption{Quicksort algorithms}
\label{fig:parametric_quicksort}
\end{figure}

In \Cref{fig:parametric_quicksort} we present a program that implements a randomized quicksort parametric
on a total order on the type $\tau$. For the purposes of this analysis, we assume that the input list only 
has distinct elements. 

If we simply interpret the expected cost of this program denotationally, it will be a 
function mapping lists to real numbers. This is not how such an analysis is done in practice,
where the cost is given as a function mapping list lengths to cost.

In our semantics, the denotation of the program is hiding the fact that its cost only depends on
the length of its argument. We make this precise by defining a measurable function $qck_\nat : \nat \to \R \times \Psub \nat$
using the program in \Cref{fig:qcknat} that corresponds to the quicksort structure
assuming that the input is a natural number.
\begin{lemma}[cf. \Cref{app:proofs}]
\label{lem:permInv}
Assuming that $p : \tau \times \tau \to F\nat$ is a total order and the the input list has
no repetitions, the following program equation holds. %$n \leftarrow \mathsf{length}\, l; \mathsf{qck}_\nat \, n = l' \leftarrow \mathsf{quicksort}\, l'; \mathsf{length}\, l'$.
\[
\begin{aligned}
    &n \leftarrow \mathsf{length}\, l\\
    & \mathsf{qck}_\nat \, n
\end{aligned}
=
\begin{aligned}
    \quad
    &l' \leftarrow \mathsf{quicksort}\, (\charge 1; p)\, l\\
    & \mathsf{length}\, l'
\end{aligned}
\]
%
% https://q.uiver.app/#q=WzAsNCxbMSwwLCJcXG5hdCJdLFsxLDEsIkYgXFxuYXQiXSxbMCwwLCJcXGxpc3R0eSBcXG5hdCJdLFswLDEsIkYgXFxsaXN0dHkgXFxuYXQiXSxbMiwzLCJxY2tfe1xcbGlzdHR5IFxcbmF0fSIsMl0sWzIsMCwibGVuIl0sWzAsMSwicWNrX3tcXG5hdH0iXSxbMywxLCJGKGxlbikiLDJdXQ==
% \[\begin{tikzcd}
% 	{\listty \nat} & \nat \\
% 	{F \listty \nat} & {F \nat}
% 	\arrow["{qck_{\listty \nat}}"', from=1-1, to=2-1]
% 	\arrow["len", from=1-1, to=1-2]
% 	\arrow["{qck_{\nat}}", from=1-2, to=2-2]
% 	\arrow["{F(len)}"', from=2-1, to=2-2]
% \end{tikzcd}\]
\end{lemma}
We now conclude our analysis by using the program equation above and the following lemma
which is proved by induction on $n$.
\begin{lemma}
    For every $n : \nat$, $(\pi_2 \circ \mathsf{qck_\nat})(n) = \delta_n$, i.e. it terminates with probability $1$.
\end{lemma}
By unfolding the definition of $\pi_1\circ \semec{qck_\nat}$, we obtain the following expression
$\fix (f \mapsto n \mapsto \ifthenelse n {0} {(n-1) + \sum_i \frac 1 n (f(i) + f(n-i))})$ which
can be further simplified to the recurrence relation:
\begin{align*}
    T(0) &= 0\\
    T(n) &= n - 1 + \frac{2}{n}\sum_{i = 0}^{i-1}T(i)
\end{align*}
This allows us to conclude that $\mathsf{quickSort}$ has an expected cost of $O(n \log (n))$, which
can be proved by induction and using the observation that for monotonic functions $f$, $\sum_{i = 1}^{n-1}f(i) \leq \int_1^n f(x)\diff x$.

\subsection{Stochastic Convex Hull}

We conclude this section by going over a stochastic convex hull algorithm \cite{golin1988analysis}.
To the best of our knowledge, this example is outside of reach of
other logic/PL approaches to expected cost analysis, due to its combination
of continuous distributions, modular cost-probability interaction and 
non-trivial probabilistic reasoning. This stochastic variant of the algorithm
has linear expected cost analysis.

In \Cref{fig:stochconvexhull} we go over the implementation
of this algorithm in $\meta$. Its main body $\mathsf{convexHull}$
has three components. First, a list of points is sampled independently
and uniformly from the square $[0,1]^2$, then a preprocessing
stage $\mathsf{sieve}$ removes points that ``obviously'' are not in
the convex hull and then we call the traditional Graham scan algorithm ---
cf. Section $3.3.2$ of Preparata and Shamos \cite{preparata2012computational}.
The output of $\mathsf{convexHull \, n}$ is the convex hull of a set of
$n$ points uniformly distributed in the $[0,1]^2$ square. Since in this
example we are only interested in reasoning about its expected cost, we will
not reason about its functional correctness.

For the sake of presentation, we make a few simplifying assumptions:
we assume that there is an operation $\mathsf{iQ}$ that checks to see
if a point is inside a given convex quadrilateral and we assume that there
is an operation $\mathsf{clockOrNot} \, p_1\, p_2\, p_3$ that checks whether 
the lines $\overline{p_1p_2}$ and $\overline{p_2p_3}$ turn counterclockwise or not. 
Finally, we assume the existence of a parametric
binary total order relation $\sqsubseteq : \R^2 \to \R^2 \to \R^2 \to F \nat$
such that $\sqsubseteq \, p$ orders two points comparing their angle between 
the x-axis and their respective straight line with $p$ at the origin.

\begin{figure}[]
    \centering
    \begin{minipage}{.33\textwidth}
    \begin{align*}
    &\mathsf{unifList} = \fix f.\lambda n.\\
    &\mathsf{ifZero}\, n\, \mathsf{then}\\
    &\quad \produce\, \mathsf{nil}\\
    &\mathsf{else}\\
    &\quad l \leftarrow \mathsf{unifList}\, (n-1)\\
    &\quad x \leftarrow \mathsf{uniform}\\
    &\quad y \leftarrow \mathsf{uniform}\\
    &\quad \produce (\mathsf{cons}\, (x,y)\, l)\\
    \\
    &\mathsf{sieve} = \lambda l.\,\\
    & p_1 \leftarrow \mathsf{min} (\lamb {(x, y)} {x + y})\, l\\
    & p_2 \leftarrow \mathsf{min} (\lamb {(x, y)} {x - y})\, l\\
    & p_3 \leftarrow \mathsf{min} (\lamb {(x, y)} {-x + y})\, l\\
    & p_4 \leftarrow \mathsf{min} (\lamb {(x, y)} {-x - y})\, l\\
    & \mathsf{filter} \, (\charge 1; \mathsf{iQ} \, p_1\, p_2 \, p_3\, p_4) \, l\\
    \end{align*}
    \end{minipage}%
    \vline
    \begin{minipage}{.33\textwidth}
    \begin{align*}
    &\mathsf{scan} = \fix f \, \lambda l .\,\lambda stk.\,\\
    & \charge 1\\
    &\mathsf{case} \, (l, stk) \, \mathsf{of}\\
    & \mid \mathsf{nil}, \_ \Rightarrow \produce stk\\
    & \mid (x :: xs), \mathsf{nil}\Rightarrow f\, xs \, [x]\\
    & \mid (x :: xs), [p] \Rightarrow f\, xs \, [x, p]\\
    & \mid (x :: xs), (p1 :: p2 :: ps) \Rightarrow \\
    & \quad \mathsf{if} \, \mathsf{clockOrNot} \, p_2 \, p_1\, x \, \mathsf{then}\\
    & \quad \quad f\, (x :: xs) \, (p2 :: ps)\\
    & \quad \mathsf{else}\\
    & \quad \quad f\, xs \, (x :: stk)
    \end{align*}
    \end{minipage}%
    \vline
    \begin{minipage}{0.33\textwidth}
    \begin{align*}
    &\mathsf{graham} = \lambda l.\\
    & p \leftarrow \mathsf{min\_xy}\, l\\
    & l' \leftarrow \mathsf{quicksort}\, (\charge 1; \sqsubseteq\, p) l\\
    & \mathsf{scan}\, l'\, \mathsf{nil}
    \\
    \\
    &\mathsf{convexHull} = \lambda n.\\
    &l \leftarrow \mathsf{unifList}\, n\\
    &l' \leftarrow \mathsf{sieve}\, l\\
    &\mathsf{graham}\, l'
    \end{align*}
    \end{minipage}
    \caption{Stochastic convex hull algorithm}
    \label{fig:stochconvexhull}
\end{figure}

\begin{theorem}(cf. \Cref{app:examples})
    The stochastic convex hull algorithm has linear expected run time. 
\end{theorem}
\begin{comment}
\begin{proof}
    Due to space constraints, we have moved the careful analysis to the appendix.
    Something quite appealing about this semantics is that the analysis of this
    algorithms follows quite closely its "textbook" analysis.
\end{proof}
\end{comment}
\section{Related work}

\paragraph{Type Theories for Cost Analysis} Recent work \cite{niu2022,grodin2023} have developed (in)equational
theories for reasoning about costs of programs inside a modal dependently-typed CBPV metalanguage. Their framework
can reason about monadic effects by using the writer monad transformer, similarly to $\meta$'s cost semantics,
but, due to not being able to having access to fixpoint operators, it can only represent total programs and finitely supported
distributions--- i.e. it cannot represent the geometric distribution. Furthermore, it is
unclear how one would go about extending their framework with continuous distributions.

Other work has focused in designing type theories for doing relational reasoning of programs \cite{cciccek2017, rajani2021}.
Even though these approaches can reason about functional programs as well, they are limited to deterministic programs.

In work by Avanzini et al.~\cite{avanzini2019}, the authors define a graded, substructural type system for reasoning about
expected cost of functional programs, even using a randomized quicksort as an example. One of the main limitations of their system with respect to $\meta$ is that, due to the
substructural invariants of their type system, it can only type check a limited subset of the programs that $\meta$ can.
For instance, it cannot type check common functional idioms like fold and map functions over lists. Furthermore,
they have not addressed how feasible type checking in their system is or if it is even decidable, which in the context of
type-based reasoning is an important property to have.

In other work by Avanzini et al.~\cite{avanzini2021}, the authors describe a continuation passing style (CPS) transformation into
a metalanguage for reasoning about expected cost of programs. Compared to $\meta$, both metalanguages can handle functional programming,
though their language is restricted to a CBV semantics and does not validate useful program equations, such as \Cref{th:expectedcostcenter}.
Furthermore, using continuation-passing style to reason about programs
creates undesired gaps between the transformed and original programs which are avoided when using the expected cost monad's direct-style reasoning.
Recent work by Rajani et al.~\cite{rajani2024modal} uses a Kripke-style semantics for reasoning about expected cost and, as such, has looser
connections to probability theory.

\paragraph{Automatic Resource Analysis} One fruitful research direction has been the automatic amortized resource analysis (AARA) 
\cite{ngo2018, hoffmann2017, hoffmann2015}
which uses a type system to annotate programs with their
cost and automatically infer the cost of the program. These techniques have been extended to reason
about recursive types \cite{grosen2023}, probabilistic programs \cite{wang2020} and programs with local state \cite{lichtman2017}.

Something quite appealing about their approach is that it is completely automatic, whereas our approach requires solving a, possibly
hard, recurrence relation by hand. That being said, their system can only accommodate polynomial bounds, meaning that
they cannot infer the $n\log(n)$ bound for the probabilistic quicksort like we do. Recently, AARA has been extended to 
accommodate exponential bounds \cite{kahn2020exponential} in deterministic programs, though it is still unclear if the same
technique can be extended to the probabilistic setting, meaning that they cannot analyze the
behaviour of exponentially slow programs such as the expected coin tosses one. 

There have been other type-based approach to automatically reasoning about cost of programs, such as the language TiML \cite{wang2017}.
This language allows users to annotate type signatures with cost-bounds and the type checking algorithm will infer and check these bounds.
The main limitation of TiML in comparison to our work is that it cannot handle probabilistic programs. There has also been
work done on automated reasoning about cost for first-order probabilistic programs by Avanzini et al.~\cite{avanzini2020}. The main
limitation of this work when compared to $\meta$ is that it can only handle first-order imperative programs. 
%It is an interesting line of future work understanding to what extent solving recurrence relations can be automated in the context
%of $\meta$.
\paragraph{Recurrence for Expected Cost} There has been some work done in exploring languages for expressing
recurrence relations for expected cost. For example, Sun et al.~\cite{sun2023} provides a language for representing probabilistic
recurrence relations and a tool for analyzing their tail-bounds. The main drawback of these approaches is that the languages 
are not very expressive. In particular they do not have higher-order functions.

Leutgeb el al.~\cite{leutgeb2022automated} define a first-order probabilistic functional language for manipulating data structures and
automatically infer bounds on the expected cost of programs. The main limitation of their approach compared to ours is that
their language is first-order.

Reasoning about expected cost has also been explored for imperative languages. For instance, Batz et al.~\cite{batz2023calculus} develop a weakest pre-condition calculus for reasoning about the expected cost of programs. Again, they can only reason about first-order imperative programs.

\section{Conclusion} 

In this work we have presented $\meta$, a metalanguage for reasoning about expected cost of recursive probabilistic programs.
It extends the existing work of \cite{kavvos2019} to the probabilistic setting. We have proposed two different semantics, one 
based on the writer monad transformer while the second one uses a novel \emph{expected cost} monad. Furthermore,
we have showed that in the absence of unbounded recursion, these two semantics coincide, while when programming with subprobability distributions
we have proved that the expected cost semantics is an upper bound to the cost semantics.

We have justified the versatility of our expected cost semantics by presenting a few case-studies. In particular, the expected cost
semantics obtains, compositionally, the familiar recurrence cost relations for non-trivial programs. In particular, for the randomized 
quicksort algorithm, the semantic recurrence relation recovers the $O(n\log n)$ bound.

%We conjecture that the techniques presented in this paper can be extended in various other tantalizing directions. For instance, it would
%be interesting to generalize the construction of the expected cost monad so that it can reason about other kinds of effects. For instance, 
%when reasoning about non-deterministic programs it is useful to give bounds on the worse/best cost.
% scenarios.

% Going beyond probability, by adopting a more abstract approach on the expected cost monad, it might be possible to reason about other
% kinds of effects. For instance, when reasoning about non-deterministic programs it is useful to give bounds on the worse/best case
% scenarios.

% More generally, we are also interested in laying on firm categorical grounds our logical relations technique. Furthermore, the lack
% of a monad morphism in the subprobabilistic case suggests that there might extensions of this techniques based on $2$-category theory,
% where the monad morphism laws do not hold exactly, only up to a $2$-cell, or inequality in this particular case.

\begin{acks}
We would like to thank Christopher Lam, Lars Birkedal, Alex Kavvos, Philip Saville and Hugo Paquet for giving feadback on the paper; 
and Joe Tassarotti, Chris Barrett, Oliver Richardson, Sam Staton and Matthijs V\'{a}k\'{a}r for helpful technical discussions. This work
was supported by the ERC Consolidator Grant BLAST, and the ARIA programme on Safeguarded AI.
\end{acks}

\bibliographystyle{ACM-Reference-Format}
\bibliography{mybib}

% Thank 
\appendix
\onecolumn
\section{Monadic Semantics of CBPV}
\label{app:cbpvsemantics}

Let $\cat{C}$ be a Cartesian closed category and $T : \cat{C} \to \cat{C}$ a strong monad over it. An alternative
definition of monads is it being a triple $(T, \eta, \mu)$, where $T$ and $\eta$ are natural transformations as before,
but $\mu : T^2 \to T$, the multiplication, replaces the bind natural transformation. The monad laws under this
definition become:
% https://q.uiver.app/#q=WzAsOCxbMiwwLCJUIl0sWzMsMCwiVF4yIl0sWzMsMSwiVCJdLFs0LDAsIlQiXSxbMCwwLCJUXjMiXSxbMSwwLCJUXjIiXSxbMSwxLCJUIl0sWzAsMSwiVF4yIl0sWzAsMiwiMSIsMl0sWzAsMSwiVFxcZXRhIl0sWzEsMiwiXFxtdSJdLFszLDIsIjEiXSxbMywxLCJcXGV0YV9UIiwyXSxbNCw3LCJcXG11X1QiLDJdLFs3LDYsIlxcbXUiLDJdLFs1LDYsIlxcbXUiXSxbNCw1LCJUXFxtdSJdXQ==
\[\begin{tikzcd}
	{T^3} & {T^2} & T & {T^2} & T \\
	{T^2} & T && T
	\arrow["1"', from=1-3, to=2-4]
	\arrow["T\eta", from=1-3, to=1-4]
	\arrow["\mu", from=1-4, to=2-4]
	\arrow["1", from=1-5, to=2-4]
	\arrow["{\eta_T}"', from=1-5, to=1-4]
	\arrow["{\mu_T}"', from=1-1, to=2-1]
	\arrow["\mu"', from=2-1, to=2-2]
	\arrow["\mu", from=1-2, to=2-2]
	\arrow["T\mu", from=1-1, to=1-2]
\end{tikzcd}\]
It is possible to show that these definitions are equivalent: given bind $(-)^\#$, the multiplication can be defined as $\mu = id_{TA}^\#$. 
Conversely, given a multiplication, the bind is defined as $f^\# = T f; \mu$.
This alternative definition is a bit better suited for the original purposes of monads, where it was used as a unifying way of representing
concepts from universal algebra.

This alternative presentation lend itself quite well to the semantics of CBPV-based calculi where, given a monad $T$, computation types denote
$T$-algebras:

\begin{definition}
    A $T$-algebra is a pair $(A, \alpha)$, where $A$ is a $\cat{C}$ object and $\alpha : T A \to A$ is a morphism, such that
    % https://q.uiver.app/#q=WzAsNyxbMCwwLCJBIl0sWzEsMCwiVEEiXSxbMSwxLCJBIl0sWzIsMCwiVF4yQSJdLFszLDAsIlRBIl0sWzMsMSwiQSJdLFsyLDEsIlRBIl0sWzAsMSwiXFxldGFfQSJdLFsxLDIsIlxcYWxwaGEiXSxbMCwyLCJpZF9BIiwyXSxbMyw2LCJUXFxhbHBoYSIsMl0sWzYsNSwiXFxhbHBoYSIsMl0sWzMsNCwiXFxtdV9BIl0sWzQsNSwiXFxhbHBoYSJdXQ==
\[\begin{tikzcd}
	A & TA & {T^2A} & TA \\
	& A & TA & A
	\arrow["{\eta_A}", from=1-1, to=1-2]
	\arrow["\alpha", from=1-2, to=2-2]
	\arrow["{id_A}"', from=1-1, to=2-2]
	\arrow["T\alpha"', from=1-3, to=2-3]
	\arrow["\alpha"', from=2-3, to=2-4]
	\arrow["{\mu_A}", from=1-3, to=1-4]
	\arrow["\alpha", from=1-4, to=2-4]
\end{tikzcd}\]
\end{definition}

Given a $T$-algebra $(A, \alpha)$ we denote by $A_\bullet$ the object of the $T$-algebra.

\begin{example}
    Given an object $A$, the pair $(TA, \mu_A)$ is a $T$-algebra, where the algebra axioms
    follow from the monad laws.
\end{example}

\begin{example}
    Given a $T$-algebra $(A, \alpha)$ and an object $B$, we can equip $B \to A$ with the $T$-algebra structure $\alpha_{B \to A} = \varepsilon_B \Rightarrow (st; T(ev; \alpha))$,
    where $\varepsilon_A : A \to (B \Rightarrow (B \times A))$ is the unit of the Cartesian closed adjunction. This can
    be seen as a \emph{pointwise} algebra structure.
\end{example}

Algebras and their morphisms can be organized as a category, frequently denoted by $\cat{C}^T$. However, for the purposes of CBPV
a different category is used:

\begin{definition}
    The category $\widetilde{\cat{C}}^T$ is the full subcategory of $\cat{C}$ that contain $T$-algebras as objects. This category
    is also called the category of algebras and plain maps.
\end{definition}

The idea is that values are interpreted as objects in $\cat{C}$ while computation types are $T$-algebras. Assuming the only the base
type in the calculus to be $\nat$ and an object $\nat$ in the base category, The interpretation of values and
computations are as follows:
\begin{align*}
    \sem{\nat}^v &= \nat\\
    \sem{U \overline\tau}^v &= \sem{\overline\tau}^c_\bullet\\
    \sem{\tau_1 \times \tau_2}^v &= \sem{\tau_1}^v \times \sem{\tau_2}^v\\
    \\
    \sem{F \tau}^c &= (T \sem{\tau}^v, \mu_{\sem{\tau}^v})\\
    \sem{\tau \to \overline{\tau}}^c &= (\sem{\tau}^v \Rightarrow \sem{\overline \tau}^c, \alpha_{\sem{\tau}^v \to \sem{\overline \tau}^c})
\end{align*}
It is also possible to give semantics to the terms of the language as depicted in \Cref{fig:cbpvsemantics}. The semantics of if-statements
use the fact that $\nat \cong 1 + \nat$, so you can define its semantics by using the universal property of coproducts $[t, (!_\nat; u)]$, where
$!_A : A \to 1$ is the unique arrow into the terminal object. The abstraction and application rule use the adjoint structure $(\Lambda, ev)$,
of Cartesian closed categories, where $\Lambda$ and $ev$ are the unit and counit of the adjunction, respectively. The produce rule uses the
unit of the monad while the bind rule is the sequential composition of a free algebra with a non-free algebra and, therefore, requires applying
the functor $T$ and using the algebra structure of the output --- when the output map is a free algebra, this operation is equal to the bind
of the monad. Thunk and force are basically no-ops in this semantics, while the rules let and unpair are sequential compositions. Pair is
the universal property of products.

\begin{figure*}
    \begin{mathpar}
      \inferrule[var]{~}{\Gamma_1 \times (\tau  \times \Gamma_2) \xto{!\times \pi_1} \tau}
      \and
      \inferrule[if]{\Gamma \xto V \nat \quad \Gamma \xto t \overline{\tau} \quad \Gamma \xto u \overline{\tau}}{\Gamma \xto {\langle id;  V\rangle; [t, (!; u)]} \overline{\tau}}
      \and
      \inferrule[abstraction]{\Gamma \times \tau \xto t \overline{\tau}}{\Gamma \xto {\Lambda_\Gamma; \tau \Rightarrow t} \tau \Rightarrow \overline{\tau}}
      \and
      \inferrule[application]{\Gamma \xto V \tau \quad \Gamma \xto t : \tau \to \overline{\tau}}{\Gamma \xto {\langle t, V\rangle; ev} : \overline{\tau}}
      \\
      \inferrule[produce]{\Gamma \xto V \tau }{\Gamma \xto{V; \eta_\tau} T \tau}

      \and
      \inferrule[sequencing]{\Gamma \xto t T \tau' \\ \Gamma \times \tau' \xto u (\overline \tau, \alpha_{\overline \tau})}{\Gamma \xto{\langle id_\Gamma, t\rangle; st; T u; \alpha_{\overline\tau}} (\overline \tau, \alpha_{\overline \tau})}
      \and
      \inferrule[thunk]{\Gamma \xto t (\overline \tau, \alpha_{\overline \tau})}{\Gamma \xto{t} \overline\tau}
      \and
      \inferrule[force]{\Gamma \xto V \overline \tau}{\Gamma \xto V (\overline \tau, \alpha_{\overline{\tau}})}
      \and
      \inferrule[let]{\Gamma \xto V \tau' \\ \Gamma \times \tau' \xto t \overline \tau}{\Gamma \xto{\langle id_\Gamma, V\rangle; t} \overline \tau}
      \and
      \inferrule[pair]{\Gamma \xto {t_1} \tau_1 \quad \Gamma \xto {t_2} \tau_2}{\Gamma \xto {\langle t_1, t_2\rangle} \tau_1\times \tau_2}
      \and
      \inferrule[unpair]{\Gamma \xto V \tau_1 \times \tau_2 \quad \Gamma \times \tau_1 \times \tau_2 \xto t \overline{\tau}}{\Gamma \xto {\langle id, V\rangle; t} \overline{\tau}}
    \end{mathpar}
    \caption{CBPV monadic semantics}
    \label{fig:cbpvsemantics}
  \end{figure*} 

  This semantics validates the following compositionality properties. Note that we assume, without loss of generality, that the
  free variable $x$ is the first one in the context.

\begin{theorem}
    For every computation $x : \tau', \Gamma \vdashcomp t : \overline{\tau}$ and values $x : \tau', \Gamma \vdashval V : \tau $, $\Gamma \vdashval V' : \tau'$, $\sem{\subst t V x} = \sem{t} \circ \langle\sem{V}, id\rangle$ and $\sem{\subst {V'} V x} = \semcs{V'} \circ \langle\sem{V}, id\rangle$. 
\end{theorem}
\begin{proof}
    The proof follows by mutual structural induction on the typing derivations of $t$ and $V'$.
    \begin{description}
        \item[Variable]: This case follows on case analysis on whether the substituted variable
        is equal to the term or not.
        \item[Pair]: Using the equality $\subst{(V_1, V_2)}{V}{x} = (\subst{V_1}{V}{x}, \subst{V_2} V x)$,
        the induction hypothesis and the universal property of products, we can conclude.
        \item[Unpair]: Using the equality $\subst {\letin{(y_1, y_2)}{V'}{t}} V x = \letin{(y_1, y_2)}{\subst {V'} {V} {x}}{\subst {t} {V} {x}}$,
        the induction hypothesis and the bifunctoriality of the Cartesian product, we can conclude.
        \item[If]: Follows from the universal property of coproducts, the induction hypothesis,
        the naturality of of the diagonal morphism $A \to A \times A$ and the substitution definition.
        \item[Constants]: Follows by unfolding definitions and using the fact that constants
        do not have free variables.
        \item[Thunk]: Using the equality $\subst{\thunk t}{V}{x} = \thunk (\subst t V x)$, the
        definition $\sem{\thunk t}^v = \sem{t}^c$ and the inductive hypothesis we can conclude.
        \item[Force]: Using the equality $\subst{\force V'}{V}{x} = \force (\subst {V'} V x)$, the
        definition $\sem{\force V'}^v = \sem{V'}^c$ and the inductive hypothesis we can conclude.
        \item[Produce]: This case follows from the equality $\subst{(\produce V')} V x = \produce (\subst {V'} V x)$,
        the definition $\sem{\produce V'}^c = \eta \circ \sem{V'}^v$ and the inductive hypothesis.
        \item[Let]: Using the equation $\subst {(\letbe y {V'} t)} V x = \letbe y {(\subst {V'} V x)} {(\subst t V x)}$
        and the induction hypothesis we can write the following equations:
        \begin{align*}
        &\sem{\letbe y {(\subst {V'} V x)} {(\subst t V x)}}=\\
        &\sem{(\subst t V x)} \circ \langle \sem{(\subst {V'} V x)}, id \rangle=\\
        &(\sem{t} \circ \langle \sem{V}, id\rangle) \circ \langle \sem{V'} \circ \langle \sem{V}, id\rangle, id \rangle=\\
        &\sem{t} \circ \langle \sem{V'}, id\rangle \circ \langle \sem{V}, id\rangle\\
        \end{align*}
        Note that the last equation holds up to the symmetric natural isomorphisms $A \times B \cong B \times A$.
        \item[Abstraction]: This case follows by using the equality $\subst {\lamb y t} V x = \lamb y {(\subst t V x)}$,
        the induction hypothesis and the naturality of the Cartesian closed isomorphism $\cat{C}(A \times B, C) \cong \cat{C}(A, B \Rightarrow C)$.
        \item[Application]: This case follows from the equality $\subst {(\app t {V'})} V x = \app {(\subst t V x)} {(\subst {V'} V x)}$,
        the induction hypothesis and the universal property of Cartesian products.
        \item[Sequencing]: This case follows from the equation $\subst {y \leftarrow t; u} V x = y \leftarrow (\subst t V x); (\subst u V x)$,
        the induction hypothesis and the naturality of the monad strength. These are summarized by the commutative diagram below.
% https://q.uiver.app/#q=WzAsOSxbMCwwLCJcXEdhbW1hIl0sWzEsMCwiXFx0YXUnXFx0aW1lc1xcR2FtbWEiXSxbMiwwLCJcXHRhdSdcXHRpbWVzXFxHYW1tYVxcdGltZXNcXEdhbW1hIl0sWzMsMCwiVFxcdGF1XFx0aW1lc1xcR2FtbWEiXSxbNSwwLCJUXFx0YXUgXFx0aW1lcyhcXHRhdScgXFx0aW1lcyBcXEdhbW1hKSJdLFs1LDEsIlQoXFx0YXUgXFx0aW1lcyBcXHRhdSdcXHRpbWVzIFxcR2FtbWEpIl0sWzUsMiwiVFxcb3ZlcmxpbmVcXHRhdSJdLFszLDEsIlQoXFx0YXVcXHRpbWVzXFxHYW1tYSkiXSxbMSwxLCJUXFx0YXUgXFx0aW1lcyBcXEdhbW1hIl0sWzAsMSwiXFxsYW5nbGUgViwgaWRcXHJhbmdsZSJdLFsxLDIsImlkIFxcdGltZXMgXFxEZWx0YSJdLFsyLDMsIlxcc2VtIHQgXFx0aW1lcyBpZCJdLFszLDQsImlkIFxcdGltZXMgXFxsYW5nbGUgViwgaWRcXHJhbmdsZSJdLFszLDcsInN0Il0sWzQsNSwic3QiXSxbNSw2LCJUXFxzZW0gdSJdLFs3LDUsIlQoaWQgXFx0aW1lcyBcXGxhbmdsZSBWLCBpZFxccmFuZ2xlKSJdLFs4LDcsInN0Il0sWzAsOCwiXFxzZW17XFxzdWJzdCB0IHtWfSB4fSIsMix7ImN1cnZlIjoyfV0sWzcsNiwiVChcXHNlbXtcXHN1YnN0IHUge1Z9IHh9KSIsMl1d
\[\begin{tikzcd}
	\Gamma & {\tau'\times\Gamma} & {\tau'\times\Gamma\times\Gamma} & {T\tau\times\Gamma} && {T\tau \times(\tau' \times \Gamma)} \\
	& {T\tau \times \Gamma} && {T(\tau\times\Gamma)} && {T(\tau \times \tau'\times \Gamma)} \\
	&&&&& {T\overline\tau}
	\arrow["{\langle V, id\rangle}", from=1-1, to=1-2]
	\arrow["{\sem{\subst t {V} x}}"', curve={height=12pt}, from=1-1, to=2-2]
	\arrow["{id \times \Delta}", from=1-2, to=1-3]
	\arrow["{\sem t \times id}", from=1-3, to=1-4]
	\arrow["{id \times \langle V, id\rangle}", from=1-4, to=1-6]
	\arrow["st", from=1-4, to=2-4]
	\arrow["st", from=1-6, to=2-6]
	\arrow["st", from=2-2, to=2-4]
	\arrow["{T(id \times \langle V, id\rangle)}", from=2-4, to=2-6]
	\arrow["{T(\sem{\subst u {V} x})}"', from=2-4, to=3-6]
	\arrow["{T\sem u}", from=2-6, to=3-6]
\end{tikzcd}\]
    \end{description}
\end{proof}

\begin{theorem}
    For every computation context $C$, if $\sem{t} = \sem{u}$ then $\sem{C[t]} = \sem{C[u]}$.
\end{theorem}
\begin{proof}
    The proof follows by induction on the context $C$.
\end{proof}

\subsection{Equational presentation of $\meta$}
For the sake of simplicity of the equational theory, we will assume the barycentric operations $\oplus_p$.

In \Cref{fig:fulleqtheory} we present the non-structural equations of $\meta$. The left-hand side is present in every CBPV calculus
with natural numbers and recursion, where the recursion equation is the last one. The right-hand side is split in two blocks: the 
first block are the barycentric algebra equations, the second one are the monoid equations and the last one are the list equations.

\begin{figure}[]
    \centering
    \begin{minipage}{0.4\textwidth}
\begin{align*}
&\ifthenelse 0 t u \equiv t\\
&\ifthenelse {(n + 1)} t u \equiv u\\
&t \equiv \ifthenelse x t t\\
\\
&\app{(\lamb x t)} V \equiv \subst t V x\\
&\letbe x V t \equiv \subst t V x \\
& t \equiv \lamb x {\app t x}\\
& x \leftarrow t; (\lamb y u) \equiv \lamb y {(x \leftarrow t; u)}\\
& \force(\thunk\, t) \equiv t\\
& \thunk(\force\, V) \equiv V\\
& x \leftarrow (\produce V); t \equiv \subst t V x\\
& x \leftarrow t; \produce x \equiv t\\
& \fix x.\, t = \subst t {(\thunk (\fix x.\, t))} x
    \end{align*}
    \end{minipage}%
    \begin{minipage}{0.6\textwidth}
    \begin{align*}
&t \oplus_0 u \equiv t\\
&t \oplus_{p} u \equiv u \oplus_{1-p} t\\
&t \oplus_p t \equiv t\\
&t\oplus_p (u \oplus_q t') \equiv (t \oplus_{\frac{p(1-q)}{1 - pq}} u)  \oplus_{pq} t'\\
\\
&\charge n; \charge m \equiv \charge {(n + m)}\\
&\charge n; \charge m \equiv \charge m; \charge n\\
&\charge 0; t \equiv t\\
\\
&\caseList{\mathsf{nil}}{t}{u} \equiv t\\
&\caseList{(\mathsf{cons}\, V_1\, V_2)}{t}{u} \equiv \subst u {V_1, V_2}{x, xs}\\
&t \equiv \caseList{y}{\subst t {\mathsf{nil}} y}{\subst t {\mathsf{cons }\, x \, xs} y}\\
& \letin{(x_1, x_2)}{(V_1, V_2)}t \equiv \subst t {V_1, V_2} {x_1, x_2}\\
& t \equiv \letin {(x_1, x_2)} V {\subst t {(V_1, V_2)} z}
\end{align*}
    \end{minipage}
    \caption{$\meta$ equational theory}
    \label{fig:fulleqtheory}
\end{figure}

\section{Denotational Soundness Proof}
\label{app:denproof}

In order to prove the fundamental theorem of logical relations, we first show the following lemmas, where the first one
follows by induction:

\begin{lemma}
    For every $\tau$ (resp. $\overline{\tau})$, the relation $\VRel_\tau$ (resp. $\CRel_{\overline\tau}$) is an $\omega$CPO, 
    where the partial order structure is the same as the one from $\semcs{\tau}^v \times \semec{\tau}^v$ 
    (resp. $\semcs{\overline\tau}^c \times \semec{\overline\tau}^c$). Furthermore, the computation relations have
    a least element.
\end{lemma}

\begin{lemma}
    For every type $\tau$ (resp. $\overline \tau$), there is a set $M_\tau$ (resp. $M_{\overline{\tau}}$) and partial order $\leq$ 
    such that the triple $(\VRel_\tau, M_\tau, \leq)$ (resp. $(\CRel_{\overline\tau}, M_{\overline{\tau}}, \leq)$)
    is an $\omega$-quasi Borel space such that the injection function is a morphism in $\wqbs$. 
\end{lemma}
\begin{proof}
    We only make explicit the proof for value types, since the case of computation types is basically the same.
    We define the order $\leq$ to be the same as the one in $\semcs{\tau}^v \times \semec{\tau}^v$ and $M$ to be 
    the \emph{restricted} random elements $\{f \in M_{\tau_1} \smid f(\R) \subseteq \VRel_\tau\}$.
    
    Since by the lemma above the logical relations are $\omega$CPOs, $M$ is closed under suprema of ascending chains.
    and $(\VRel_\tau, M, \leq)$ is an $\omega$-quasi Borel space. The injection into $\semcs{\tau}^v \times \semec{\tau}^v$
    being a morphism follows by construction.
\end{proof}

\begin{lemma}
\label{lem:couplingineq}
    If $(r, \nu) \CRel_{F\tau} \mu$ then for every pair of functions $f_1 : \semec{\tau} \to \R$ and $f_2 : \semcs{\tau} \to \R$, such that for every
    $a_1 \VRel_\tau a_2$, $f_1(a_1) \leq f_2(a_2)$, $\int f_1 \diff \nu \leq \int f_2 \diff \mu_2$, where $\mu_2$ is the second marginal of $\mu$.
\end{lemma}
\begin{proof}
Since by assumption $(r, \nu) \CRel_{F\tau} \mu$, there is a coupling $\gamma$ over the support of $\VRel_\tau$, which allows us to conclude:
\[\int f_1 \diff\nu \leq \int \frac 1 2 (f_1 + f_2) \diff \gamma \leq \int f_2 \diff \mu_2\]
The equalities above hold because, in the support of $\gamma$, $f_1(a_1) \leq f_2(a_2)$, making $f_1 \leq \frac 1 2 (f_1 + f_2) \leq f_2$ and since
$\gamma$ is a joint distribution with marginals $\nu$ and $\mu_2$, we have the (in)equality of integrals above.
\end{proof}

At first, it is reasonable to postulate that the expected cost semantics should coincide with the cost semantics.
Unfortunately, it does not hold in the subprobabilistic case, as alluded to in \Cref{lem:monmor}. 
\begin{example}
 Consider the programs 
 \begin{align*}
 t &= \charge 2; \produce 0\\
 u &= \lamb x {\ifthenelse x {(\bot \oplus (\charge 4; \produce 0))} {\bot}}
 \end{align*}
By unfolding definitions, we can show $E(\semcs{x \leftarrow t; \app u x}^c) = (3, \frac 1 2 \delta_0) \neq (4, \frac 1 2 \delta_0) = \semec{x \leftarrow t; \app u x}^c$.
\end{example}

\begin{comment}
Even though the counterexample above invalidates our soundness theorem, it is still possible to prove a weaker
variant of \Cref{cor:sound}. This is achieved by using essentially the same logical relations as before, with the
exception of $\CRel_{F\tau}$, which now becomes 
\[\CRel_{F \tau} = \set{((r, \nu), \mu)}{\mathbb E(\mu_1) \leq r \land \nu \V^\#_\tau \mu_2}\]
\end{comment}

In the probabilistic case we can prove stronger soundness theorems. Consider the recursion-free fragment of $\meta$
and the alternative logical relation for $F \tau$ types:
\[
\C_{F \tau} = \set{((r, \nu), \mu)}{\mathbb E(\mu_1) = r \land \nu \V^\#_\tau \mu_2}
\]

Since the subprobabilistic and probabilistic soundness proofs are nearly identical, we will only present the
subsprobabilistic one and explicitly mention where they differ.

The following lemma is the most technical aspect of the soundness proof and, intuitively, is saying that the logical relations for
computation types can be equipped with "algebra" structures. Furthermore, since we are proving two similar looking theorems for the 
probabilistic and subprobabilistic cases, and the soundness proof in both cases
is basically the same, we will only present the proof to the subprobabilistic case, and highlight in the proof what would differ for
the probabilistic case.

\begin{lemma}
\label{lemma:convexrel}
    Let $\tau_1$, $\tau_2$ and $\overline{\tau}$ be types and $f_1 : \semcs{\tau_1}^v \times \semcs{\tau_2}^v \to \semcs{\overline{\tau}}^c$, and
    $f_2 : \semec{\tau_1}^v \times \semec{\tau_2}^v \to \semec{\overline{\tau}}^c$ be $\wqbs$ morphisms such that $f_1 \times f_2$, when the
    input is restricted to $\VRel_{\tau_1} \times \VRel_{\tau_2}$, the output is restricts to $\CRel_{\overline\tau}$. It is true
    that $(st; T_1(f_1); \alpha_{\overline{\tau}}) \times (st; T_2(f_2); \alpha_{\overline{\tau}})$, when its input is restricted to
    $\VRel_{\tau_1} \times \CRel_{F\tau_2}$, has its output still be restricted to $\CRel_{\overline\tau}$.
\end{lemma}
\begin{proof}
    This can be proved by induction on the computation type $\overline{\tau}$:
    \begin{description}
        \item[$F\tau$:] In order to prove $f_1^\#(r, \nu) \CRel_{F\tau'} f^\#_2(\mu)$ we have to prove that their expected costs are
        related by the inequality given by the definition of $\CRel_{F\tau}$ 
        and show that there is a coupling over $\nu$ and $\mu_2$, where $\mu_2$ is the second marginal of $\mu$, such that it
        factors through the inclusion $\Psub (\VRel_{\tau'}) \hookrightarrow \Psub (\semcs{\tau'}^v \times \semec{\tau'}^v)$. 
            \newline\newline
        By unfolding the definitions, we get 
        \begin{align*}
        \pi_1(f_1^\#(r, \nu)) &= r + \int (\pi_1 \circ f_1) \diff \nu \\
        \mathbb E(f_2^\#(\mu)) &= \int\int n \norm{f_2(a)} \mu(\diff n, \diff a) + \int n \diff(f_2^\#(\mu)_1)
        \end{align*}
        In the second expression, the left hand side term being summed corresponds to the expected cost of the input
        while the second one corresponds to the cost of the continuation. As such, it is sensible that, in order to
        reason about their difference, we should reason individually about $r - \int\int n \norm{f_2(a)}$ and 
        $\int (\pi_1 \circ f_1) \diff \nu - \int n \diff(f_2^\#(\mu))$, and both should be greater than $0$.
        The first inequality is immediate:
        \begin{align*}
            & \int\int n \norm{f_2(a)} \diff \mu \leq \int\int n \diff \mu \leq r 
        \end{align*}
        For the second expression, assuming $\forall a' \VRel_{\tau'} a, \mathbb{E}(f_2(a)) \leq (\pi_1 \circ f_1)(a')$, we
        can apply \Cref{lem:couplingineq} and use the equality $\int n \diff(f_2^\#(\mu)_1) = \int \mathbb{E}(f_2(a)_1)\diff \mu_2$.
        % \begin{align*}
        %     & \int \mathbb{E}(f_2(a)_1) \gamma(\diff a, \diff a') \leq \int (\pi_1 \circ f_1)(a') \gamma(\diff a, \diff a') = \int (\pi_1 \circ f_1)(a') \nu(\diff a')\\
        %     %
        %     & \int n \diff(f_2^\#(\mu)_1) = \int \mathbb{E}(f_2(a)_1) \mu_2(\diff a) = \int \mathbb{E}(f_2(a)_1) \gamma(\diff a, \diff a')\\
        % \end{align*}
\newline\newline
        By adding these two inequalities we obtain exactly the first condition of the relation $\CRel_{F\tau}$. In the probabilistic
        case every inequality is an equality, since $\norm{f(a)} = 1$ and the inequalities in the definition of $\CRel_{F\tau}$ would 
        be equalities as well.
        The second condition follows from observing that when restricting the domain of $f_1\times f_2$
        to $\VRel_\tau$, we can extract from it a morphism $g : \VRel_\tau \to \Psub(\VRel_{\tau'})$ such that, given inputs $(v_1, v_2)$, 
        the marginals of $g(v_1, v_2)$ are equal to $\pi_2(f_1(v_1))$ and $f_2(v_2)_2$ since, by assumption, $f_1(v_1) \CRel_{F\tau'} f_2(v_2)$.
    \newline\newline
        Given this function, we define the coupling $g^\#(\mu')$, where $\mu'$ is the coupling given by the ``witness'' of $(r, \nu)\CRel_{F\tau} \mu$. Showing
        that it has the right marginals follows from linearity of the marginal function, concluding the proof. This part of the proof remains the 
        same in the probabilistic case
            \newline\newline
        \item[$\tau \to \overline{\tau}$:] This case relies more on notation and, therefore, in order to simplify the presentation, we will
        rely on the symmetry of $f_1$ and $f_2$ and work on the generic expression $st; T(f); \alpha_{\tau \to \overline{\tau}}$ that 
        can be instatiated to both $f_1$ and $f_2$.
            \newline\newline
        By definition of $\CRel_{\tau\to\overline{\tau}}$, in order to define a morphism
        $\VRel_{\tau_1} \times \VRel_{\tau_2} \to \CRel_{\tau \to \overline{\tau}}$, it suffices to defines its transpose 
        $\VRel_{\tau} \times (\VRel_{\tau_1} \times \VRel_{\tau_2}) \to \CRel_{\overline{\tau}}$.
        Since the algebra structure of $\alpha_{\tau \to \overline{\tau}}$ is defined as $\eta; id_\tau \Rightarrow (st; T(ev); \alpha_{\overline{\tau}})$,
        we want to show that the the map $id_\tau \times (st; T f; \eta; id_\tau \Rightarrow (st; T(ev); \alpha_{\overline{\tau}})); ev$, i.e.
        can be rewritten in the format $st; T(f'); \alpha_{\overline\tau}$, so that we can apply the induction hypothesis. This equation
        holds, up to isomorphism, by the following commutative diagram:
% https://q.uiver.app/#q=WzAsMTEsWzAsMCwiXFx0YXUgXFx0aW1lcyAoXFx0YXVfMSBcXHRpbWVzIFRcXHRhdV8yKSAiXSxbMCwyLCIoXFx0YXUgXFx0aW1lcyBcXHRhdV8xKSBcXHRpbWVzIFRcXHRhdV8yIl0sWzEsMCwiXFx0YXUgXFx0aW1lcyBUKFxcdGF1XzEgXFx0aW1lcyBcXHRhdV8yKSAiXSxbMiwwLCJcXHRhdSBcXHRpbWVzIFQoXFx0YXUgXFxSaWdodGFycm93XFxvdmVybGluZXtcXHRhdX0pIl0sWzQsMCwiXFx0YXUgXFx0aW1lcyAoXFx0YXUgXFxSaWdodGFycm93KFxcdGF1IFxcdGltZXMgVChcXHRhdSBcXFJpZ2h0YXJyb3dcXG92ZXJsaW5le1xcdGF1fSkpKSJdLFs0LDEsIlxcdGF1IFxcdGltZXMoXFx0YXUgXFxSaWdodGFycm93IFxcb3ZlcmxpbmV7XFx0YXV9KSJdLFs0LDIsIlxcb3ZlcmxpbmVcXHRhdSJdLFszLDIsIlRcXG92ZXJsaW5lXFx0YXUiXSxbMiwyLCJUKFxcdGF1IFxcdGltZXMgKFxcdGF1IFxcUmlnaHRhcnJvd1xcb3ZlcmxpbmVcXHRhdSkiXSxbMSwyLCJUKFxcdGF1IFxcdGltZXMgKFxcdGF1XzEgXFx0aW1lcyBcXHRhdV8yKSkiXSxbMiwxLCJcXHRhdVxcdGltZXMgVChcXHRhdVxcUmlnaHRhcnJvd1xcb3ZlcmxpbmVcXHRhdSkiXSxbMCwyLCJpZF9cXHRhdVxcdGltZXMgc3QiXSxbMSw5LCJzdDsgVChhXnstMX0pIiwyXSxbOSw4LCJUKFxcdGF1IFxcdGltZXMgZikiLDJdLFs4LDcsIlQgZXYiLDJdLFs3LDYsIlxcYWxwaGFfe1xcb3ZlcmxpbmVcXHRhdX0iLDJdLFswLDEsImEiLDJdLFsyLDksInN0IiwyXSxbMiwzLCJcXHRhdVxcdGltZXMgVGYiXSxbNSw2LCJldiJdLFs0LDUsImlkX1xcdGF1IFxcdGltZXMgKGlkX1xcdGF1IFxcUmlnaHRhcnJvdyAoc3Q7IFRldiA7IFxcYWxwaGFfe1xcb3ZlcmxpbmVcXHRhdX0pKSIsMV0sWzMsNCwiXFx0YXVcXHRpbWVzXFxldGEiXSxbNCwxMCwiZXYiXSxbMTAsOCwic3QiXSxbMywxMCwiIiwyLHsibGV2ZWwiOjIsInN0eWxlIjp7ImhlYWQiOnsibmFtZSI6Im5vbmUifX19XSxbMCwyXV0=
\[\begin{tikzcd}
	{\tau \times (\tau_1 \times T\tau_2) } & {\tau \times T(\tau_1 \times \tau_2) } & {\tau \times T(\tau \Rightarrow\overline{\tau})} && {\tau \times (\tau \Rightarrow(\tau \times T(\tau \Rightarrow\overline{\tau})))} \\
	&& {\tau\times T(\tau\Rightarrow\overline\tau)} && {\tau \times(\tau \Rightarrow \overline{\tau})} \\
	{(\tau \times \tau_1) \times T\tau_2} & {T(\tau \times (\tau_1 \times \tau_2))} & {T(\tau \times (\tau \Rightarrow\overline\tau)} & T\overline\tau & \overline\tau
	\arrow["{id_\tau\times st}", from=1-1, to=1-2]
	\arrow[from=1-1, to=1-2]
	\arrow["a"', from=1-1, to=3-1]
	\arrow["{\tau\times Tf}", from=1-2, to=1-3]
	\arrow["st"', from=1-2, to=3-2]
	\arrow["\tau\times\eta", from=1-3, to=1-5]
	\arrow[Rightarrow, no head, from=1-3, to=2-3]
	\arrow["ev", from=1-5, to=2-3]
	\arrow["{id_\tau \times (id_\tau \Rightarrow (st; Tev ; \alpha_{\overline\tau}))}"{description}, from=1-5, to=2-5]
	\arrow["st", from=2-3, to=3-3]
	\arrow["ev", from=2-5, to=3-5]
	\arrow["{st; T(a^{-1})}"', from=3-1, to=3-2]
	\arrow["{T(\tau \times f)}"', from=3-2, to=3-3]
	\arrow["{T ev}"', from=3-3, to=3-4]
	\arrow["{\alpha_{\overline\tau}}"', from=3-4, to=3-5]
\end{tikzcd}\]
From left to right, the first diagram commutes by definition of strong monad, the second commutes from naturality of the
strength of $T$, the triangular diagram commutes by the Cartesian closed adjunction and the final diagram commutes by
naturality of $ev$.\qedhere
\end{description}
\end{proof}

We are interested in the case where the type $\tau_1$ will be a context $\Gamma$. We now state the denotational soundness theorem:

\begin{theorem}
    For every $\Gamma = x_1 : \tau_1, \dots, x_n : \tau_n$, $\Gamma \vdashval V : \tau$, $\Gamma \vdashcomp t : \overline\tau$ and if for every $1 \leq i \leq n$, $\cdot \vdashval V_i : \tau_i$ and $\semcs{V_i} \VRel_{\tau_i} \semec{V_i}$, then 
    \begin{align*}
    &\semcs{\letin {\overline{x_i}} {\overline{V_i}} t}^c \CRel_{\overline \tau}^c \semec{\letin {\overline{x_i}} {\overline{V_i}} t}^c \text{ and }\\   
    &\semcs{\letin {\overline{x_i}} {\overline{V_i}} V}^v \VRel_\tau \semec{\letin {\overline{x_i}} {\overline{V_i}} V}^v,
    \end{align*}

    where the notation $\overline{x_i} = \overline{V_i}$ means a list of $n$ let-bindings or, in the case of values, 
    a list of substitutions.
\end{theorem}
\begin{proof}
    The proof follows from mutual induction on $\Gamma \vdash^v V : \tau$ and $\Gamma \vdash^c t : \overline{\tau}$. Many
    of the cases follow by just applying the induction hypothesis or by assumptions in the theorem statement. We go over 
    the most interesting cases:
    \begin{description}
        \item[Comp] This case follows from \Cref{lemma:convexrel}.
        \item[Fix] This theorem follows from the induction hypothesis and from the fact that the relations $\CRel^c_{\overline{\tau}}$
        are closed under suprema of ascending chains.
        \item[Produce]
        First apply the induction hypothesis to $V$ and assume that $\semcs{V} = v_1$ and $\semec{V} = v_2$. 
        By construction, $\eta^{T_1}(v_1) \CRel^v_{F\tau}\eta^{T_2}(v_2)$, since they both have the same expected
        value and the coupling is $\delta_{(v_1, v_2)}$.
        \item[Case]
        By applying the inductive hypothesis to $V$ we may do case analysis on it and if it is the empty list,
        we use the inductive hypothesis on $t$ and, otherwise, we use the inductive hypothesis on $u$. \qedhere
    \end{description}
\end{proof}

We now have a very precise sense in which the expected-cost semantics is related to the cost semantics:

\begin{corollary}
\label{cor:sound1}
    The expected-cost semantics is sound with respect to the cost semantics, i.e. for every program
    $\cdot \vdash_c t : F\tau$, the expected cost of the second marginal of $\semcs{t}^c$ greater than
    $\pi_1(\semec{t}^c)$.
\end{corollary}

In the recursion-free case, the definition of $\CRel_{F\tau}$ gives us a stronger soundness property.

\begin{corollary}
\label{cor:sound3}
    The recursion-free expected-cost semantics is sound with respect to the cost semantics, i.e. for every program
    $\cdot \vdash_c t : F\tau$, the expected cost of the second marginal of $\semcs{t}^c$ is equal to
    $\pi_1(\semec{t}^c)$.
\end{corollary}

\section{Operational Soundness Proof}
\label{app:opsound}
We now prove the soundness theorem.
\begin{proof}
    Proof by induction on $n$ and on $t$. The base case $n = 0$ is trivial because $\bot$ is the least element and
    such an element is preserved by the algebra structure. The inductive ones follow basically from the inductive 
    hypothesis and the CBPV equational theory.
    \begin{description}
        \item[Terminal]: The evaluation rules for terminal computations $T$ output the unit of the monad, which allows us to
        conclude $\sem{\DownarrowFun_n (T)} = (x \leftarrow (\delta_{(0, T)}); \semcs{x}) =  \semcs{T}$.
        \item[Charge]: $\sem{\DownarrowFun_n (\charge r)} = \delta_{(r, ())} = \semcs{\charge r}$.
        \item[Sampling]: $\sem{\DownarrowFun_n (\uniform)} = \delta_0\otimes \lambda = \semcs{\uniform}$.
        \item[Seq]:
        % \begin{align*}
        % &\sem{\Downarrow (x \leftarrow t; u)} = \\
        % &(\alpha \circ (v \mapsto \sem{\Downarrow(\subst u v x)})^\#)(\alpha(\sem{\Downarrow(t))}) \leq\\
        % &(\alpha \circ (v \mapsto \sem{\Downarrow(\subst u v x)})^\#)(\semcs{t}) \leq\\
        % &\semcs{x \leftarrow t; u}
        % \end{align*}
        \begin{align*}
        &\sem{\DownarrowFun_n (x \leftarrow t; u)} = \\
        & (\produce V) \leftarrow \DownarrowFun_n(t); y \leftarrow \DownarrowFun_{n-1}(\subst u V x); \semcs{y} \leq\\
        & (\produce V) \leftarrow \DownarrowFun_n(t); \semcs{\subst u V x} =\\
        & (\produce V) \leftarrow \DownarrowFun_n(t); x \leftarrow \semcs{\produce V}; \semcs{u} =\\
        & x \leftarrow (y \leftarrow \DownarrowFun_n(t); \semcs{y}); \semcs{u} \leq\\
        & x \leftarrow \semcs{t}; \semcs{u} = \semcs{x \leftarrow t; u}\\
        \end{align*}
        \item[App]: 
        \begin{align*}
        &\sem{\DownarrowFun_n (t\, V)} = (\lamb x t') \leftarrow \DownarrowFun_n(t); y \leftarrow \DownarrowFun_{n-1}(\subst {t'} V x); \semcs{y}\leq\\
        & (\lamb x t') \leftarrow \DownarrowFun_n(t); \semcs{(\lamb x t')\, V} = f \leftarrow \DownarrowFun_n(t); \semcs{f}(\semcs{V}) = \\
        & (f \leftarrow \DownarrowFun_n(t); \semcs{f})(\semcs{V}) \leq \semcs{t}(\semcs{V}) = \semcs{t\, V}
        \end{align*}
        \item[Fix]: By unfolding the operational semantics, 
        \begin{align*}
        &\sem{\DownarrowFun_n (\fix x.\, t)} =\\
        &\sem{\DownarrowFun_{n-1}(\subst t {\thunk \fix x.\, t} x )} \leq\\
        &\semcs{\subst t {\thunk \fix x.\, t} x} =\\
        &\semcs{\fix x.\, t}
        \end{align*}
        \item[IzZ0]: By unfolding the operational semantics, 
        \begin{align*}
        & \sem{\DownarrowFun_n (\ifthenelse{0}{t}{u})} =\\
        & \sem{\DownarrowFun_n (t)} \leq \semcs{t} =\\
        & \semcs{\ifthenelse{0}{t}{u}}
        \end{align*}
        \item[IzZS]: By unfolding the operational semantics, 
        \begin{align*}
        & \sem{\DownarrowFun_n (\ifthenelse{n+1}{t}{u})} =\\
        & \sem{\DownarrowFun_n (u)} \leq \semcs{u} =\\
        & \semcs{\ifthenelse{n+1}{t}{u}}
        \end{align*}
        \item[UnPair]: 
        \begin{align*}
        & \sem{\DownarrowFun_n (\letin{(x_1, x_2)}{(V_1, V_2)}{t})} =\\
        & \sem{\DownarrowFun_{n-1} (\subst t {V_1, V_2}{x_1, x_2})} \leq\\
        & \semcs{t}(\semcs{V_1}, \semcs{V_2}) =\\
        & \semcs{\letin{(x_1, x_2)}{(V_1, V_2)}{t})}
        \end{align*}
        \item[caseNil]: 
        \begin{align*}
        & \sem{\DownarrowFun_n (\caseList{\mathsf{nil}}{t}{u})} =\\
        & \sem{\DownarrowFun_n (t)} \leq \semcs t =\\
        & \semcs{\caseList{\mathsf{nil}}{t}{u})}
        \end{align*}
        
        \item[caseCons]: 
        \begin{align*}
        & \sem{\DownarrowFun_n (\caseList{\mathsf{cons} \, V_1\, V_2}{t}{u})} =\\
        & \sem{\DownarrowFun_{n-1} (\subst u {V_1, V_2} {x_1, x_2})} \leq\\
        & \semcs {\subst u {V_1, V_2} {x_1, x_2})} =\\
        & \semcs{\caseList{\mathsf{cons} \, V_1\, V_2}{t}{u}}
        \end{align*}
        \qedhere
    \end{description}
\end{proof}

\section{Operational Adequacy Proof}
\label{app:opproof}
As it is usually the case with adequacy proofs, it follows by a logical relations
argument. Before defining it, we define a relation lifting for the subprobability
cost monad.
Next, we define an extension of the logical relations $\vartriangleright_\tau$ to contexts.

\begin{definition}
    Let $\Gamma = x_1 : \tau_1, \dots, x_n :\tau_n$, $\cdot \vdashval V_i : \tau_1$ and
    $\gamma : \semcs{\Gamma}^v$. We say that $(V_1, \dots, V_n) \vartriangleright_\Gamma \gamma$
    if, and only if, $V_i \vartriangleright_{\tau_i} \pi_i(\gamma)$, for every $i \in \{1,\dots, n\}$.
    We will use the letter $G$ to denote the list of values of $(V_1, \dots, V_n)$.
\end{definition}

In order to prove the fundamental theorem of logical relations we require a couple of lemmas that are proved
by induction.

\begin{lemma}
\label{lem:adequacyLQconvex}
    Let $T$ be the subprobability cost monad $\Psub(\nat \times -)$, $\Gamma$ be a context, $\tau_1$ and $\overline{\tau}$ types, $\Gamma, x : \tau_1 \vdash t : \overline\tau$ a computation and $f : \semcs{\Gamma}^v \times \semcs{\tau_1}^v \to \semcs{\overline{\tau}}^c$ a $\wqbs$ morphism such that for every $\cdot \vdashval V : \tau_1$ with $V \vartriangleright_{\tau_1} v$,
    $\DownarrowFun(\subst t {G, V} {\Gamma, x}) \LHD_{\overline{\tau}} f(\gamma, v)$, then $\DownarrowFun(x \leftarrow u; \subst t {G} {\Gamma}) \LHD_{\overline{\tau}} (\alpha_{\overline{\tau}} \circ Tf(\gamma)) (\nu)$, whenever $\DownarrowFun(u) \LHD_{F\tau_1} \nu $ and $G \vartriangleright_{\Gamma} \gamma$.
\end{lemma}
\begin{proof} 
    Fix $G \vartriangleright_\Gamma \gamma$. The proof follows by induction on the computation type $\overline{\tau}$.
    \begin{description}
        \item[$F\tau$]: This proof follows by unfolding the definitions.
        Given the definition of the relational lifting for the cost monad, let 
        $h : \nat \times T^{\cdot \vdashcomp F\tau} \to [0, 1]$ and $g : \nat \times \semcs{\tau} \to [0,1]$ be functions such that
        for every $V'' \vartriangleright_\tau v''$ and $n : \nat$, $g(n, v'') \leq h(n, v'')$. Therefore, we have to show that
        \[
        \int g(n + n', y) \diff((n', v) \leftarrow \nu; f(\gamma, v)) \leq \int f(n + n', y) \diff((n', V') \leftarrow \DownarrowFun(u); \DownarrowFun(\subst t {G, V'} {\Gamma, x}))
        \]
        Using the equational theory of CBPV and the commutativity equation, the expression above is equivalent to 
        \[
        (n', v) \leftarrow \nu; \int g(n+n',y) \diff(f(\gamma, v)) \leq (n', V') \leftarrow \DownarrowFun(u); \int f(n+n',y) \diff(\DownarrowFun(\subst t {G, V'} {\Gamma, x}))
        \]
        By assumption, for every $V \vartriangleright_{\tau_1} v$, $\int g \diff(f(v)) \leq \int f \diff(\DownarrowFun(\subst t {V} x))$.
        We conclude this case by using the assumption $\DownarrowFun(u) \LHD_{F \tau_1} \nu$ and the functions 
        $g'(n', v) = \int g(n + n', y) \diff(f(\gamma, v))$ and $h'(n', V) = \int h(n + n', y) \diff(\DownarrowFun(\subst t {G, V} {\Gamma, x}))$. Since, by construction, 
        $g'$ and $h'$ satisfy the property that whenever $V_1 \vartriangleright_{\tau_1} v_1$, $h'(n', v_1) \leq g'(n', V_1)$, for every 
        $n' : \nat$, we can show:
        \begin{align*}
        & (n', v) \leftarrow \nu; \int g(n + n', y) \diff(f(\gamma, v)) =\\    
        & \int g' \diff(\nu)\leq \\
        & \int h' \diff (\DownarrowFun(u)) =\\
        & (n', V') \leftarrow \DownarrowFun(u); \int f(n + n', y) \diff(\DownarrowFun(\subst t {G, V'} {\Gamma, x}))
        \end{align*}
        \item[$\tau \to \overline{\tau}$]: For this case, let $V' \vartriangleright_\tau v'$. We have to show that $((\lamb y {t'}) \leftarrow \DownarrowFun(x \leftarrow u; t); \DownarrowFun(\subst {t'} {G, V'} {\Gamma, y}) \LHD_{\overline{\tau}} x \leftarrow \nu; f(\gamma, x, v')$. Rewriting it, we obtain the following
        equivalent relation $V \leftarrow \DownarrowFun(u); (\lamb y {t'}) \leftarrow \DownarrowFun(t); \DownarrowFun(\subst {t'} {G, V'} {\Gamma, y}) \LHD_{\overline{\tau}} x \leftarrow \nu; f(\gamma, x, v')$. We prove this by using the induction hypothesis for the type $\overline{\tau}$. We choose $\Gamma, x : \tau_1 \vdash t \, V'$ and $v \mapsto f(\gamma, v, v')$ in the inductive hypothesis, which allows us to conclude.
    \end{description}
\end{proof}

\begin{comment}
\begin{lemma}
    If $t \LHD f$ (resp. $V \vartriangleright v$) then $\semcs{t} \leq f$ (resp. $\semcs{V} \leq v$).
\end{lemma}
\end{comment}

\begin{lemma}
\label{lem:bot}
    For every computation type $\overline \tau$ and for every distribution $\mu$, 
    $\mu \LHD_{\overline{\tau}} \bot$ and if $\mu \LHD_{\overline{\tau}} x_n$,
    for an ascending chain $x_0 \leq \cdots \leq x_n \leq \cdots$, then
    $\mu \LHD_{\overline{\tau}} \sup_n (x_n)$.
\end{lemma}
\begin{proof}
    The proof follows by induction on $\overline{\tau}$. For the base case,
    let $f : T^{\cdot \vdashcomp F \tau} \to [0,1]$ and $g : \semcs{\tau} \to [0,1]$ be
    functions such that whenever $V \vartriangleright_\tau x$, $g(n, x) \leq f(V)$.
    Since $\bot$ in $F\tau$ is the $0$, measure, $\int g \diff 0 = 0 \leq \int f\diff(\DownarrowFun(t))$.
    The stability under suprema of ascending chains follows from Scott-continuity
    of integration.

    For the case $\tau \to \overline \tau$, we use the inductive hypothesis and the
    fact that the order of functions is given pointwise.
\end{proof}

\begin{theorem}[Fundamental Theorem of Logical Relations]
    If $\Gamma \vdashcomp t : \overline\tau$ (resp. $\Gamma \vdashval V : \tau$) and $G \vartriangleright_\Gamma \gamma$ then $\DownarrowFun(\subst t {G} {\Gamma}) \LHD_{\overline{\tau}} \semcs{t}(\gamma)$ (resp. $\subst V {G} {\Gamma} \vartriangleright_\tau \semcs{V}(\gamma)$).
\end{theorem}
\begin{proof}
    The proof follows by mutual induction on the typing derivations of $t$ and $V$.
        \begin{description}
        \item[Var]: Assuming that $\subst {x_i} {G} {\Gamma} = V_i$, where we assume that the $i$-th elements
        of $G$ and $\Gamma$ are, respectively, $V_i$ and $x_i$, which implies the conclusion by the assumption $V_i \vartriangleright v_i$.
        \item[Arithmetic constants]: Follows by inspection.
        \item[Pair]: Follows directly from the induction hypotheses.
        \item[Charge]: By unfolding the definitions, $\DownarrowFun (\charge {\subst V G \Gamma}) = \delta_{(\sem{\subst V G \Gamma}, ())} = \semcs{\charge {\subst V G \Gamma}}$. Therefore,
        since $\vartriangleright_1 = \{((), ())\}$ and the distributions $\DownarrowFun (\charge {\subst V G \Gamma})$ and $\semcs{\charge {\subst V G \Gamma}}$ are the same, we can conclude by definition of $\LHD_{F\tau}$.
        \item[Sample]: By unfolding the definitions, $\DownarrowFun (\uniform) = \delta_0 \otimes \lambda = \semcs{\uniform}$. Therefore,
        since $\vartriangleright_\R = \{(r, r) \mid r \in \R\}$ and the distributions $\DownarrowFun (\uniform)$ and $\semcs{\uniform}$ 
        are the same, we can conclude by definition of $\LHD_{F\tau}$
        \item[Fix]: Follows mostly from \Cref{lem:bot}. We begin by proving that $\DownarrowFun(\fix x. (\subst t G \Gamma)) \LHD_{\overline{\tau}} (\semcs{t}(\gamma))^n(\bot)$, for every $n : \nat$. The proof follows by induction on $n$ and, in order to avoid visual pollution, let $t' = \subst t G \Gamma$.
        \begin{description}
            \item[$0$]: In this case, $(\semcs{t}(\gamma))^0(\bot) = \bot$, so we can apply \Cref{lem:bot}.
            \item[$n + 1$]: For this case, we apply the induction hypothesis and get 
            \[\DownarrowFun(\fix x. t') \LHD_{\overline{\tau}} (\semcs{t}(\gamma))^n(\bot)\] 
            Next, using the definition of $\vartriangleright_{U\overline{\tau}}$, we can conclude that
            \[\thunk (\fix x. t') \vartriangleright_{U\overline{\tau}} (\semcs{t}(\gamma))^n(\bot)\]
            Now, we apply the global induction hypothesis and conclude 
            \[\DownarrowFun(\fix x. t') = \DownarrowFun(\subst {t'} {\thunk (\fix x. t')} x) \LHD_{\overline{\tau}} (\semcs{t}(\gamma))^{n+1}(\bot)\]
        \end{description}
        Therefore, by \Cref{lem:bot} $\DownarrowFun(\fix x. t') \LHD_{\overline\tau} \bigsqcup_n (\semcs{t}(\gamma))^n(\bot)$.
        \item[Abstraction]: Follows directly from the equation $\subst {(\lamb x t)} G \Gamma = \lamb x {(\subst t G \Gamma)}$,
        the equation $\DownarrowFun(\lamb x {\subst t G \Gamma}) = \delta_{(0, \lamb x t)}$ and the induction hypothesis. We will now
        show show that $\DownarrowFun(\lamb x {\subst t G \Gamma}) \LHD_{\tau \to \overline \tau} \semcs{t}(\gamma)$. Let $V \vartriangleright_\tau v$,
        we have to show 
        \[((\lamb x {t'}) \leftarrow \DownarrowFun(\lamb x {\subst t G \Gamma}); \DownarrowFun(\subst {t'} V x)) \LHD_{\overline{\tau}} \semcs{t}(\gamma, v)\]
        The LHS of that expression is equal to $\DownarrowFun(\subst t {G, V} {\Gamma, x})$. We conclude by applying the induction hypothesis
        to $\Gamma, x : \tau \vdashcomp t : \overline{\tau}$ and $V \vartriangleright_\tau v$.
        \item[Application]: Follows directly from equation $\subst {(\app t V)} G \Gamma = \app {(\subst t G \Gamma)} {\subst V G \Gamma}$,
        the induction hypothesis and the definition of $\LHD_{\tau \to \overline \tau}$.
        \item[Produce]: We conclude by the induction hypothesis for a value $V$, the equalities
        $\DownarrowFun(\subst V {G} {\Gamma}) = \delta_{(0, \subst V {G} {\Gamma})}$ and $\semcs{\produce \subst V {G} {\Gamma}} = \delta_{(0, \semcs{V}( {\semcs{G}}))}$, and the fact that for every measurable function $f : A \to [0, 1]$, $\int f \diff(\delta_x) = f(x)$.
        \item[Force]: Follows from the fact that the only well-typed closed programs of type $U \overline{\tau}$ are those of the form
        $\thunk t$, for some computation $t$, the equation $\DownarrowFun(\force \thunk t) = \DownarrowFun(t)$ and the induction hypothesis.
        \item[Thunk]: Follows directly from the substitution equality $\subst {(\thunk t)} G \Gamma = \thunk (\subst t G \Gamma)$, the 
        induction hypothesis and the definition of $\vartriangleright_{U\overline\tau}$.
        \item[List Case]: Using the distributivity of substitution, the fact that closed values of type list are either the
        empty list or the cons of a list, and the induction hypothesis, we can conclude.
        \item[Seq]: First, use the distributivity of substitution $\subst {(x \leftarrow t; u)} {G} {\Gamma} = x \leftarrow \subst t {G} {\Gamma}; \subst u {G}{\gamma}$. We conclude by applying \Cref{lem:adequacyLQconvex} and noting that its assumptions are exactly the induction hypotheses
        applied to $\subst t {G} {\Gamma}$ and $\subst u {G} {\Gamma}$.\qedhere
    \end{description}
\end{proof}

\section{Examples}
\label{app:examples}

\subsection{Random Walks}
For this example we are interested in
the symmetric random walk over the natural numbers. At every point $n$ the probability of moving to
$n-1$ or $n+1$ is $\frac{1}{2}$. Furthermore, we are assuming the variant where at $0$
you move to $1$ with probability $1$. We can write a program that simulates such a random walk 
with a point of departure $i : \nat$ and a point of arrival $j : \nat$:
\begin{align*}
 &\mathsf{randomWalk} = \mu f : \nat \to \nat \to F 1.\, \lambda i : \nat \, j : \nat.\\
 &\mathsf{if} \, i = j \, \mathsf{then}\\
 &  \quad \produce ()\\
 & \mathsf{else}\\
 &\quad \charge 1;\\
 &\quad\mathsf{if} \, i \, \mathsf{then}\\
 & \quad \quad (\force f) \, 1 \, j\\
 &\quad \mathsf{else}\\
 &\quad \quad ((\force f) \, (i-1) \, j) \oplus ((\force f) \, (i+1) \, j)
\end{align*}
The program receives the starting and end points, $i$ and $j$, respectively, as arguments,
and if they are equal, you stop the random walk. Otherwise, you take one step 
of the random walk, i.e. you take step to either $i - 1$ or $i + 1$ with equal probability,
with the  exception of when $i = 0$, in which case you go to $1$. This iterative behaviour
can be straightforwardly captured with recursion, as illustrated by the program above. This
is basically the mathematical specification of the one-dimensional random walk with barrier.
As such, we know that it terminates with probability $1$, cf Section 1.6 of \cite{norris1998}).

By unfolding the denotational semantics, we obtain that the cost expression is given by the
fixed point of the operator 
\begin{align*}
&\lambda F : \nat \to \nat \to \weight.\, \lamb {i : \nat j : \nat} ~\\
&\mathsf{if}\, i = j \,\mathsf{then}\\
& \quad 0\\
& \mathsf{elseif} \, i = 0\, \mathsf{then}\\
&\quad 1 + F(1, j)\\
&\mathsf{else}\\
&\quad 1 + \frac 1 2 (F(i-1, j) + F(i + 1, j))
\end{align*}
It is now possible to compute the expected value on the number of rounds that are necessary in order
to reach your target, which following the expression above, is given by the following two-argument 
recursive relation.
\begin{align*}
    T(i, i) &= 0\\
    T(0, j) &= 1 + T(1, j)\\
    T(i, j) &= 1 + \frac{1}{2} (T (i-1, j) + T(i + 1, j))
\end{align*}
This recurrence relation is well-known in the theory of Markov chains --- see \cite{norris1998} for an introduction.
Something interesting about it is that when $i > j$, this stochastic process reduces to the symmetric random walk
without an absorbing state, which is known to have $\infty$ expected cost.

\subsection{Stochastic Convex Hull}

In computational geometry, finding the convex hull of a set
of points in space is an important algorithm that has been
thoroughly studied. In this case study, we go over a variant
of this algorithm that has a linear expected runtime cost.

This algorithm can be divided into three separate components.
The first one generates a uniformly and independently sampled
list in the square $[0,1]^2$. The second one sieves the original 
list of points so that points that are ``obviously'' not a part 
of the convex hull are eliminated. The last part is any convex 
hull algorithm that runs in $O(n\log n)$. It is important to note 
that this algorithm only has this time complexity under the assumption
that the input list is uniformly distributed. Indeed, the
time complexity of this algorithm hinges on the following lemma.

\begin{lemma}(Th.~2.2 of Golin \cite{golin1988analysis})
\label{lem:filterconvexhull}
    Let $P \subseteq [0,1]^2$ be an independent and uniformly distributed 
    finite set of $n$ points. There is a square $R$ which is inside $P$'s 
    convex hull such that the size of the set of points in $P$ outside of 
    $R$ is $O(\sqrt n)$.
\end{lemma}

We can also prove that $\mathsf{convexHull}$ terminates with probability $1$.
\begin{lemma}
    The total mass of $\pi_2(\semec{\mathsf{convexHull}})$ is $1$.
\end{lemma}
\begin{proof}
    The only component that might be problematic is the $\mathsf{scan}$ function,
    which terminates since it is structurally recursive on the lexicographic order
    on lists.
\end{proof}

\begin{figure}[]
    \centering
    \begin{minipage}{.33\textwidth}
    \begin{align*}
    &\mathsf{unifList} = \fix f.\lambda n.\\
    &\mathsf{ifZero}\, n\, \mathsf{then}\\
    &\quad \produce\, \mathsf{nil}\\
    &\mathsf{else}\\
    &\quad l \leftarrow \mathsf{unifList}\, (n-1)\\
    &\quad x \leftarrow \mathsf{uniform}\\
    &\quad y \leftarrow \mathsf{uniform}\\
    &\quad \produce (\mathsf{cons}\, (x,y)\, l)\\
    \\
    &\mathsf{sieve} = \lambda l : \nat.\,\\
    & p_1 \leftarrow \mathsf{min} (\lamb {(x, y)} {x + y})\, l\\
    & p_2 \leftarrow \mathsf{min} (\lamb {(x, y)} {x - y})\, l\\
    & p_3 \leftarrow \mathsf{min} (\lamb {(x, y)} {-x + y})\, l\\
    & p_4 \leftarrow \mathsf{min} (\lamb {(x, y)} {-x - y})\, l\\
    & \mathsf{filter} \, (\charge 1; \mathsf{iQ} \, p_1\, p_2 \, p_3\, p_4) \, l\\
    \end{align*}
    \end{minipage}%
    \vline
    \begin{minipage}{.33\textwidth}
    \begin{align*}
    &\mathsf{scan} = \fix f. \, \lambda l .\,\lambda stk.\,\\
    & \charge 1\\
    &\mathsf{case} \, (l, stk) \, \mathsf{of}\\
    & \mid \mathsf{nil}, \_ \Rightarrow \produce\, stk\\
    & \mid (x :: xs), \mathsf{nil}\Rightarrow f\, xs \, [x]\\
    & \mid (x :: xs), [p] \Rightarrow f\, xs \, [x, p]\\
    & \mid (x :: xs), (p1 :: p2 :: ps) \Rightarrow \\
    & \quad \mathsf{if} \, \mathsf{clockOrNot} \, p_2 \, p_1\, x \, \mathsf{then}\\
    & \quad \quad f\, (x :: xs) \, (p2 :: ps)\\
    & \quad \mathsf{else}\\
    & \quad \quad f\, xs \, (x :: stk)
    \end{align*}
    \end{minipage}%
    \vline
    \begin{minipage}{0.33\textwidth}
    \begin{align*}
    &\mathsf{graham} = \lambda l.\\
    & (x, y) \leftarrow \mathsf{min\_xy}\, l\\
    & l' \leftarrow \mathsf{quicksort}\, (\charge 1; \sqsubseteq\, p) l\\
    & \mathsf{scan}\, l'\, \mathsf{nil}
    \\
    \\
    &\mathsf{convexHull} = \lambda n : \nat.\\
    &l \leftarrow \mathsf{unifList}\, n\\
    &l' \leftarrow \mathsf{sieve}\, l\\
    &\mathsf{graham}\, l'
    \end{align*}
    \end{minipage}
    \caption{Stochastic convex hull algorithm}
    \label{fig:qcknat}
\end{figure}

\begin{lemma}
    The graham-scan function has runtime $O(n\log n)$, where $n$ is the
    length of the input.
\end{lemma}
\begin{proof}
    By unfolding the denotational semantics, we see that its cost is given
    by adding the cost of the $\mathsf{min_{xy}}$, $\mathsf{quicksort}$
    and $\mathsf{scan}$ functions. By the quicksort case study, its cost
    is $O(n\log n)$. We are assuming that the cost of $\mathsf{min_{xy}}$
    is linear on the length of the list. Finally, the cost of $\mathsf{scan}$
    is also linear, though the proof is a bit more involved, so we will not
    go over it and, instead, will point to a standard analysis of Graham scan 
    \cite{graham1972efficient}. Therefore, the overall cost is bounded above by $O(n\log n)$.
\end{proof}

\begin{theorem}
    The stochastic convex hull algorithm has linear expected run time. 
\end{theorem}
\begin{proof}
    The cost analysis is given by the expected cost of the sieve function
    plus the average of the graham function. The cost of the sieve function is
    $O(n)$, since it iterates over the input list $5$ times. The graham function has its
    cost bounded by the sorting function, which costs $O(n' \log n')$,
    where $n'$ is the length of the output from the sieve function.
    By \Cref{lem:filterconvexhull}, the output list has an expected length
    of $O(\sqrt n)$. Thus, assuming that the $n$ points in $l$ have been sampled 
    uniformly and independently from $[0,1]^2$ and $\mu = (\pi_2 \circ \semec{\mathsf{sieve}})(l)$, 
    the cost structure becomes:
    \begin{align*}
    &(\pi_1 \circ \semec{\mathsf{sieve}}^\#)((n)) + \int (\pi_1\circ\semec{\mathsf{graham}})(l') \diff\mu(l') \leq \\
    &O(n) + \int (\pi_1\circ\semec{\mathsf{graham}})(l') \diff\mu(l') \leq \\
    &O(n) + \int \log(\mathsf{length}(l'))\mathsf{length}(l') \diff\mu(l') \leq \\
    &O(n) + \log(n)\int \mathsf{length}(l')\diff\mu(l') \leq \\
    &O(n) + O(\sqrt{n})\log n  \leq O(n) \qedhere
    \end{align*}
\end{proof}

\section{Proofs of miscellaneous lemmas and theorems}
\label{app:proofs}
\subsection{Proof of \Cref{th:expectmonad}}
\begin{proof}
    Since $\Psub$ is a monad, and the second component of the monad operations of $[0,\infty] \times \Psub -$
    are identical to the ones of $\Psub$, we only need to prove the monad laws for the first component.
    The unit laws follow from:
    \begin{align*}
        &\pi_1(\eta^\#(r ,\mu)) = r + 0 = r\\
        &\pi_1((f^\# \circ \eta)(x)) = \pi_1(f^\#(0, \delta_x)) = 0 + \pi_1(f(x))
    \end{align*}
    While the last law requires a bit more work:
    \begin{align*}
        &\pi_1((f^\# \circ g^\#)(r, \mu)) = \pi_1(f^\#(r + \int (\pi_1 \circ g)\diff\mu, (\pi_2 \circ g)^\#_{\Psub}(\mu))) = \\
        &r + \int (\pi_1 \circ g)\diff\mu + \int (\pi_1 \circ f) \diff((\pi_2 \circ g)^\#_{\Psub}(\mu)) = \pi_1((f^\# \circ g)^\#(r, \mu))
    \end{align*}
The last equation follows from the monad laws of $\Psub$.
\end{proof}

\subsection{Proof of \Cref{th:eqsound}}

\begin{proof}
    Since both semantics are the monadic semantics of CBPV, we know by Proposition~121 
    of \cite{levy2001call} that they satisfy the equations that do not reference the
    sampling and cost operations. We only have to prove that the equations that are
    specific to them are satisfied by both semantics.
    \begin{description}
        \item[$\oplus_0$ unit]: $\semcs{t \oplus_0 u} = \semcs{t} + 0 \semcs{u} = \semcs t$ and
        $\semec{t \oplus_0 u} = \semec{t} + 0 \semec{u} = \semec t$, where addition and scalar
        multiplication for the expected cost semantics are defined componentwise.
        \item[$\oplus_p$ symmetry]: $\semcs{t \oplus_p u} = (1-p)\semcs{t} + p\semcs{u} = \semcs {u \oplus_{1-p} t}$ and
        $\semec{t \oplus_p u} = (1-p)\semec{t} + p\semec{u} = \semec {u \oplus_{1-p} t}$.
        \item[$\oplus_{p}$ idempotent]: $\semcs{t \oplus_p t} = (1-p)\semcs{t} + p\semcs{t} = \semcs {t}$ and
        $\semec{t \oplus_p t} = (1-p)\semec{t} + p\semec{t} = \semec {t}$
        \item[$\oplus_p \oplus_q$ associativity]: 
        \begin{align*}
            &\semcs{t \oplus_p (u \oplus_q t')} = \\
            &(1-p) \semcs{t} + p((1-q)\semcs{u} + q \semcs{t'})=\\
            &(1-p) \semcs{t} + p(1-q)\semcs{u} + pq\sem{t'} = \\
            &(\semcs{t} \oplus_{\frac{p(1-q)}{1 - pq}} \semcs{u}) \oplus_{pq} \semcs{t'}
        \end{align*}
        The reasoning for the expected cost semantics is analog.
        \item[$\charge n$ monoid action]: $\semcs{\charge n; \charge m} = \delta_{(n+m, ())} = \semcs{\charge {(n+m)}}$ and
        for the expected cost semantics $\semec{\charge n; \charge m} = (n + m, \delta_{()}) = \semec{\charge {(n+m)}}$.
        \item[$\charge n$ commutativity]: $\semcs{\charge n; \charge m} = \delta_{(n+m, ())} = \semcs{\charge m; \charge n}$ and
        $\semec{\charge n; \charge m} = (n + m, \delta_{()}) = \semec{\charge m; \charge n}$.
        \item[$\charge 0$ unit]: $\semcs{\charge 0; t} = \semcs{t}^\#(\delta_{(0, ())} = \semcs{t}$ and
        $\semec{\charge 0; t} = (0 + (\pi_1 \circ \semec{t})( () ), (\pi_2\circ \semec{t})(()) ) = ((\pi_1\circ \semec{t})(()), (\pi_2\circ \semec{t})(())) = \semec{t}$ \qedhere
    \end{description}
\end{proof}

\subsection{Proof of \Cref{th:commutativity}}
\begin{proof}
The proof follows basically by commutativity of $\Psub$:
   \begin{align*}
    & \semcs{x \leftarrow t; y \leftarrow u; t'} = \\
    &\int_{\nat \times A} \int_{\nat \times B} \Psub(f)(\semcs{t'}(a, b)) \semcs u(\diff n_1 \diff a)\semcs t(\diff n_2 \diff b) = \\       
    &\int_{\nat \times B} \int_{\nat \times A} \Psub(f)(\semcs{t'}(a, b)) \semcs t(\diff n_2 \diff b)\semcs u(\diff n_1 \diff a) = \\       
    & \semcs{y \leftarrow u; x \leftarrow t; t'}, \text{ where $f(n, c) = (n + n_1 + n_2, c)$}
    \tag*{\qedhere}
   \end{align*}
\end{proof}

\subsection{Proof of \Cref{th:expectedcostcenter}}

We will formalize the maximality of the equation using the concept of the center of a monad,
which we now start to define.

\begin{definition}[\cite{carette2023central}]
    Let $X : \cat{C}$ be an object in a Cartesian category and $T : \cat{C} \to \cat{C}$ a monad. 
    A central cone at $X$ is a pair $(Z, \iota)$, where $\iota : Z \to T X$ is a morphism making 
    the following diagram commute for every $Y$:
    % https://q.uiver.app/#q=WzAsOCxbMCwwLCJaXFx0aW1lcyBUWSJdLFswLDEsIlRYIFxcdGltZXMgVFkiXSxbMCwyLCJUKFQgWCBcXHRpbWVzIFkpIl0sWzIsMiwiVF4yKFhcXHRpbWVzIFkpIl0sWzQsMiwiVChYXFx0aW1lcyBZKSJdLFs0LDAsIlQoWCBcXHRpbWVzIFRZKSJdLFsyLDAsIlRYIFxcdGltZXMgVFkiXSxbNCwxLCJUXjIoWFxcdGltZXMgWSkiXSxbNyw0LCJcXG11X3tYXFx0aW1lcyBZfSJdLFszLDQsIlxcbXVfe1ggXFx0aW1lcyBZfSIsMl0sWzIsMywiVChzdCdfe1gsWX0pIiwyXSxbMSwyLCJzdF97VFgsIFl9IiwyXSxbMCwxLCJcXGlvdGEiLDJdLFswLDYsIlxcaW90YSJdLFs2LDUsInN0J197WCwgVFl9Il0sWzUsNywiVChzdF97WCxZfSkiXV0=
\[\begin{tikzcd}
	{Z\times TY} && {TX \times TY} && {T(X \times TY)} \\
	{TX \times TY} &&&& {T^2(X\times Y)} \\
	{T(T X \times Y)} && {T^2(X\times Y)} && {T(X\times Y)}
	\arrow["\iota", from=1-1, to=1-3]
	\arrow["\iota"', from=1-1, to=2-1]
	\arrow["{st'_{X, TY}}", from=1-3, to=1-5]
	\arrow["{T(st_{X,Y})}", from=1-5, to=2-5]
	\arrow["{st_{TX, Y}}"', from=2-1, to=3-1]
	\arrow["{\mu_{X\times Y}}", from=2-5, to=3-5]
	\arrow["{T(st'_{X,Y})}"', from=3-1, to=3-3]
	\arrow["{\mu_{X \times Y}}"', from=3-3, to=3-5]
\end{tikzcd}\]
\end{definition}

This definition is the categorification of choosing elements of $TX$
that commute over every element of $TY$, for every $Y$. These cones
can be naturally organized as a category, by defining it as the appropriate
subcategory of the slice category $\cat{C}/TX$.

\begin{definition}[\cite{carette2023central}]
    A monad $T : \cat{C} \to \cat{C}$ is centralizable if for every object $X : \cat{C}$,
    there exists a terminal central cone on $X$.
\end{definition}

Something quite nice about this definition is that if for every $X$, there is a terminal
central cone $\mathcal Z(X)$, this assignment $X \mapsto \mathcal Z(X)$ extends to a commutative submonad of 
$T$ \cite{carette2023central}. With this definition, we can now restate \Cref{th:expectedcostcenter}
as follows.

\begin{theorem}
    The center of $\weight \times \Psub$ is the probability monad $P$.
\end{theorem}
\begin{proof}
    For every $X : \wqbs$, we define $\iota(\mu) = (0, \mu)$. To show
    that this is indeed a central cone, we observe that the upper leg
    of the central cone commutative diagram is the function $f(\mu, (r, \nu)) = (r, \mu \otimes \nu)$,
    where $\mu \otimes \nu$ is the product distribution on $\mu$ and $\nu$. A direct calculation
    shows that the lower leg is equal to $(r, \mu \otimes \nu)$ as well.

    In order to show that it is the terminal object, let $\iota : Z \to \weight \times \Psub X$ be a central cone.
    Let $z : Z$ and assume that $\iota(z) = (r', \mu)$. The upper leg then
    becomes $(r + r'\nu(Y), \mu \otimes \nu)$ while the lower leg is $(r' + r\mu(X), \mu \otimes \nu)$.
    Since this equation has to hold for every $r$ and $\nu$, it holds if, and only if, $r' = 0$ and
    $\mu(X) = 1$, i.e. its total mass is $1$. Therefore, $\iota : Z \to \weight \times \Psub X$ factors 
    through the inclusion $P X \hookrightarrow \weight \times \Psub X$, concluding the proof.
\end{proof}

\subsection{Proof of \Cref{th:premorphism}}
\begin{proof}
        The proof follows by showing that there are monad structure morphisms $\weight \times ([0,1] \to -) \twoheadrightarrow \weight \times \Psub$
    and $\varphi : \weight \times \Psub \hookrightarrow K_{\weight}$. The proof is concluded by the uniqueness of 
    factorizations.
    
    The first monad morphism is $(id \times \psi)$, where $\psi$ is the monad structure morphism $([0,1] \to -_\bot) \twoheadrightarrow \Psub$, and the proof that this transformation is component-wise image-dense follows from
    the observation that the product order in $\wqbs$ is given pairwise and $\psi$ is, by assumption, dense in its image. The monad
    morphisms axioms follow from:

    \begin{align*}
        &(id \times \psi)(\eta(a)) = (id \times \psi)(0, \lamb r a) = (0, \delta_a) = \eta(a)\\
        \\
        &(((id\times \psi) \circ f)^\# \circ (id \times \psi))(r, g) = ((id\times \psi) \circ f)^\#(r, \psi(g), (\psi \circ \pi_2 \circ f)^\#(g)) =\\
        &(r + \int (\pi_1 \circ f)\diff (\psi(g)), (\psi \circ (\pi_2 \circ f)^\#))(g) = \\
        &(r + \int (\pi_1 \circ f \circ g)\diff (\lambda), (\psi \circ (\pi_2 \circ f)^\#)(g))= ((id \times \psi) \circ f^\#)(r, g)
    \end{align*}
    
    The second monad morphism has components $\varphi(r, \mu) = \lamb f {r + \int f \diff \mu}$. The
    proof that this is indeed order reflecting is a direct consequence of $\Psub \hookrightarrow K_{\weight}$
    being order reflecting.
    
    The proof that this is indeed a monad morphism follows, once again, from a series of direct calculations.\qedhere 
\end{proof}

\subsection{Proof of \Cref{lem:permInv}}

\begin{proof}
    This can be proved by strong induction on the length of the input list.
    If the list is empty, then its length is $0$ and the diagram commutes.
    Next, assume that the length is greater than $0$. We can prove the 
    following program equality:
    \[
    \begin{aligned}
    & len \leftarrow \app \mathsf{length}\, l\\
    & r \leftarrow \app \rand len\\
    & pivot \leftarrow l[r]\\
    \Aboxed{&(l_1, l_2) \leftarrow \app {\app{biFilter} {(\lamb n {\charge 1; n \leq pivot})}} (\mathsf{drop}\, l)}\\
    & l_1' \leftarrow \mathsf{quicksort}\, l_1\\
    & l_2' \leftarrow \mathsf{quicksort}\, l_2\\
    & \mathsf{length}\, (l_1' \mdoubleplus pivot :: l_2')
    \end{aligned}
    =
    \begin{aligned}
    \quad & len \leftarrow \app \mathsf{length}\, l\\
    & r \leftarrow \app \rand len\\
    & pivot \leftarrow l[r]\\
    \Aboxed{& {\begin{aligned}[t]&\charge (len - 1)\\
    & (l_1, l_2) \leftarrow \app {\app{biFilter} {(\lamb n {n \leq pivot})}} (\mathsf{drop}\, l)\end{aligned}}}\\
    & l_1' \leftarrow \mathsf{quicksort}\, l_1\\
    & l_2' \leftarrow \mathsf{quicksort}\, l_2\\
    & \mathsf{length}\, (l_1' \mdoubleplus pivot :: l_2')
    \end{aligned}
    \]

    Next, we want to consider the interaction of the regular quicksort algorithm and applying
    the length function to it. Furthermore, by using the additive properties of the
    length function and the strong induction hypotheses, we get the following program equation:
\[  \begin{aligned}
    \quad & len \leftarrow \app \mathsf{length}\, l\\
    & r \leftarrow \app \rand len\\
    & pivot \leftarrow l[r]\\
    &\charge (len - 1)\\
    & (l_1, l_2) \leftarrow \app {\app{biFilter} {(\lamb n {n \leq pivot})}} (\mathsf{drop}\, l)\\
    \Aboxed{& {\begin{aligned}[t]&l_1' \leftarrow \mathsf{quicksort}\, l_1\\
    & l_2' \leftarrow \mathsf{quicksort}\, l_2\\
    & \mathsf{length}\, (l_1' \mdoubleplus pivot :: l_2')\end{aligned}}}
    \end{aligned}
    =
    \begin{aligned}
    \quad & len \leftarrow \app \mathsf{length}\, l\\
    & r \leftarrow \app \rand len\\
    & pivot \leftarrow l[r]\\
    & \charge (len - 1)\\
    & (l_1, l_2) \leftarrow \app {\app{biFilter} {(\lamb n {n \leq pivot})}} (\mathsf{drop}\, l)\\
    \Aboxed{& {\begin{aligned}[t]&n_1 \leftarrow (\mathsf{length}\, l_1)\\
    & n_2 \leftarrow (\mathsf{length}\, l_2)\\
    & n_1' \leftarrow \mathsf{qck}_\nat \, n_1\\
    & n_2' \leftarrow \mathsf{qck}_\nat \, n_2\\
    & \produce (n_1' + 1 + n_2')\end{aligned}}}   
    \end{aligned}
    \]

    Since the pivot is chosen uniformly at random, we can deduct denotationally that the following equations hold:
  \[
    \begin{aligned}
    &len \leftarrow \app \mathsf{length}\, l\\
    & r \leftarrow \app \rand len\\
    & pivot \leftarrow l[r]\\
    & \charge (len - 1)\\
    & (l_1, l_2) \leftarrow \app {\app{biFilter} {(\lamb n {n \leq pivot})}} (\mathsf{drop}\, l)\\
    &n_1 \leftarrow (\mathsf{length}\, l_1)\\
    & n_2 \leftarrow (\mathsf{length}\, l_2)\\
    & n_1' \leftarrow \mathsf{qck}_\nat \, n_1\\
    & n_2' \leftarrow \mathsf{qck}_\nat \, n_2\\
    & \produce (n_1' + 1 + n_2')
    \end{aligned}
    =
    \begin{aligned}
    \quad
    & len \leftarrow \mathsf{length}\, l\\
    & r \leftarrow \rand \, len \\
    & \charge (len - 1)\\
    & n_1'' \leftarrow \mathsf{qck}_\nat \, r\\
    & n_2'' \leftarrow \mathsf{qck}_\nat \, (len - r - 1)\\
    & \produce (n_1'' + 1 + n_2'')
    \end{aligned}
    =
    \begin{aligned}
    \quad
    & len \leftarrow \mathsf{length}\, l\\
    & \mathsf{qck}_\nat\, len
    \end{aligned}
    \]
    The last equation holds under the inductive hypothesis that the list $l$ is non-empty. This concludes the inductive case and the proof.
    \end{proof}

\subsection{Proof of \Cref{th:excostmin}}

    \begin{lemma}
        The canonical monad morphism $e : F \to \weight \times \Psub$ is densely strong epic. 
    \end{lemma}
    \begin{proof}
        Since the product order is given componentwise, it suffices to show that the components $\pi_1 \circ e : F \to \weight$
        and $\pi_2 \circ e : F \to \Psub$ are densely strong epic. We start by reasoning about $\pi_1 \circ e$. Assume without loss of
        generality that $X$ is non-empty and that $x_0 \in X$. Let $f : R^+ \to \weight$ 
        be a (restricted) element of $M_{\weight}$ which is everywhere finite; this subset is Scott-dense in $M_{\weight}$. We define the function
        $g(r) = \charge {\floor{f(r)}}; (\charge 0) \oplus_{f(r) - \floor{f(r)}} (\charge 1); \produce x_0 : \R^+ \to F X$,
        where $\floor{\cdot}$ is the floor function. By construction,
        and using the fact that $e$ is a monad morphism that preserves the sampling and cost operations, 
        $\pi_1 \circ e \circ g = f$ and we can conclude that $\pi_1 \circ e$ is densely strong epic.

        For the function $\pi_2 \circ e$, it suffices to use a variation of the argument made in Lemma~4.4 by V\'ak\'ar et al.~\cite{wqbs},
        where they prove a similar property but for the monad of measures which are not-necessarily bounded and an operation for ``sampling''
        from the Lebesgue measure over the real line.
    \end{proof}
    We can now conclude our proof of \Cref{th:excostmin}.
    \begin{proof}
    Let $T \hookrightarrow K_{\weight}$ be a submonad of $K_{\weight}$ containing the operations $\mathsf{charge}$
    and $\uniform$. By assumptions and the lemma above, we have the following commutative diagram for every $X$.
    % https://q.uiver.app/#q=WzAsNCxbMCwwLCJGIl0sWzEsMSwiS197XFx3ZWlnaHR9Il0sWzEsMCwiVCJdLFswLDEsIlxcd2VpZ2h0IFxcdGltZXMgXFxQc3ViIl0sWzAsMywibSIsMix7InN0eWxlIjp7ImhlYWQiOnsibmFtZSI6ImVwaSJ9fX1dLFszLDEsIiIsMCx7InN0eWxlIjp7InRhaWwiOnsibmFtZSI6Imhvb2siLCJzaWRlIjoidG9wIn19fV0sWzAsMl0sWzIsMSwiIiwyLHsic3R5bGUiOnsidGFpbCI6eyJuYW1lIjoiaG9vayIsInNpZGUiOiJ0b3AifX19XSxbMywyLCIiLDEseyJzdHlsZSI6eyJib2R5Ijp7Im5hbWUiOiJkYXNoZWQifX19XV0=
    \[\begin{tikzcd}
	   F X & T X \\
	   {\weight \times \Psub X} & {K_{\weight} X} 
	   \arrow[from=1-1, to=1-2]
	   \arrow["e_X"', two heads, from=1-1, to=2-1]
	   \arrow[hook, from=1-2, to=2-2]
	   \arrow[dashed, from=2-1, to=1-2]
	   \arrow[hook, from=2-1, to=2-2]
    \end{tikzcd}\]
The dashed arrow is defined by the universal property of factorization systems. Since the morphism $\weight \times \Psub X \to K_{\weight} X$ is
monic, so is $\weight\times \Psub X \to T X$.
    \end{proof}
\end{document}